\documentclass[12pt]{article}
\usepackage{amsmath,amsfonts,amssymb, amsthm}
\usepackage{hyperref} 
\hypersetup{colorlinks=true, citecolor=blue}
\usepackage{multirow} 
\usepackage{caption} 
\usepackage{kotex}
\usepackage{xcolor} 
\usepackage[sort]{natbib}
\usepackage{fullpage} 
\usepackage{float} 
\usepackage{chngpage} 
\usepackage{lscape}
\usepackage{mathtools}
\usepackage[blocks]{authblk}
\usepackage{booktabs} 
\usepackage{pbox} 
\usepackage{algpseudocode, algorithm}
\usepackage{enumerate}
\usepackage{bigints}
\usepackage{tikz}

\usepackage{graphicx}
\usepackage{url}
\usepackage{stmaryrd}

\usepackage{textcomp} 

\newtheorem{thm}{Theorem}

\newtheorem{prop}{Proposition}
\newtheorem{fact}{Fact}
\newtheorem{lem}{Lemma}

\newtheorem{assumption}{Assumption}

\newcommand{\bA}{\boldsymbol{A}}
\newcommand{\bB}{\boldsymbol{B}}

\newcommand{\bI}{\boldsymbol{\rm I}}

\newcommand{\bM}{\boldsymbol{M}}

\newcommand{\bP}{\boldsymbol{P}}
\newcommand{\bQ}{\boldsymbol{Q}}

\newcommand{\bS}{\boldsymbol{S}}

\newcommand{\bU}{\boldsymbol{U}}
\newcommand{\bV}{\boldsymbol{V}}

\newcommand{\bX}{\boldsymbol{X}}

\newcommand{\br}{\boldsymbol{r}}

\newcommand{\bt}{\boldsymbol{t}}
\newcommand{\bu}{\boldsymbol{u}}

\newcommand{\bw}{\boldsymbol{w}}
\newcommand{\bx}{\boldsymbol{x}}
\newcommand{\by}{\boldsymbol{y}}
\newcommand{\bz}{\boldsymbol{z}}

\newcommand{\bzero}{\boldsymbol{0}}

\newcommand{\bbeta}{\boldsymbol{\beta}}

\newcommand{\btheta}{\boldsymbol{\theta}}
\newcommand{\brho}{\boldsymbol{\rho}}

\newcommand{\bSigma}{\boldsymbol{\Sigma}}

\newcommand{\diag} {{{\rm {diag}}}}
\newcommand{\tr} {{{\rm {tr}}}}

\DeclareMathOperator*{\argmin}{arg\,min}
\DeclareMathOperator*{\Argmin}{Arg\,min}

\newcommand{\pkg}[1]{\texttt{#1}} 
\newcommand{\indep}{\perp\!\!\!\!\perp}

\title{Linear Shrinkage Convexification of Penalized Linear Regression With Missing Data}

\author[1,2]{Seongoh Park}

\author[3]{Seongjin Lee}

\author[4]{Nguyen Thi Hai Yen}
\author[4]{Nguyen Phuoc Long}

\author[5]{Johan Lim\footnote{To whom all correspondence should be addressed. Email: \texttt{johanlim@snu.ac.kr}}}

\affil[1]{School of Mathematics, Statistics and Data Science, Sungshin Women's University, Seoul, Korea} 
\affil[2]{Data Science Center, Sungshin Women's University, Seoul, Korea}
\affil[3]{Statistics and Operations Research, University of North Carolina at Chapel Hill, North Carolina, U.S.} 
\affil[4]{Department of Pharmacology and PharmacoGenomics Research Center, Inje University College of Medicine, Busan, Korea} 
\affil[5]{Department of Statistics, Seoul National University, Seoul, Korea}

\date{}

\begin{document}

\maketitle

\begin{abstract}
\noindent \medskip 

One of the common challenges faced by researchers in recent data analysis is missing values. In the context of penalized linear regression, 
which has been extensively explored over several decades, missing values introduce bias and yield a non-positive definite covariance matrix of the 
covariates, rendering the least square loss function non-convex. In this paper, we propose a novel procedure called the linear shrinkage positive 
definite (LPD) modification to address this issue. The LPD modification aims to modify the covariance matrix of the covariates in order to ensure consistency 
and positive definiteness. Employing the new covariance estimator, we are able to transform the penalized regression problem into a convex one, 
thereby facilitating the identification of sparse solutions. Notably, the LPD modification is computationally efficient and can be expressed analytically. 
In the presence of missing values, we establish the selection consistency and prove the convergence rate of the $\ell_1$-penalized regression estimator with LPD, showing an $\ell_2$-error convergence rate of square-root of $\log p$ over $n$ {\color{black}by a factor of $(s_0)^{3/2}$ ($s_0$: the number of non-zero coefficients).} 
To further evaluate the effectiveness of our approach, we analyze real data from the 
Genomics of Drug Sensitivity in Cancer (GDSC) dataset. This dataset provides incomplete measurements of drug sensitivities of cell lines and their protein expressions. 
We conduct a series of penalized linear regression models with each sensitivity value serving as a response variable and protein expressions as explanatory variables.

\noindent{Keyword:} General missing dependency, lasso, positive definiteness.

\end{abstract}

\section{Introduction}

Regularized or penalized linear regression has been largely explored for decades, motivated from a variety of modern applied fields \citep{Lee:2003,Ghosh:2005,Daye:2012,Han:2020} where the sample size is much smaller than the number of variables to be analyzed. Among different regularizations in linear regression such as ridge \citep{Hoerl:1970}, lasso \citep{Tibshirani:1996, Zou:2006}, Dantzig selector \citep{Candes:2007}, elastic net \citep{Zou:2005}, SCAD \citep{Fan:2001}, the lasso regression has gained its popularity because its statistical properties \citep{Zhao:2006, Lee:2015,Zou:2006,Geer:2009,Fu:2000} and computational aspects \citep{Efron:2004,Osborne:2000,Friedman:2007} are well established.

Though the technology for data collection has exceptionally advanced in recent years, one common issue that researchers face in data analyses is missing values.
Our motivating example is drug response data (\url{https://www.cancerrxgene.org/}, Release v8.4, July 2022) and the pan-cancer proteomic profile of 8,498 proteins from 949 human cancer cell lines (28 tissue types, more than 40 cancer types) \citep{Goncalves:2022}. This study was to measure the sensitivities (IC50/AUC) of cells to different drugs and aimed to find the association between drug responses and protein levels. Missing data are widely seen in mass spectrometry (MS)-based proteomics \citep{Bobbie-Jo:2015} or metabolomics \citep{Wei:2018}. Causes for missing values could be biological or technical (e.g., stochastic fluctuations during data acquisition) and of random or not at-random \citep{Karpievitch:2012}. Unless treated appropriately, incomplete data often lead to biased results and hamper study reproducibility \citep{Dabke:2021}. For instance, for the lasso regression \cite{Sorensen:2015} showed that a naive approach using the incomplete data without correction does not satisfy estimation consistency (see Proposition 1 therein). 


%

Many researchers have come up with different solutions to address this issue under linear regression models.
First, the expectation-maximization (EM) algorithm is developed by \cite{Stadler:2010} where they aimed to find the sparse inverse covariance matrix and used it in the sparse linear regression. However, the EM algorithm is model-specific and known to converge slowly. 
Alternatively, variable selection can be combined with multiple imputation that is commonly used in practice. For example, one can perform majority votes based on selection results from multiply imputed datasets \citep{Heymans:2007,Wood:2008,Lachenbruch:2011,Long:2015}. To avoid the ad-hoc rules for combining different sets of selected variables, \cite{Wan:2015} and \cite{Li:2023} considered stacking imputed datasets and selected the same variables across all datasets, which is termed as a stacked method in \cite{Du:2022}. In \cite{Chen:2013}, they proposed the group-wise selection approach to consistently choose variables across imputed datasets, which is named a grouped method in \cite{Du:2022}. These methods exhibited satisfactory performance in simulated and real data analyses; however, theoretical evidences are elusive.

%
%
%
%
%
%

To fill this gap, researchers have paid attention on de-biasing approaches. These are based on the observation that a loss function, for example, mean squared error, is biased if data are not completely observed. Thus, related work adjusted it by adding or multiplying de-biasing constants to the covariance part or Gram matrix (e.g. see (\ref{eq:IPW_correction})) and solved the corrected optimization problem with different penalization methods; for example, \cite{Liang:2009} used the SCAD penalty, and \cite{Loh:2012} adopted the lasso penalty. Following \cite{Loh:2012} where estimation consistency is proved, \cite{Sorensen:2015} additionally showed sign consistency under the irrepresentable condition adapted to their contexts.
This line of work, however, has a computational issue that the modified loss function is no longer convex. It was sidestepped in \cite{Rosenbaum:2010} and \cite{Wang:2019} by using Dantzig selector that is always defined as a linear programming regardless of the modification.

A more fundamental remedy for the non-convexity is to modify the corrected covariance factor $\widehat{\bSigma}$ to be positive definite (PD). To this end, \cite{Datta:2017} found the closest PD matrix to $\widehat{\bSigma}$ using the element-wise maximum norm:
\begin{equation}\label{eq:PD_elem_max}
\tilde{\bSigma}^{CoCo} = \argmin_{\bSigma\succ 0} \|\widehat{\bSigma} - \bSigma \|_{\max}.	
\end{equation}
Using it, they solved the $\ell_1$-penalized regression problem, which is named CoColasso, and proved estimation and selection consistency under regular conditions including the irrepresentable condition. This area of research has been recently studied further.
Though handling the measurement error not missing data, \cite{Zheng:2018} and \cite{Zhang:2022} proposed to use different penalty functions, a combination of $\ell_1$- and concave penalty, and $\ell_0$-penalty, respectively, to ensure better theoretical properties of estimators (i.e. faster oracle inequality).
\cite{Escribe:2021} considered partially corrupted data where some of explanatory variables are corrupted under some measurement error model and the others are not. Thus, they only solved (\ref{eq:PD_elem_max}) for a smaller dimension at which the measurement errors are found. 
On the other hand, in solving (\ref{eq:PD_elem_max}), \cite{Takada:2019} suggested to downweight components at which samples are highly missing. To do so, they used a weighted version of Frobenius norm.


However, solving (\ref{eq:PD_elem_max}) is computationally demanding in general because it does not have a closed form solution. More specifically, the eigen-decomposition of a $p$-dimensional symmetric matrix and projection of a $p^2$-dimensional vector to $\ell_1$-ball are repeated until convergence \citep{Datta:2017,Han:2014}. \cite{Takada:2019} used the (weighted) Frobenius norm to find the closest PD matrix in which the eigen-decomposition is also repeated.
Because of this, the existing methods mentioned above may not be practically useful. The heavy workload can greatly impede further inference procedures using regularized estimators such as bolasso (bootstrapped enhanced lasso, \cite{Bach:2008}) and a modified residual bootstrapped lasso, which are based on resampling procedures (\cite{Chatterjee:2011,Chatterjee:2013} or stability selection \citep{Meinshausen:2010}). Moreover, there is a need for solving the penalized regression recursively; e.g. online learning procedure \citep{Duchi:2009,Langford:2008,Xiao:2009}.

In this paper, we propose the linear shrinkage positive definite (LPD) modification of the covariance matrix for the high-dimensional regression problem with incomplete data.
The key idea is to reduce the class of PD matrices over which the minimization (\ref{eq:PD_elem_max}) is taken. We consider the linear shrinkage class defined in (\ref{eqn:lin-class}). In other words, we shrink the non-PD $\widehat{\bSigma}^{\rm IPW}$ (corrected estimator defined in (\ref{eq:IPW_correction})) to $\mu \bI$ as $\alpha \widehat{\bSigma}^{\rm IPW} + (1-\alpha) \mu \bI$ for some $\alpha$ and $\mu$. 
The proposed way is easy and straightforward due to its simple form, and above all, computationally fast since the optimal $\alpha$ and $\mu$ have explicit forms (see (\ref{eqn:alpha-sol}) and Proposition \ref{prop:LPD_solution}). 
Based on the new covariance estimators, we convexify the penalized regression problem and thus can easily find the sparse solution $\widehat{\bbeta}^{\rm LPD}$ to (\ref{eqn:lasso-problem}). 
Furthermore, under the irrepresentable condition, we establish the selection consistency and prove the rate of convergence by $O_p \left(\sqrt{\log p/n} \right)$ in $\ell_2$-error, which is comparable to what was previously achieved by CoColasso \citep{Datta:2017}. One of the key tools to prove the results is the non-asymptotic inequality of the IPW estimator (Theorem \ref{thm:IPW_spectral} in Supplementary Materials \ref{sec:IPW_spectral}), which can be of independent interest.
Our numerical study also reveals the proposed one performs comparatively in the finite sample scenarios.
We also analyze real data from Genomics of Drug Sensitivity in Cancer (GDSC) where sensitivity to different drugs and protein expressions was measured but incompletely. We separately run a list of penalized linear regression models with each of sensitivity values as a response variable and protein expressions as explanatory variables, which would have not been feasible if our estimation procedure were not scalable like CoColasso.


The remainder of the paper is organized as follows. In Section 2, we define different classes of linear shrinkage estimators from different matrix norms. Then, we describe how to use the modified Gram matrix in the lasso regression and verify theoretical properties of the resulting lasso estimator under some conditions. In Section 3, we examine the finite sample performance of the proposed method compared to existing methods through simulated data. In Section 4, the proposed regularized regression is applied to incomplete data from Genomics of Drug Sensitivity in Cancer (GDSC) to identify the most predictive proteins for two example drugs. In Section 5, we conclude this paper with a discussion of limitations and potential extensions.

\section{Convexification of Lasso using LPD} 

\subsection{Problem formulation}

We assume a linear relationship between explanatory variables $\bx_i= (x_{i1}, \ldots, x_{ip})^\top$ and a response variable $y_i$, which is represented by regression coefficients $\bbeta=(\beta_1, \ldots, \beta_p)^\top$:
\begin{equation}\label{eq:reg_model}
y_i = \bx_i^\top \bbeta + \epsilon_i, \quad i=1,\ldots, n,
\end{equation}
where $\epsilon_i$ is an error term independent of $\bx_i$, and samples are independent across $i=1,\ldots, n$.
For ease of exposition, we assume all the variables are centered; $\mathbb{E}x_{ij} = \mathbb{E}\epsilon_i=0$ and thus $\mathbb{E}y_i=0$.
Due to the missing structure, we can only observe $\tilde{y}_i, \tilde{\bx}_i = (\tilde{x}_{i1}, \ldots, \tilde{x}_{ip})^\top$ where
\begin{equation}\label{eq:data_str}
\tilde{y}_i = \begin{cases}
y_i, & \text{if } y_i \text{ is observed},\\
0, & \text{otherwise},
\end{cases}
\quad
\tilde{x}_{ij} = \begin{cases}
x_{ij}, & \text{if } x_{ij} \text{ is observed},\\
0, & \text{otherwise}.
\end{cases}
\end{equation}
Adopting matrix notations, we write $\tilde{\by} = (\tilde{y}_1, \ldots, \tilde{y}_n)^\top$ and $\tilde{\bX} = [\tilde{\bx}_1, \ldots, \tilde{\bx}_n]^\top$.
The penalized regression problem of our interest would be defined by minimizing the residual sum of squares 
$$
\min_{\bbeta} \frac{1}{2n} \| \tilde{\by} - \tilde{\bX}\bbeta \|_2^2 + J_{\lambda}  ( {\bbeta})
$$
for some penalty function $J_\lambda$ indexed by a tuning parameter $\lambda >0$. The problem can be depicted with covariance terms, $\bS= {\tilde{\bX}}^{\top} \tilde{\bX}/n$ and $\br = {\tilde{\bX}}^{\top}\tilde{\by} / n$, i.e.
\begin{equation}\label{eq:reg_lse_complete}
\min_{\bbeta }\frac{1}{2}{\bbeta}^{\top} \bS {\bbeta} - \br^{\top} {\bbeta} + J_{\lambda}( {\bbeta}) \equiv g(\bbeta; \bS, \br, J_\lambda).
\end{equation}
However, bias caused by missing values in $\bS$ and $\br$ renders the optimal solution of the above inconsistent. A straightforward remedy is to adjust the bias through an inverse probability weighting (IPW) technique and to use the corrected estimators: i.e. $\bS \leftarrow \widehat{\bSigma}^{\rm IPW}, \br \leftarrow \hat{\brho}^{\rm IPW}$. The IPW estimators are defined by correcting every component with an observation probability:
\begin{equation}\label{eq:IPW_correction}
\widehat{\bSigma}^{\rm IPW} = \bS * \left[\dfrac{1}{\pi^{xx}_{jk}}, 1 \le j,k\le p \right], \quad
\hat{\brho}^{\rm IPW} = \br * \left[\frac{1}{\pi^{xy}_j},  1\le j \le p \right],	
\end{equation}
where $*$ is the element-wise product between two matrices (or vectors) of the same size. $\pi^{xx}_{jk}$ is a probability that the $(j,k)$-th explanatory variables are observed, and $\pi^{xy}_{j}$ that the $j$-th explanatory variable and response variable are observed. They are precisely defined in Assumption \ref{assm:multi_bern}.
The idea of replacing the sample covariances by the IPW estimators has been used in covariance/precision matrix estimation \citep{Park:2022_stat,Lounici:2014,Park:2019,Park:2021,Pavez:2021,Cai:2016:JMA}.
However, $\widehat{\bSigma}^{\rm IPW}$ is not PD in general, and thus $g(\bbeta; \widehat{\bSigma}^{\rm IPW}, \hat{\brho}^{\rm IPW}, J_\lambda)$ in (\ref{eq:reg_lse_complete}) is not convex, even if $J_\lambda$ is convex (e.g. lasso penalty). Thus, we use a PD alternative based on the linear shrinkage method \citep{Ledoit:2004, Choi:2019}, which finds a PD matrix closest to the non-PD in the linear shrinkage class. It solves 
\begin{equation} \label{eqn:dist-min}
\Phi_{\mu, \alpha}(\widehat{\bSigma}^{\rm IPW}) \in 
\Argmin_{\Phi_{\mu, \alpha} \in \mathcal{C}_{\epsilon}(\widehat{\bSigma}^{\rm IPW})} \left\| \widehat{\bSigma}^{\rm IPW} - \Phi_{\mu, \alpha} \right\|,
\end{equation} 
for some matrix norm $\|\cdot \|$, where $\mathcal{C}_{\epsilon}$ is a class of the linear shrinkage matrices defined in (\ref{eqn:PD_general}). Hereafter, we name the PD modification using the linear shrinkage method as LPD and denote the solution $\Phi_{\mu, \alpha}(\widehat{\bSigma}^{\rm IPW})$ by $\widehat{\bSigma}^{\rm LPD}$ for notational simplicity. 
In the following sections, we give a detailed account of explicit forms of LPDs in different matrix norms (Section 2.2). In the next section (Section 2.3), we study theoretical properties of the solution of the lasso regression:
\begin{equation} \label{eqn:lasso-problem} 
\min_{\bbeta} ~ \frac{1}{2}{\bbeta}^{\top} \widehat{\bSigma}^{\rm LPD} {\bbeta} - \bbeta^{\top} \hat{\brho}^{\rm IPW}  + \lambda  \| {\bbeta}\|_1,
\end{equation} 
where $\widehat{\bSigma}^{\rm LPD}$ is applied as the Gram matrix.

We end this section by introducing the results of \cite{Lee:2015} where the authors study a generalized framework for the regularized M-estimators that includes our problem (\ref{eqn:lasso-problem}). To prove the rate of convergence in terms of $\ell_2$-error and consistent recovery of the support, they assumed three conditions (i) restricted strong convexity (RSC), (ii) irrepresentability condition (IR), and (iii) bounded gradient condition (BG). We refer to {Supplementary Materials \ref{sec:proof_Lee_result}} or the original reference for more details about the formulation. 
In our context, the IR and BG conditions are simplified to the condition (C1) and (C2) of Proposition \ref{prop:LPD_solution}, while {\color{black}the RSC condition is reduced to (C3) of it due to the linear shrinkage structure.}

To describe the results, we introduce notations. Consider the model space $M_{\mathcal{A}} = \{\bbeta \in \mathbb{R}^p: \beta_j=0, j \in \mathcal{A}^c \}$ where $\mathcal{A} \subset [p]$ is the support of true parameter $\bbeta^*$. We divide a square matrix using the support $\mathcal{A}$ and denote by $\bA_{\mathcal{A} \mathcal{A}}, \bA_{\mathcal{A} \mathcal{A}^c}, \bA_{\mathcal{A}^c \mathcal{A}}, \bA_{\mathcal{A}^c \mathcal{A}^c}$, each of which restricts rows and columns of $\bA$ on corresponding index sets. We denote by $\lambda_{\min}(\bA)$ or $\lambda_{\max}(\bA)$ the smallest or largest eigenvalue of $\bA$, respectively.
Then, we can easily derive the following based on the results in \cite{Lee:2015}. 
Remark that the norm in (C1) is the matrix $\ell_\infty$-norm (i.e. maximum of column-wise sum) and the one in (C2) is the element-wise maximum norm of a vector.
\begin{prop} \label{prop:consistency}
	{\color{black}Assume $\lambda_{\min}(\widehat{\bSigma}^{\rm IPW}) < 0$.
		For $\epsilon>0$ such that $\epsilon < \lambda_{\min}(\bSigma)$, define by $\widehat{\bSigma}^{\rm LPD}$ the LPD of $\widehat{\bSigma}^{\rm IPW}$ over the class $\mathcal{C}_{\epsilon}(\widehat{\bSigma}^{\rm IPW})$.} Suppose there exists constants $\tilde{\tau} \in (0,1)$ and $\lambda>0$ such that:
	\begin{align*}
	&\text{(C1)} \quad \left\|{\widehat{\bSigma}^{\rm LPD}}_{\mathcal{A}^c \mathcal{A}}({\widehat{\bSigma}^{\rm LPD}}_{\mathcal{A} \mathcal{A}})^{-1} \right\|_\infty \le 1-\tilde{\tau}, \\
	&\text{(C2)} \quad \frac{4(2-\tilde{\tau})}{\tilde{\tau}} \left\|\widehat{\bSigma}^{\rm LPD} \bbeta^* - \hat{\brho}^{\rm IPW} \right\|_\infty < \lambda,\\
	&\text{(C3)} \quad 
	\min_{\bt: \bt \neq 0, \bt_{\mathcal{A}^c}=0} \bt^{\top} \widehat{\bSigma}^{\rm LPD} \bt / \bt^\top \bt \ge \min\{0.5\lambda_{\min}(\bSigma_{\mathcal{A}\mathcal{A}}), \mu \},
	\end{align*}
	Then, the followings hold:
	\begin{align*}
	&\text{(R1)} \quad \text{The minimizer }\widehat{\bbeta}^{\rm LPD} \text{ of (\ref{eqn:lasso-problem}) is unique,} \\
	&\text{(R2)} \quad \|\widehat{\bbeta}^{\rm LPD} - {\bbeta}^*\|_2 \le
	\frac{4}{\color{black} \min\{\lambda_{\min}(\bSigma_{\mathcal{A}\mathcal{A}}), \mu \}} \left( 1 + \frac{\tilde{\tau}}{4} \right)\sqrt{|\mathcal{A}|} \lambda,\\
	&\text{(R3)}  \quad \hat{\beta}^{\rm LPD}_{j} = 0, \quad j \in \mathcal{A}^c.
	\end{align*}
	
\end{prop}
\noindent 
The proof of Proposition \ref{prop:consistency} is postponed to Supplementary Materials \ref{sec:proof_Lee_result}, which is offered solely for completeness. We do not assert any contribution to it.

\subsection{Explicit forms of LPD} 

In the estimation of high dimensional covariance matrix \citep{Bickel:2008_banding,Bickel:2008_thresholding, Rothman:2012}, structural assumptions on true covariance matrix are often made, and many regularized estimators are proposed accordingly. However, the regularization typically does not impose PDness, which makes the resulting estimate not PD in general. Several efforts are made to find an estimator that attains both sparsity and PDness \citep{Bien:2011,Lam:2009,Liu:2014,Rothman:2012,Xue:2012,Choi:2019}. Among them, the fixed 
support positive definite modification (FSPD) by \cite{Choi:2019} is initially designed to make a covariance matrix estimator PD while preserving its support as its name indicates. However, FSPD is still tempting even for cases where we do not have structural assumptions on covariance matrices but need PDness.
Since it is computationally easy and is applicable to any non-PD matrix, we adopt this idea for estimating the PD gram matrix under the missing data structure.

Let $\bA$ be a real symmetric matrix to be modified PD. For a given $\epsilon>0$, we define the class of LPD by
\begin{equation} \label{eqn:lin-class} 
\mathcal{C}_{\epsilon}(\bA) = \left\{ \alpha \bA + (1-\alpha)\mu \bI : \alpha \in (0,1), \mu \in \mathbb{R}, \alpha \lambda_{\min}(\bA) + (1-\alpha)\mu \ge \epsilon \right\}.
\end{equation}
Following \cite{Choi:2019} and \cite{Cho:2021}, we minimize a distance induced by any matrix norm $|| \cdot||$:
\begin{equation}\label{eqn:PD_general}
\min_{\Phi_{\mu, \alpha} \in \mathcal{C}_{\epsilon}(\bA)} \left\| \bA - \Phi_{\mu, \alpha} \right\|.
\end{equation} 
Note that the minimization is taken over $(\mu, \alpha)$, and the distance in (\ref{eqn:dist-min}) is indeed rewritten as 
\begin{align*}
\left\| \alpha \bA + (1-\alpha) \mu \bI - \bA \right\| = (1-\alpha) \left\| \mu \bI - \bA \right\|.
\end{align*}
In the meantime, if $\lambda_{\min}(\bA) < \epsilon \le \mu$, the constraint can be expressed as
\begin{equation} \nonumber 
\alpha \lambda_{\min}(\bA)  + (1-\alpha) \mu \ge \epsilon 
\iff  \alpha \le \frac{\mu-\epsilon}{\mu - \lambda_{\min}(\bA)}.
\end{equation}
We thus know that the optimal solution $\alpha^*$ for fixed $\mu \ge \epsilon$ is
\begin{equation}\label{eqn:alpha-sol}
\alpha^*  = \alpha^*(\mu) =\frac{\mu - \epsilon}{\mu - \lambda_{\min}(\bA)}.
\end{equation}
regardless of the type of the norm. On the other hand, the solution to $\mu$ depends on the distance we use. The following proposition summarizes the results. We define matrix norms as $||\bA||_2 = \sqrt{\lambda_{\max}(\bA^\top \bA)}$, $||\bA||_F = \sqrt{\tr(\bA^\top \bA)/d_2}$, $||\bA||_{\infty}=\max_{i \in [d_1]} \sum_{j=1}^{d_2} |a_{ij}|$, $||\bA||_{\max}=\max_{i \in [d_1], j \in [d_2]} |a_{ij}|$ for any real matrix $\bA \in \mathbb{R}^{d_1 \times d_2}$.
\begin{prop}\label{prop:LPD_solution}
	For a given symmetric matrix $\bA = (a_{ij})_{1 \le i,j \le p}$ with positive diagonals, assume $\lambda_{\min}(\bA) {< 0} < \epsilon \le \mu$. The linear shrinkage $\Phi_{\mu, \alpha^*}$ of $\bA$ achieves the minimum at different values of $\mu$ according to different matrix norms.
	
	\begin{enumerate}
		\item (Spectral norm, Lemma 2 of \cite{Choi:2019})
		$$
		\left\| \bA - \Phi_{\mu, \alpha^*} \right\|_2 = \epsilon -\lambda_{\min}(\bA)
		$$
		for any	$\mu \ge \max\{\epsilon, (\lambda_{\max}(\bA) + \lambda_{\min}(\bA)) / 2\}$.
		
		\item ((Scaled) Frobenius norm, Lemma 3 of \cite{Choi:2019})		
		$$
		\left\| \bA - \Phi_{\mu^*_F, \alpha^*} \right\|_F = (\epsilon -\lambda_{\min}(\bA))\sqrt{\mu^*_F}
		$$
		where  $\mu^*_F = \sum_{j=1}^p (\lambda_j (\bA) - \bar{\lambda})^2 / \sum_{j=1}^p (\lambda_j (\bA) - \lambda_{\min}(\bA))^2$ and $\bar{\lambda}$ is an average of the eigenvalues of $\bA$, $\lambda_1(\bA), \ldots, \lambda_p (\bA)$.

		\item ($\ell_\infty$-norm, Lemma 3 of \cite{Cho:2021})
		
		$$
		\begin{array}{l}
		\left\| \bA - \Phi_{\mu, \alpha^*} \right\|_{\infty} \\
		= \begin{cases}
		\searrow \epsilon -\lambda_{\min}(\bA) \text{ as } \mu \to \infty, & \text{if } \lambda_{\min}(\bA) + M_2 > 0, \\
		\epsilon -\lambda_{\min}(\bA), & \text{for any } \mu \ge (M_1 - M_2) / 2, \\
		& \qquad \text{if } \lambda_{\min}(\bA) + M_2 = 0,\\
		(\epsilon -\lambda_{\min}(\bA))\dfrac{(M_1 + M_2)/2}{(M_1 - M_2)/2 - \lambda_{\min}(\bA)}, & \text{at } \mu = (M_1 - M_2) / 2, \\
		& \qquad  \text{if } \lambda_{\min}(\bA) + M_2 < 0,
		\end{cases}
		\end{array}
		$$
		where $M_1 = \max_j \big(a_{jj} + \sum_{i:i\neq j} |a_{ij}| \big)$ and $M_2 = \max_j \big(-a_{jj} + \sum_{i:i\neq j} |a_{ij}|\big)$. Note that if $\lambda_{\min}(\bA) + M_2 > 0$, there is no solution.
		
		\item (Element-wise maximum norm)
		$$
		\begin{array}{l}
		\left\| \bA - \Phi_{\mu, \alpha^*} \right\|_{\max} \\
		= \begin{cases}
		\dfrac{(\epsilon -\lambda_{\min}(\bA))(a_{d,\max} - a_{d,\min})/2}{(a_{d,\max} + a_{d,\min})/2 - \lambda_{\min}(\bA)}, &
		\text{at } \mu = (a_{d,\max} + a_{d,\min})/2,\\
		& \quad \text{if } (a_{d,\max} - a_{d,\min})/2  > a_{\text{off}, \max},\\
		\dfrac{(\epsilon -\lambda_{\min}(\bA)) a_{\text{off}, \max}}{a_{d,\min} + a_{\text{off}, \max} - \lambda_{\min}(\bA)}, & 
		\text{at } \mu = a_{d,\min} + a_{\text{off}, \max},\\
		& \quad \text{if } (a_{d,\max} - a_{d,\min})/2  \le a_{\text{off}, \max}.
		\end{cases}
		\end{array}
		$$
		where $a_{d,\max} = \max_j a_{jj}$, $a_{d,\min} = \min_j a_{jj}$, and $a_{\text{off}, \max} = \max_{i\neq j} |a_{ij}|$.		
	\end{enumerate}
\end{prop}
\noindent
We only provide a proof of the last case of Proposition \ref{prop:LPD_solution}, which is in Supplementary Materials \ref{sec:pf_LPD_solution}, and for the others we refer readers to the original references. 
It should be noted that in some cases, for example, when the spectral norm is used, any choice of $\mu$ beyond some threshold is sufficient for the optimality of shrinkage. Thus, one may simply pick $\mu$ that is large enough depending on the context of the data considered. However, the choice is not sensitive in practice, which is verified in our simulation study where different candidates of $\mu$ are compared.

\subsection{Main results for consistency}

In this section, we check the two conditions in Proposition \ref{prop:consistency}, and compute the convergence rate of $\widehat{\bbeta}^{\rm LPD}$ in $\ell_2$-norm. Prior to it, we state the assumptions and data structure more precisely.

%
%
%
%

We introduce binary random variables that indicate whether each entry of data is observed or not: $\delta^y_{i}=\text{I}(y_i \text{ is observed})$, $\delta^x_{ij}=\text{I}(x_{ij} \text{ is observed})$, $i=1,\ldots, n$, $j=1,\ldots, p$. Then, we can concisely express the observed data by the product of the indicator variable and the data, i.e. 
$\tilde{y}_i = \delta^y_i y_i, \tilde{x}_{ij}= \delta^x_{ij} x_{ij}$, which is equivalent to (\ref{eq:data_str}).

We define the sub-Gaussian (or $\psi_2$-) norm of a random variable $X$ in $\mathbb{R}$ by
$$
||X||_{\psi_2} = \sup_{p \ge 1} \dfrac{(\mathbb{E}|X|^p)^{1/p}}{\sqrt{p}},
$$
and $X$ is called sub-Gaussian if its $\psi_2$-norm is bounded. Under the regression setting (\ref{eq:reg_model}), we assume the following.
\begin{assumption}\label{assm:sub_G}
	For $i=1,\ldots, n$, $\max\limits_{1\le j \le p} ||x_{ij}/\sqrt{\sigma_{jj}}||_{\psi_2} \le K^x$ and $||\epsilon_i/\sqrt{\sigma_{\epsilon\epsilon}}||_{\psi_2} \le K^{\epsilon}$, where $\sigma_{jj}=\text{Var}(x_{1j}), \sigma_{\epsilon\epsilon}=\text{Var}(\epsilon_1)$.
\end{assumption}
\noindent
Assume the indicators are Bernoulli variables with general dependency structure \citep{Dai:2013, Park:2021}, that is:
\begin{assumption}\label{assm:multi_bern}
	For $i=1,\ldots, n$, $(\delta^y_{i}, \delta^x_{i1}, \ldots, \delta^x_{ip})$ is from the multivariate Bernoulli distribution with the first two moments written by
	$$
	\mathbb{E}\delta^x_{ij} = \pi^{xx}_{jj}, \quad
	\mathbb{E}\delta^x_{ij}\delta^y_{i} = \pi^{xy}_{j}, \quad
	\mathbb{E}\delta^x_{ij}\delta^x_{ik} = \pi^{xx}_{jk}.
	$$
	More general moment is denoted as $\mathbb{E}\delta^x_{ij_1}\delta^x_{ij_2}\delta^x_{ij_3}\cdots = \pi^{xx}_{j_1j_2j_3\ldots}$.
\end{assumption}
\noindent
The missing mechanism we use is the missing completely at random (MCAR). In the current data structure, we can specify the assumption as follows.
\begin{assumption}\label{assm:mcar}
	The data and indicator variables are independent, i.e. 
	$$
	\{\epsilon_i, \bx_i\} \indep \{\delta^y_i, \delta^x_{i1}, \ldots, \delta^x_{ip}\}, \quad i=1,\ldots, n.
	$$
\end{assumption}
\noindent
The last assumption is about the class of covariance matrices for the covariates. Without loss of generality, assume the variables of interest (i.e. in the set $\mathcal{A}$) are located in front and the covariance matrix $\bSigma$ is decomposed in blocks accordingly.
\begin{assumption}\label{assm:class_cov}
	Assume the population covariance matrix $\bSigma=\text{Cov}(\bx_i)$ satisfies 
	\begin{enumerate}[(a)]
		\item $\bSigma_{\mathcal{A}\mathcal{A}}$ is positive definite, and
		
		\item the irrepresentability condition for $\bSigma$ is satisfied with respect to the support set $\mathcal{A}$, i.e., there exists $\tau \in (0,1)$ such that
		$\|\bSigma_{\mathcal{A}^c \mathcal{A}} \bSigma_{\mathcal{A} \mathcal{A}}^{-1} \|_{\infty} \le 1 - \tau$.
	\end{enumerate}	
\end{assumption}
\noindent
The first condition that the smallest eigenvalue is away from zero is not very restrictive, and the other condition is known to be sufficient and ``almost'' necessary for selection consistency \citep{Geer:2009, Lee:2015,Wainwright:2009}.

{
	Throughout this section, we define the LPD estimator as follows.
	If $\lambda_{\min}(\widehat{\bSigma}^{\rm IPW}) > 0$, construct the LPD estimator $\Phi_{\mu, \alpha}(\widehat{\bSigma}^{\rm IPW})$ by choosing $\alpha=1$ {(and any real-valued $\mu$).} {Otherwise, for $\epsilon>0$ such that $\epsilon < \lambda_{\min}(\bSigma)$}, set $\alpha =  (\mu - \epsilon) / (\mu - \lambda_{\min}(\widehat{\bSigma}^{\rm IPW}))$ and choose any $\mu$ greater than $2\epsilon$.}
Based on the assumptions, we present results that guarantee the two conditions (C1) and (C3) in Proposition \ref{prop:consistency} with high probability. 
\begin{thm}[Irrepresentability condition and RSC condition]\label{thm:IR_LPD_combined}
	Let Assumption \ref{assm:sub_G}, \ref{assm:multi_bern}, \ref{assm:mcar}, \ref{assm:class_cov} hold.
	Assume $\widehat{\bSigma}^{\rm IPW}_{\mathcal{A} \mathcal{A}}$ is non-singular. 
	Then, the LPD estimator satisfies the irrepresentability condition for some constant $\tilde{\tau} \in (0,1)$ with probability greater than $1-3/p^u$ for $u > 0$ if the sample size satisfies 
	$$
	\dfrac{n}{\pi_{\max}^{(4)} \log p} \ge c \bigg\{ 
	\dfrac{\tr(\bSigma) \max\{(K^x)^2, 1\} \sqrt{u+1}}{\min\{\tau/\left\|\bSigma_{\mathcal{A}\mathcal{A}}^{-1}\right\|_{\infty}, {  \lambda_{\min}(\bSigma_{\mathcal{A}\mathcal{A}})} \}}
	\bigg\}^2, \quad n > c \,\pi_{\max}^{(4)} (u + 1)^3 \log^3(p \vee n),
	$$
	for some $c>0$. 
	Here, $\pi_{\max}^{(4)} = \max_{k_1, k_2,\ell_1, \ell_2} \pi^{xx}_{k_1 k_2 \ell_1 \ell_2}/(\pi^{xx}_{k_1 \ell_1} \pi^{xx}_{k_2\ell_2})$. Moreover, under the same conditions, (C3) of Proposition \ref{prop:consistency} holds; {if $\lambda_{\min}(\widehat{\bSigma}^{\rm IPW}) > 0$, $\mu$ is excluded in the lower bound of (C3).}
\end{thm}
\noindent
To prove the theorem, we first show in Theorem \ref{thm:IR_LPD_neg} and \ref{thm:IR_LPD_pos} that the irrepresentability condition holds for $\widehat{\bSigma}^{\rm LPD}$ if $\bSigma$ is in the small neighborhood of the IPW estimator in terms of $\ell_\infty$, $2$-norms. The probability of being in the neighborhood is calculated in the proof of Theorem \ref{thm:IR_LPD_combined}. Technical details can be found in Supplementary Materials \ref{sec:pf_thm_IR}. 
In Lemma 6 of \cite{Datta:2017}, they also showed similar results: if a surrogate estimator $\tilde{\bSigma}$, which is the LPD estimator in our context, is close enough to $\bSigma$, then $\tilde{\bSigma}_{\mathcal{A}^c \mathcal{A}}\tilde{\bSigma}_{\mathcal{A} \mathcal{A}}^{-1}$ is to $\bSigma_{\mathcal{A}^c \mathcal{A}}\bSigma_{\mathcal{A} \mathcal{A}}^{-1}$. In the theorem below, we use a new notation $\|\bB\|_{\infty, \mathcal{A}} = \max\limits_{1 \le j \le p} \sum\limits_{k\in \mathcal{A}} |b_{jk}|$.

The following guarantees (C2) of Proposition \ref{prop:consistency} with high probability.
\begin{thm}[Bound on the gradient]\label{thm:BG}
	Let Assumption \ref{assm:sub_G}, \ref{assm:multi_bern}, \ref{assm:mcar} hold.
	Then, if $n$ and $p$ satisfy 
	$$
	n > c \max\Big\{\log p / \pi^{xy}_{\min}, \pi^{(4)}_{\max} \log^3(p \vee n) \Big\}
	$$ 
	for some $c>0$, the gradient vector of the mean squared error satisfies the upper bound with probability greater than $1-9/p$
	$$
	\left\|\Phi_{\mu, \alpha}(\widehat{\bSigma}^{\rm IPW})\bbeta^* - \hat{\brho}^{\rm IPW}\right\|_{\infty} \le 
	L |\mathcal{A}|\sqrt{\dfrac{\log p}{n}},
	$$
	where $L>0$ is a function of parameters given by
	$$
	\begin{array}{l}
	L = C_1 \beta_{\max}^*  \max\{(K^x)^2, 1\}  \sqrt{\pi_{\max}^{(4)}} \cdot h_1(\mu; \bSigma, \mathcal{A}) + 
	C_2 \dfrac{\max\Big\{\sqrt{\sigma_{\max} \sigma_{\epsilon\epsilon}} K^x  K^{\epsilon}, \sigma_{\max} (K^x)^2\Big\}}{\sqrt{\pi^{xy}_{\min}}},
	\end{array}
	$$
	for some positive constants $C_1, C_2$.
	Here, 
	$\pi_{\max}^{(4)} = \max\limits_{k_1, k_2,\ell_1, \ell_2} \pi^{xx}_{k_1 k_2 \ell_1 \ell_2}/(\pi^{xx}_{k_1 \ell_1} \pi^{xx}_{k_2\ell_2})$, $ \pi^{xy}_{\min} = \min_{k} \pi_{k}^{xy}$, $\beta_{\max}^* = \max\limits_{1 \le j \le p} |\beta^*_j|$, and $h_1(\mu; \bSigma, \mathcal{A}) =  \tr(\bSigma)\big(1 + \|\bSigma \|_{\infty, \mathcal{A}}/\mu \big)$ if $\lambda_{\min}(\widehat{\bSigma}^{\rm IPW}) \le 0$ and $\sigma_{\max}$ otherwise.
\end{thm}
\noindent
Proof of the theorem can be found in Supplementary Materials \ref{sec:pf_thm_BG}. 
\cite{Loh:2013, Loh:2017} also required the bounded gradient condition (see Theorem 1 in \cite{Loh:2013} or \cite{Loh:2017}). Also, one remarks that dependency of the bound on $\beta^*_{\max}$ is similarly observed in the literature of missing data (see SNR conditions in \cite{Chen:2013, Datta:2017}; Theorem 1 in \cite{Rosenbaum:2010}). 

Combining these results with Proposition \ref{prop:consistency}, we present the properties of the solution $\widehat{\bbeta}^{\rm LPD}$ of (\ref{eqn:lasso-problem}). 
\begin{thm} \label{thm:consistency_LPD}
	Let Assumption \ref{assm:sub_G}, \ref{assm:multi_bern}, \ref{assm:mcar}, \ref{assm:class_cov} hold.
	Assume $\widehat{\bSigma}^{\rm IPW}_{\mathcal{A} \mathcal{A}}$ is non-singular. 
	We choose the tuning parameter $\lambda \propto L |\mathcal{A}|(\log p/n)^{1/2}$ for the lasso regression. 
	If $n$ and $p$ satisfy 
	$$
	\dfrac{n}{\pi_{\max}^{(4)} \log p} \ge c \bigg\{ 
	\dfrac{\tr(\bSigma) \max\{(K^x)^2, 1\}}{\min\{\tau/\left\|\bSigma_{\mathcal{A}\mathcal{A}}^{-1}\right\|_{\infty}, { \lambda_{\min}(\bSigma_{\mathcal{A}\mathcal{A}}) } \}}
	\bigg\}^2,
	\quad 
	n > c \max\Big\{
	\dfrac{\log p}{\pi^{xy}_{\min}}, \pi^{(4)}_{\max} \log^3(p \vee n)
	\Big\}
	$$
	for some $c>0$,
	then there exist some $C>0, d>0, \tilde{\tau} \in (0,1)$ such that we can guarantee with probability greater than $1 - d/p$
	\begin{align*}
	&(R1) \quad \text{The minimizer }\widehat{\bbeta}^{\rm LPD} \text{ is unique.} \\
	&(R2) \quad \|\widehat{\bbeta}^{\rm LPD} - {\bbeta}^*\|_2 \le 
	C\frac{L}{ \tilde{\tau} \cdot h_2(\mu, \lambda_{\min}(\bSigma_{\mathcal{A}\mathcal{A}}))} 
	\sqrt{\dfrac{|\mathcal{A}|^{3}\log p}{n}} \\
	&(R3) \quad \hat{\beta}^{\rm LPD}_{j} = 0, \quad j \in \mathcal{A}^c.
	\end{align*}
	{
		Here, $h_2(\mu, \lambda_{\min}(\bSigma_{\mathcal{A}\mathcal{A}})) = \min\{ \lambda_{\min}(\bSigma_{\mathcal{A}\mathcal{A}}), \mu \}$ if $\lambda_{\min}(\widehat{\bSigma}^{\rm IPW}) \le 0$ and $\lambda_{\min}(\bSigma_{\mathcal{A}\mathcal{A}})$ otherwise.}
	The factor $L$ appears in Theorem \ref{thm:BG}.
\end{thm}

\noindent
We have some remarks regarding this main result.
First, the results hold regardless of the choice of matrix norms in (\ref{eqn:dist-min}) because the optimal choice of $\alpha$ in LPD is independent of the matrix norms. Also, no terms are involved with $\epsilon$ in the theorems, though the actual performance of LPD can change according to different $\epsilon$ due to the numerical stability.

Second, the constant $L$ depends on $\tr(\bSigma)$, which is an order of $p$ in general. {This trace term is introduced when we control the magnitude of the gradient vector of the loss function based on the LPD. This condition related to the gradient vector is commonly used in literature (e.g. (3.1) of \cite{Loh:2012}). We believe that the additional factor is the expense we need to pay for convexification of the loss function.}
However, {as in the literature on covariance estimation (\cite{Lounici:2014,Mendelson:2020,Koltchinskii:2017})}, we can express the trace of $\bSigma$ by the effective rank that measures intrinsic dimension of a symmetric matrix, defined by $\textbf{r}(\bSigma)=\tr(\bSigma)/ ||\bSigma||_2$. Note that $\textbf{r}(\bSigma) \le \text{rank}(\bSigma) \le p$ for general matrices, but the effective rank would be much smaller than $p$ if $\bSigma$ is approximately low-rank. See more discussion in Section 2.2 of \cite{Lounici:2014} or Remark 5.53 of \cite{Vershynin:2011}. 
Hence, the constant $L$ would not depend on $p$ if we consider a class of covariance matrices satisfying that
(1) approximately low-rank, or $\text{r}(\bSigma):=\text{tr}(\bSigma) / ||\bSigma||_2 \le R$ (independent of $p$) and (2) the largest eigenvalue is bounded, or $||\bSigma||_2 \le B$ (independent of $p$).
Then, Theorem \ref{thm:consistency_LPD} states that {under this class of distributions for covariates,} the sample size $n \gtrsim \log p$ is enough to guarantee that the solution $\widehat{\bbeta}^{\rm LPD}$ is (R1) unique, (R2) $\ell_2$-consistent, and (R3) has no false positive with probability close to $1$.

Third, we would like to compare our result with the ones previously obtained in \cite{Datta:2017} and \cite{Loh:2012}.  To facilitate a fair comparison, we reorganize all the results into the following format:
if the sample size and dimension satisfies $n/\log p > \mathcal{M}$, then with probability at least $1 - c / p$, it holds that	
$$
||\widehat{\beta} - \beta^*||_2 \le 
C \cdot \mathcal{L} \cdot |\mathcal{A}|^{\mathcal{K}} \sqrt{\dfrac{\log p}{n}},
$$
where $c, C>0$ are some positive constants. Here, $\widehat{\beta}$ is a coefficient estimator from one of \cite{Datta:2017}, \cite{Loh:2012}, or the proposed,  and $\beta^*$ is the true value to be estimated. The specific forms of $\mathcal{K}$, $\mathcal{L}$, and $\mathcal{M}$ depend on parameters such as (but not limited to) (1) observation probability, (2) tail thickness (or sub-Gaussian parameter) of the response variable, (3) tail thickness of the covariates, (4) covariance matrix of the covariates. While the triplet $(\mathcal{K},\mathcal{L}, \mathcal{M})$ is not directly comparable as each paper uses slightly different assumptions, we aim to highlight the general tendencies.

The convergence rate $\mathcal{L}$ commonly depends on (1) observation probability, (2) tail thickness (or sub-Gaussian parameter) of the response variable, (3) tail thickness of the covariates, (4) magnitude of the true value $\beta^*$, and (5) well-conditionedness of $\bSigma$. Regarding (5), the result from \cite{Loh:2012} is $\mathcal{L} \propto 1 / \lambda_{\min}(\bSigma)$, while \cite{Datta:2017} obtained $\mathcal{L} \propto 1 / \Omega$, where  
$$
\Omega := \min_{x \in \mathcal{R}} x^\top \bSigma x, \qquad \mathcal{R}=\{x: ||x||_2=1, ||x_{\mathcal{A}^c}||_1 \le 3 ||x_{\mathcal{A}}||_1\},
$$
which is related to the compatibility condition. In contrast, our result satisfies $\mathcal{L} \propto 1 / \{\tilde{\tau} \cdot (\lambda_{\min}(\bSigma_{\mathcal{A}\mathcal{A}}) \wedge  \mu )\}$, where $\tilde{\tau}$ is a constant from the irrepresentability condition of the LPD estimator. 
{Similar quantities have appeared from restricted strong convexity in the related context (\cite{Negahban:2012}), typically with the same order of 1 in the denominator.}
The rate from \cite{Loh:2012} would get worse if the covariance matrix from covariates on $\mathcal{A}^c$ is ill-conditioned, while the other two are not affected. Additionally, while our result depends on $\mu$  (the tuning parameter of LPD procedure), this dependency is negligible if $\mu$ is chosen sufficiently large, i.e., $\mu > \lambda_{\min}(\bSigma_{\mathcal{A}\mathcal{A}})$. Lastly, our result has dependency on $\tr(\bSigma)$, i.e. $\mathcal{L} \propto  \tr(\bSigma)$.

The constant $\mathcal{M}$ characterizes the sample size required to guarantee the derived convergence rate. Across all three methods, the constant depends on (1) observation probability, (2) tail thickness of the covariates, and (3) well-conditionedness of $\bSigma$. The dependency on (3) is similar to that of $\mathcal{L}$. More specifically, 
$$
\mathcal{M}_{\text{Loh}} \propto 1 / \lambda_{\min}(\bSigma)^2, \quad
\mathcal{M}_{\text{Datta}} \propto 1 / \min\{C_1\tau^2, C_2\Omega^2\}, \quad
\mathcal{M}_{\text{{Park}}} \propto 1 / \{\tau \cdot \lambda_{\min}(\bSigma_{\mathcal{A}\mathcal{A}})\}^2
$$
where $C_1, C_2>0$ are constants.
In \cite{Datta:2017}, $\mathcal{M}$ also depends on $\beta_{\max}^*$ and the tail thickness of the response variable. In our case, $\mathcal{M} \propto \tr(\bSigma)$, which can be explained similarly to its appearance in $\mathcal{L}$.

The constant $\mathcal{K}$ represents the order of sparsity in the convergence rate. Both \cite{Datta:2017} and our result share the same order $\mathcal{K}=3/2$, while \cite{Loh:2012} achieves a smaller order $\mathcal{K}=1$. 
{The order of sparsity may have room for improvement in proof techniques, as the exponent $\mathcal{K}=1/2$ in $|\mathcal{A}|^{\mathcal{K}}$ is commonly observed in the high-dimensional regression literature (e.g. \cite{Geer:2009,Wainwright:2009,Negahban:2012}). In contrast, our result yields  $\mathcal{K}=3/2$, which is attributed to the linear shrinkage of the non-PD matrix. This can also be seen as a cost incurred for convexification.}

In conclusion, this comparison shows that our method still gurantees similar results from the previous work, but with an extra term $\tr(\bSigma)$. Theoretically, this difference is the price we need to pay for convexification and faster computation. However, for a smaller class of covariance matrices (e.g., low-rank and bounded largest eigenvalue), this term becomes negligible.

\subsection{Estimation of unknown parameters}

It should be noted that our results are based on two implicit assumptions. First, we assume the observation probabilities are known, as in other error-in-variable literatures (\cite{Datta:2017,Sorensen:2015}). Second, following a convention in a regression framework, we also assume covariates are centered, i.e. mean-zero. However, these may not be the case in real-world data, and thus we would like to leave some remarks regarding these assumptions.

For estimating the observation probabilities, it is natural to use the empirical proportions (i.e. the proportion of observed pairs) under MCAR, due to the law of large numbers. In other words, we suggest using $\hat{\pi}_{jk}^{xx} = \sum_{i=1}^n \delta_{ij}^x\delta_{ik}^x / n$ and $\hat{\pi}_{j}^{xy} = \sum_{i=1}^n \delta_{ij}^x\delta_{i}^y / n$.
Then, the new IPW estimator is 
$$
\widehat{\bSigma}^{\text{IPW},\hat{\pi}} = \Big((\widehat{\bSigma}^{\text{IPW}})_{jk} \dfrac{\pi_{jk}^{xx}}{\hat{\pi}_{jk}^{xx}}, ~ 1\le j, k \le p \Big).
$$
We have found throughout our numerical study that the penalized regression based on the above estimator performs quite well.

Next, we consider the case when covariates may have non-zero means. The most straightforward way is to center each covariate by the IPW mean estimator $\hat{\mu}_j = \dfrac{\sum_{i=1}^n \tilde{x}_{ij}}{n\pi^{xx}_{jj}}$. As used in \cite{Kolar:2012} and \cite{Cai:2016:JMA}, this type of IPW estimator is defined by
$$
\widehat{\bSigma}^{\text{IPW},2}_{jk} = \sum_{i=1}^n \delta_{ij}^x\delta_{ik}^x(\tilde{x}_{ij}-\hat{\mu}_j) (\tilde{x}_{ik}-\hat{\mu}_k) /(n \pi_{jk}^{xx}).
$$
However, this is not unbiased (in finite sample), which often complicates theoretical analyses (e.g. concentration inequality). To address it, we proposed another type of IPW estimator in our earlier work (\cite{Park:2021}):
$$
\widehat{\bSigma}^{\text{IPW},3}_{jk}=
\dfrac{\sum_{i=1}^n \tilde{x}_{ij} \tilde{x}_{ik}}{n\pi_{jk}^{xx}} - 
\dfrac{\sum_{i \neq i'}^n \tilde{x}_{ij} \tilde{x}_{i'k}}{n (n-1) \pi_{jj}^{xx} \pi_{kk}^{xx}}.
$$

We remark that our theory is based on two types of concentration inequalities for IPW estimators: one is about the element-wise maximum norm and the other is the spectral norm. The former has been investigated in our earlier work (\cite{Park:2022_stat}), but the latter has not yet in literature. Though we tried to derive the non-asymptotic inequality based on the spectral norm, it is not as simple as the other. 
We think including such an analysis in this paper would be unnecessarily complicated, and thus leave it as our future work.

\section{Numerical study}

We showcase the empirical performance of the proposed estimator LPD based on different simulation parameters (e.g. dimension $p$, missing rate of observations, covariance structure for variables). 
Our analysis consists of three parts. In the first part, we compare several methods including two existing ones and the proposed one based on different choices of $\mu$. In the second, we examine  how sensitive the models are to missing values. In the third, we time an algorithm of each method to see their scalability. 

It has to be noted that a simulation study performed by \cite{Romeo:2019} compared a group of methods available until then, but only considered additive measurement error models. In the meantime, our simulation study deals with missing data cases, which is clearly different from what was covered in their work.

\subsection{Setting}

We adopt experimental settings of \cite{Sorensen:2019} where they generate responses from the normal model, i.e.
$$
\tilde{\by} \sim N_n(\tilde{\bX} \bbeta^*, \sigma_y^2 \bI),
$$
and each row of the design matrix $\tilde{\bX}$ from $N(\bzero, \bSigma)$ where the covariance structure is the compound symmetry ($\bSigma_{ij}=0.5^{\text{I}(i\neq j)}$). The dimension $p$ of covariates varies over $p=200, 500$. 
The regression coefficients $\bbeta^*$ have non-zero values at random positions while keeping the proportion of them at $s=0.05, 0.1$ (i.e. $s$ is the level of sparsity). The non-zero coefficients are all equal to $1$. We fix $n=200$ and $\sigma_y = 3$.

Responses and covariates are subject to missing completely at random (MCAR). More specifically, we define matrices of missing indicators: $\bM_y = (\delta^y_i)$ and $\bM_X = (\delta^x_{ij})$ where $\delta^y_i \sim \text{Ber}(\theta)$, $\delta^x_{i,3j} \sim \text{Ber}(\theta)$, $j=1,\ldots, \lfloor p/3 \rfloor$, independently.
Then, the corrupted data are
$$
\by = \tilde{\by} * \bM_y, \quad \bX = \tilde{\bX} * \bM_X,
$$
where $*$ is the element-wise product. Other missing mechanisms (MAR, MNAR) will be discussed in Section \ref{sec:simulation_missing}. We control the observation probability $\theta=0.7, 0.9$. We generate 100 independent datasets to consider random variability.

Given incomplete data $(\by, \bX)$, we compute three comparative estimators: (1) linear shrinkage positive definite lasso (LPD), (2) convex conditioned lasso (CoCo) \citep{Datta:2017}, and (3) non-convex lasso (NCL) \citep{Loh:2012}. We use the R package
named \pkg{BDcocolasso} (available at \url{https://github.com/celiaescribe/BDcocolasso}) implemented by \cite{Escribe:2021} to obtain the second estimator and \pkg{hdme} \citep{Sorensen:2019} to obtain the third.
Additionally, we add two types of lasso regression in comparison. 
One uses the complete data $(\tilde{\by}, \tilde{\bX})$ and is named (4) ``true lasso'', while the other runs the lasso regression with mean imputed data and is named (5) ``naive lasso''. 
We do not include the complete-case analysis as none of the samples are completely observed in high-dimensional missing data. For instance, in the real data we analyzed, every cell line has at least 48 missing values, making the straighforward approach impractical.

In terms of LPD, we can consider a set of variants based on different choices of $\mu$, but found that LPD using $\ell_\infty$-norm empirically works well and is robust to different setups. Hence, for readability, we only report the corresponding results in this section, while the entire results are provided in Supplementary Materials \ref{sec:supp_method} and \ref{sec:supp_missing}.

The penalized regression methods mentioned earlier have hyperparameters to be tuned. To choose a penalty parameter $\lambda$ of CoCo and LPD, we use the corrected cross-validation proposed in \cite{Datta:2017}, that is, the cross-validation approach adjusted for corrupted data. Simply put, the idea is to minimize the mean square prediction error where a non-PD covariance matrix estimate is replaced by the PD matrix. More details can be found in Supplementary Materials \ref{app:corr_cv}. The grids are evenly spaced in log scale within the interval $[R/10000, R]$ where $R=2 ||\br_{\text{naive}}||_{\max}$ and $\br_{\text{naive}}$ is the naive lasso estimator. If $R=0$ (i.e. $\br_{\text{naive}}=0$), then we set $R$ by $||\bX^\top \by / n||_{\max}$. For NCL, we need to decide the radius $b$ such that the solution satisfies $||\widehat{\bbeta}||_1 \le b$. We search the optimal radius over the grid in $[R/10000, R]$ with $R=2 ||\br_{\text{naive}}||_1$. The number of grid points is $100$ throughout. Using the optimal tuning parameter, we re-fit each model and have the estimates of coefficients.

We measure six criteria to assess performance of each method. Following \cite{Datta:2017}, we compute the prediction error (PE) and mean squared error (MSE), which is respectively defined 
$$
PE(\widehat{\bbeta}) = (\widehat{\bbeta} - \bbeta^*)^\top \bSigma (\widehat{\bbeta} - \bbeta^*), \quad MSE(\widehat{\bbeta}) = (\widehat{\bbeta} - \bbeta^*)^\top (\widehat{\bbeta} - \bbeta^*).
$$
The number of covariates corrected/incorrectly identified (TP and FP) are also counted.
To see an overall accuracy of variable selection, we also compute the (partial) area under the ROC curve (pAUC) and F$_1$-score (harmonic mean of precision and recall) denoted by F1.
We also measure the time each method would take to finish. This includes the tuning parameter search. 

\subsection{Method}\label{sec:simulation_method}

In this experiment, we compare different regression methods. To reduce the workload of simulations, we fix $\theta=0.9$ under MCAR.

\begin{table}[H]
	\tiny
	\centering
	\begin{tabular}{|c|c|c|c|c|c|c|}
		\hline & \multicolumn{6}{|c|}{$ p=200,s=0.05 $}\\
		\cline{2-7}
		& PE & MSE & pAUC & F$_1$ & TP & FP \\
		\hline TL & 1.892 (0.601)& 3.653 (1.162)& 0.953 (0.032)& 0.439 (0.065)& 9.700 (0.482)& 25.370 (6.935) \\ 
		NL & 3.710 (1.279)& 6.186 (1.950)& 0.873 (0.075)& 0.397 (0.076)& 8.560 (1.157)& 25.590 (7.732) \\ 
		CoCo & 3.490 (1.276)& 6.641 (2.424)& 0.816 (0.073)& 0.398 (0.083)& 8.370 (1.236)& 24.650 (6.658) \\ 
		NCL & 5.162 (1.337)& 6.447 (1.820)& 0.519 (0.083)& 0.439 (0.118)& 8.140 (1.477)& 21.800 (15.525) \\ 
		LPD & 3.352 (1.000)& 6.320 (1.824)& 0.873 (0.070)& 0.369 (0.066)& 8.790 (1.104)& 29.710 (7.312) \\ 
		\hline & \multicolumn{6}{|c|}{$ p=500,s=0.05 $}\\
		\cline{2-7}
		& PE & MSE & pAUC & F$_1$ & TP & FP \\
		\hline TL & 6.073 (1.243)& 11.940 (2.433)& 0.815 (0.044)& 0.420 (0.054)& 22.980 (1.239)& 63.190 (16.677) \\ 
		NL & 16.327 (4.124)& 26.382 (4.161)& 0.555 (0.084)& 0.298 (0.060)& 13.130 (3.084)& 49.950 (9.090) \\ 
		CoCo & 15.738 (3.154)& 30.083 (5.651)& 0.600 (0.044)& 0.290 (0.062)& 12.530 (3.119)& 48.810 (9.018) \\ 
		NCL & 27.640 (7.481)& 26.873 (3.507)& 0.506 (0.062)& 0.218 (0.055)& 14.810 (5.025)& 105.450 (55.242) \\ 
		LPD & 13.375 (2.323)& 25.482 (3.883)& 0.717 (0.064)& 0.262 (0.050)& 15.250 (3.141)& 76.730 (16.213) \\ 
		\hline & \multicolumn{6}{|c|}{$ p=200,s=0.1 $}\\
		\cline{2-7}
		& PE & MSE & pAUC & F$_1$ & TP & FP \\
		\hline TL & 3.240 (0.841)& 6.263 (1.631)& 0.915 (0.034)& 0.535 (0.060)& 19.600 (0.651)& 34.570 (8.335) \\ 
		NL & 10.299 (3.229)& 15.240 (3.293)& 0.761 (0.068)& 0.438 (0.062)& 14.400 (2.340)& 31.500 (5.458) \\ 
		CoCo & 9.361 (2.429)& 17.288 (4.059)& 0.723 (0.055)& 0.437 (0.070)& 13.880 (2.341)& 29.950 (6.660) \\ 
		NCL & 16.726 (3.676)& 17.447 (2.445)& 0.617 (0.046)& 0.398 (0.099)& 14.170 (2.775)& 42.950 (26.712) \\ 
		LPD & 8.477 (2.144)& 15.565 (3.406)& 0.774 (0.060)& 0.419 (0.057)& 14.970 (2.115)& 36.940 (7.678) \\ 
		\hline & \multicolumn{6}{|c|}{$ p=500,s=0.1 $}\\
		\cline{2-7}
		& PE & MSE & pAUC & F$_1$ & TP & FP \\
		\hline TL & 14.001 (2.440)& 27.630 (4.914)& 0.683 (0.049)& 0.477 (0.048)& 43.950 (2.488)& 91.930 (18.908) \\ 
		NL & 48.644 (11.035)& 77.535 (11.147)& 0.391 (0.057)& 0.269 (0.055)& 16.770 (3.928)& 57.530 (9.157) \\ 
		CoCo & 47.577 (8.028)& 91.880 (15.888)& 0.548 (0.033)& 0.259 (0.051)& 15.560 (3.529)& 54.000 (8.060) \\ 
		NCL & 76.542 (26.472)& 65.129 (11.035)& 0.489 (0.039)& 0.241 (0.036)& 24.940 (7.538)& 129.610 (44.213) \\ 
		LPD & 37.225 (5.155)& 71.559 (9.319)& 0.606 (0.043)& 0.267 (0.045)& 21.020 (4.259)& 86.310 (15.103) \\ 
		\hline
	\end{tabular}		
	\caption{\footnotesize Method comparison for $p=200, 500$ and $s=0.05, 0.1$. Each performance measure is averaged over $R=100$ repetitions (standard deviation in parenthesis).}
\end{table}

Compared with the existing methods (CoCo, NCL), LPD is less sparser and has more TP and FP. LPD is proved to be successful in estimation (low MSE), prediction (low PE), and variable selection (high pAUC, high TP). Though the difference is negligible considering standard deviation, LPD performs best in almost all scenarios of the finite sample setting. This result is of great importance since LPD is much faster than its competitors (see Table \ref{tab:timing}). 
The naive lasso (NL) seems to have smaller MSE and higher F$_1$-score than LPD, but it sharply deteriorates when $p$ increases. Compared to it, LPD performs nearly best for all cases considered.

Though its more restrictive structure in LPD than CoCo, it shows the superior performance in the finite sample study. We believe this is because LPD preserves the off-diagonal elements of the initial estimator. That is, LPD does not change information about the covariance part. In constrast, CoCo focuses on element-wise approximation, which may lose such information. As a result, CoCo has good theoretical support, but LPD offers a more practical solution.

\subsection{Missing rate and missing mechanism}\label{sec:simulation_missing}

We try different missing rates and mechanisms to investigate the robustness of each method under other scenarios of missing data generation. This is similar to the idea of sensitivity analysis in missing data literature \citep{Kolar:2012,Buuren_2018}. 
We generate missing values by the three mechanisms known as missing completely at random (MCAR), missing at random (MAR), and missing not at random (MNAR). Following \cite{Kolar:2012}, every third variable ($j=1,\ldots, \lfloor p/3 \rfloor$) is subject to missing; for MAR case, $\delta^x_{i,3j} = 0$ if $X_{i,3j-2} < \Phi^{-1}(1 - \theta)$ and for MNAR case, $\delta^x_{i,3j} = 0$ if $X_{i,3j} < \Phi^{-1}(1 - \theta)$. Here, we fix  $s=0.05$ and $p=200$.

Table \ref{tab:sensitivity} confirms that a higher rate of missing in data can lead to worse performance. Also, the performance gets poorer as the missing mechanism changes from MCAR to MAR, MNAR, but interestingly, the results on relative performance are not much different.

\begin{table}[H]
	\tiny
	\centering
	\begin{tabular}{|c|c|c|c|c|c|c|}
		\hline & \multicolumn{6}{|c|}{$ \theta=0.9,MAR $}\\
		\cline{2-7}
		& PE & MSE & pAUC & F$_1$ & TP & FP \\
		\hline TL & 1.942 (0.571)& 3.691 (1.107)& 0.949 (0.038)& 0.446 (0.066)& 9.670 (0.551)& 24.670 (7.210) \\ 
		NL & 3.707 (1.101)& 6.107 (1.580)& 0.865 (0.071)& 0.385 (0.079)& 8.460 (1.158)& 26.530 (7.657) \\ 
		CoCo & 3.289 (0.881)& 6.233 (1.663)& 0.830 (0.067)& 0.389 (0.079)& 8.380 (1.170)& 25.750 (7.627) \\ 
		NCL & 4.844 (1.255)& 6.206 (1.675)& 0.542 (0.082)& 0.426 (0.111)& 8.030 (1.298)& 22.850 (16.160) \\ 
		LPD & 3.184 (0.841)& 6.002 (1.552)& 0.869 (0.066)& 0.368 (0.068)& 8.550 (1.132)& 28.780 (6.761) \\ 
		\hline & \multicolumn{6}{|c|}{$ \theta=0.7,MAR $}\\
		\cline{2-7}
		& PE & MSE & pAUC & F$_1$ & TP & FP \\
		\hline TL & 1.941 (0.569)& 3.730 (1.148)& 0.954 (0.034)& 0.430 (0.071)& 9.760 (0.515)& 26.910 (8.568) \\ 
		NL & 9.443 (2.763)& 8.787 (1.483)& 0.727 (0.095)& 0.297 (0.073)& 5.770 (1.517)& 23.730 (7.760) \\ 
		CoCo & 6.090 (1.645)& 11.032 (2.784)& 0.665 (0.086)& 0.306 (0.084)& 5.510 (1.487)& 20.890 (5.597) \\ 
		NCL & 8.309 (15.835)& 10.031 (2.132)& 0.462 (0.080)& 0.320 (0.099)& 5.140 (1.491)& 18.170 (11.910) \\ 
		LPD & 5.191 (1.226)& 9.061 (1.631)& 0.752 (0.090)& 0.285 (0.067)& 6.570 (1.423)& 30.120 (6.181) \\ 
		\hline & \multicolumn{6}{|c|}{$ \theta=0.9,MNAR $}\\
		\cline{2-7}
		& PE & MSE & pAUC & F$_1$ & TP & FP \\
		\hline TL & 1.980 (0.601)& 3.769 (1.187)& 0.951 (0.038)& 0.429 (0.069)& 9.680 (0.566)& 26.620 (8.318) \\ 
		NL & 4.002 (1.048)& 6.672 (1.419)& 0.843 (0.069)& 0.358 (0.070)& 7.950 (1.336)& 27.040 (6.280) \\ 
		CoCo & 3.740 (0.922)& 7.076 (1.738)& 0.808 (0.066)& 0.350 (0.072)& 7.930 (1.273)& 28.570 (8.519) \\ 
		NCL & 5.231 (1.128)& 7.049 (1.458)& 0.582 (0.069)& 0.365 (0.106)& 7.590 (1.342)& 28.800 (19.985) \\ 
		LPD & 3.551 (0.797)& 6.688 (1.455)& 0.843 (0.067)& 0.342 (0.059)& 8.120 (1.225)& 30.030 (6.617) \\ 
		\hline & \multicolumn{6}{|c|}{$ \theta=0.7,MNAR $}\\
		\cline{2-7}
		& PE & MSE & pAUC & F$_1$ & TP & FP \\
		\hline TL & 1.898 (0.512)& 3.625 (1.005)& 0.947 (0.039)& 0.432 (0.063)& 9.670 (0.514)& 26.030 (7.661) \\ 
		NL & 10.300 (3.496)& 9.439 (1.842)& 0.695 (0.084)& 0.285 (0.087)& 5.280 (1.422)& 22.770 (7.777) \\ 
		CoCo & 6.574 (2.109)& 11.897 (3.978)& 0.656 (0.073)& 0.292 (0.089)& 5.200 (1.421)& 20.990 (5.502) \\ 
		NCL & 7.167 (2.308)& 9.909 (2.016)& 0.473 (0.077)& 0.312 (0.099)& 5.000 (1.524)& 18.650 (13.253) \\ 
		LPD & 5.301 (1.144)& 9.334 (1.682)& 0.749 (0.081)& 0.256 (0.064)& 6.210 (1.438)& 33.130 (7.688) \\ 
		\hline
	\end{tabular}
	\caption{\footnotesize Sensitivity analysis for $\theta=0.7, 0.9$ and different missing mechanisms. Each performance measure is averaged over $R=100$ repetitions (standard deviation in parenthesis).}
	\label{tab:sensitivity}
\end{table}

\subsection{Timing}
For both LPD and CoCo, the first step is to modify the estimate of covariance matrix to be PD, and the second step is to solve the penalized regression (e.g. (\ref{eqn:lasso-problem}) for LPD) with the modified estimate. We separately measure the time elapsed for the steps, positive definite modification (PD) and lasso regression (Lasso), which is shown in Table \ref{tab:timing}. We use $\ell_\infty$-norm for LPD since the other norms take roughly the same amount of time. In this experiment, we fix the tuning parameter $\lambda$ at the middle of endpoints of search grids. 

In step ``Lasso'', both methods solve a strictly convex quadratic programming problem, which is very fast. It took less than a second for both methods and does not have much difference between the two methods.
However, in step ``PD'', CoCo takes much longer than LPD, for example, around 50 seconds when $p=1000$ compared to $0.128$ seconds for LPD. Thus, ``PD'' step is dominant in the whole process of CoCo, while it does not scale up the total time of LPD.

%
%
%
%
%


\begin{table}[H]
	\footnotesize
	\centering
	\begin{tabular}{|c|c|c|c|c|}
		\hline
		Method & Step & $p=$200 & $p=$500 & $p=$1000\\
		\hline
		CoCo & Lasso & 0.146 & 0.507 & 0.538\\
		
		CoCo & PD & 0.174 & 3.849 & 49.587\\
		\hline
		LPD & Lasso & 0.103 & 0.382 & 0.515\\
		
		LPD & PD & 0.004 & 0.033 & 0.128\\
		\hline
	\end{tabular}
	\caption{The elapsed times (unit: second) for (1) lasso estimation at a fixed tuning parameter (Lasso) and (2) positive definite modification (PD). We average over 100 independent datasets generated under $n=200$, $s=0.05$, and $p$ varying over $200, 500, 1000$.}
	\label{tab:timing}
\end{table}
%
%

\section{Real data: Genomics of Drug Sensitivity in Cancer} 
\begin{figure}[H]
	\centering
	\includegraphics[page=1,width=0.7\linewidth]{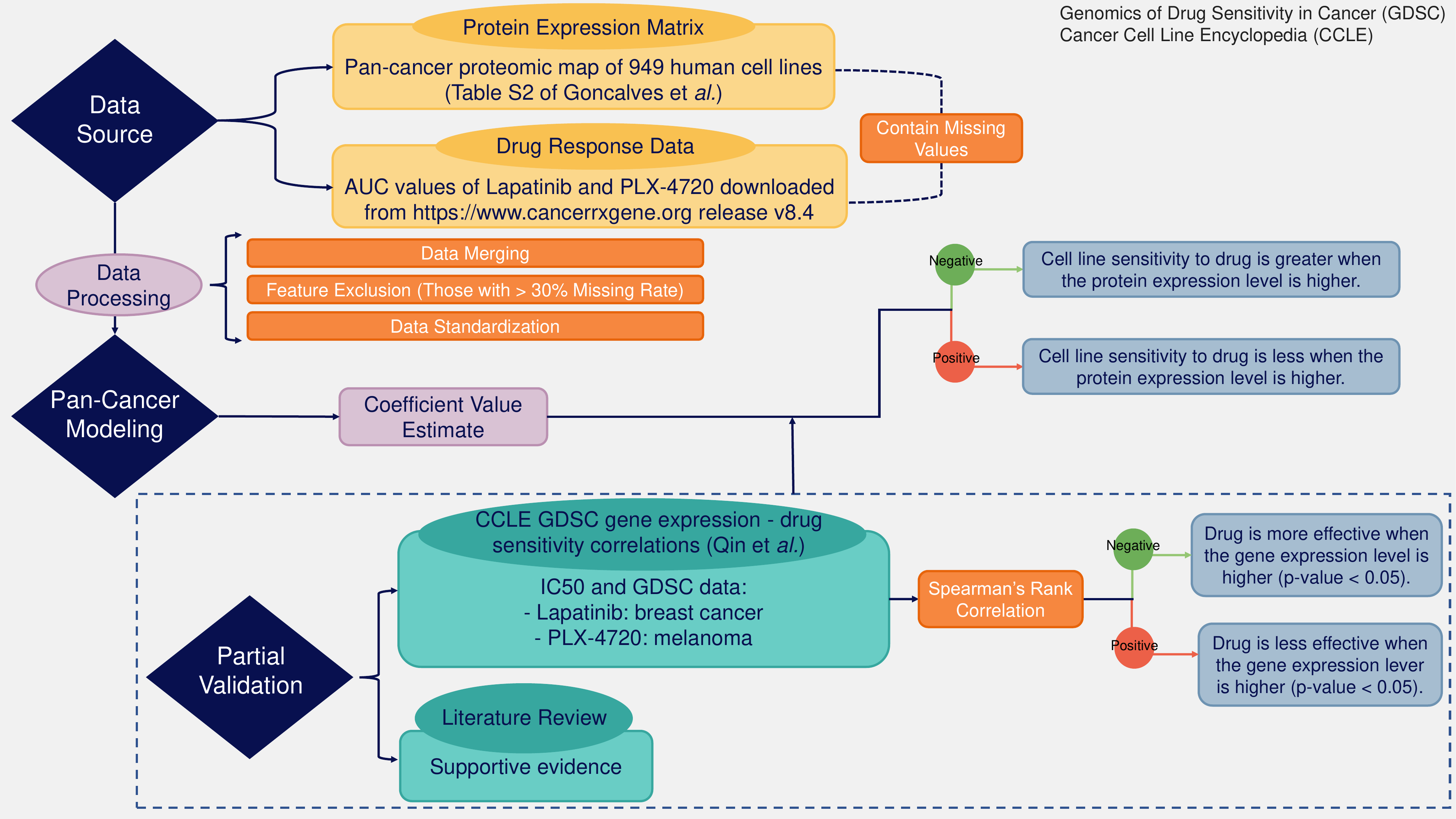}
	\caption{The overview of the pan-cancer drug sensitivity analysis and partial validation.}
	\label{fig:analysis_map}
\end{figure}

In this section, we studied the performance of the proposed method through drug response data available from Genomics of Drug Sensitivity in Cancer (GDSC). In this dataset, cancer cell lines (samples) are treated with different drugs or compounds. Sensitivity to some drugs was measured by the area under the dose–response curve ($\text{AUC}_{\text{RS}}$) (a response variable), which is to be modeled by the protein levels of cells (explanatory variables). 
A small $\text{AUC}_{\text{RS}}$ value indicates a strong drug response of the cell line to the drug. A large value of $\text{AUC}_{\text{RS}}$ means no or limited response of the cell line to the tested drug \citep{Vis:2016}.
Among many, we used the protein expression data from $949$ human cancer cell lines. We aimed to discover a list of (small portion of) proteins (biomarkers) that help explain the drug sensitivity for the anti-cancer drug of interest. These lists may also be used to identify cell lines that respond to some drugs more actively than others.

In the dataset, $949$ cell lines and $8,498$ protein expressions were incompletely measured, but we deleted proteins in which more than $30\%$ of values were missing, resulting in the bottom left of Figure \ref{fig:data_missing}. Then, the final data we used to analyze is $n = 867$ cell lines and $p = 4,183$ proteins. It has $7.15\%$ of  missing values in average across cell lines (see the top of Figure \ref{fig:data_missing}). 
However, every cell line has at least $48$ missing values (see the bottom right of Figure \ref{fig:data_missing}), meaning the listwise deletion is not feasible.

\begin{figure}[H]
	\centering
	\includegraphics[page=1,width=0.8\linewidth]{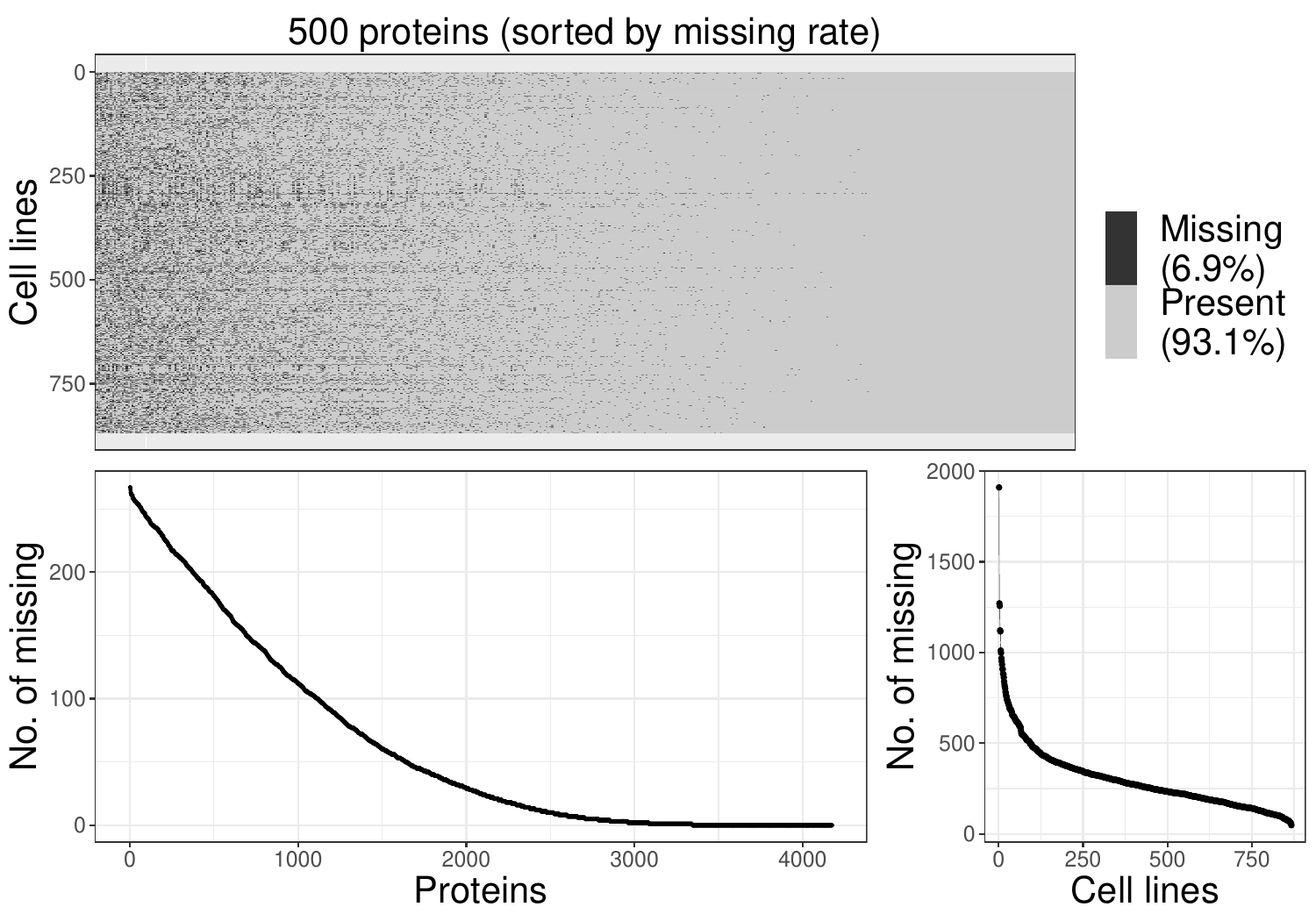}
	\caption{
		In the top figure, missing values are marked as black in the data matrix with randomly chosen 500 proteins.
		The two bottom figures show the number of missing values in either proteins (left) or cell lines (right).}
	\label{fig:data_missing}
\end{figure}

We used Lapatinib (an approved drug in treating HER2-positive breast cancers, an inhibitor of EGFR (also known as ERBB1 and HER1) \citep{Xu:2017} and HER2 (also known as ERBB2)) and PLX-4720 (selective inhibitor of BRAFV600E) as two examples to showcase the application of our method in examining the pan-cancer drug responses and exploring potential protein biomarkers of cancer vulnerabilities. 

Before running our proposed method based on $\ell_\infty$-norm, we standardized $\text{AUC}_{\text{RS}}$ and protein expressions using sample means and standard deviations calculated ignoring missing values. The grid search for the tuning parameter was similarly performed as in the simulation study; the naive lasso estimator $\br_{\text{naive}}$ was fit and used to decide the range of grids $[R/10000, R]$ with $R=2 ||\br_{\text{naive}}||_{\max}$ in which 100 evenly spaced grid points were considered. The cross-validation error curves are given in the left of Figure \ref{fig:data_cverror}.
\begin{figure}[H]
	\centering
	\includegraphics[page=1,width=0.8\linewidth]{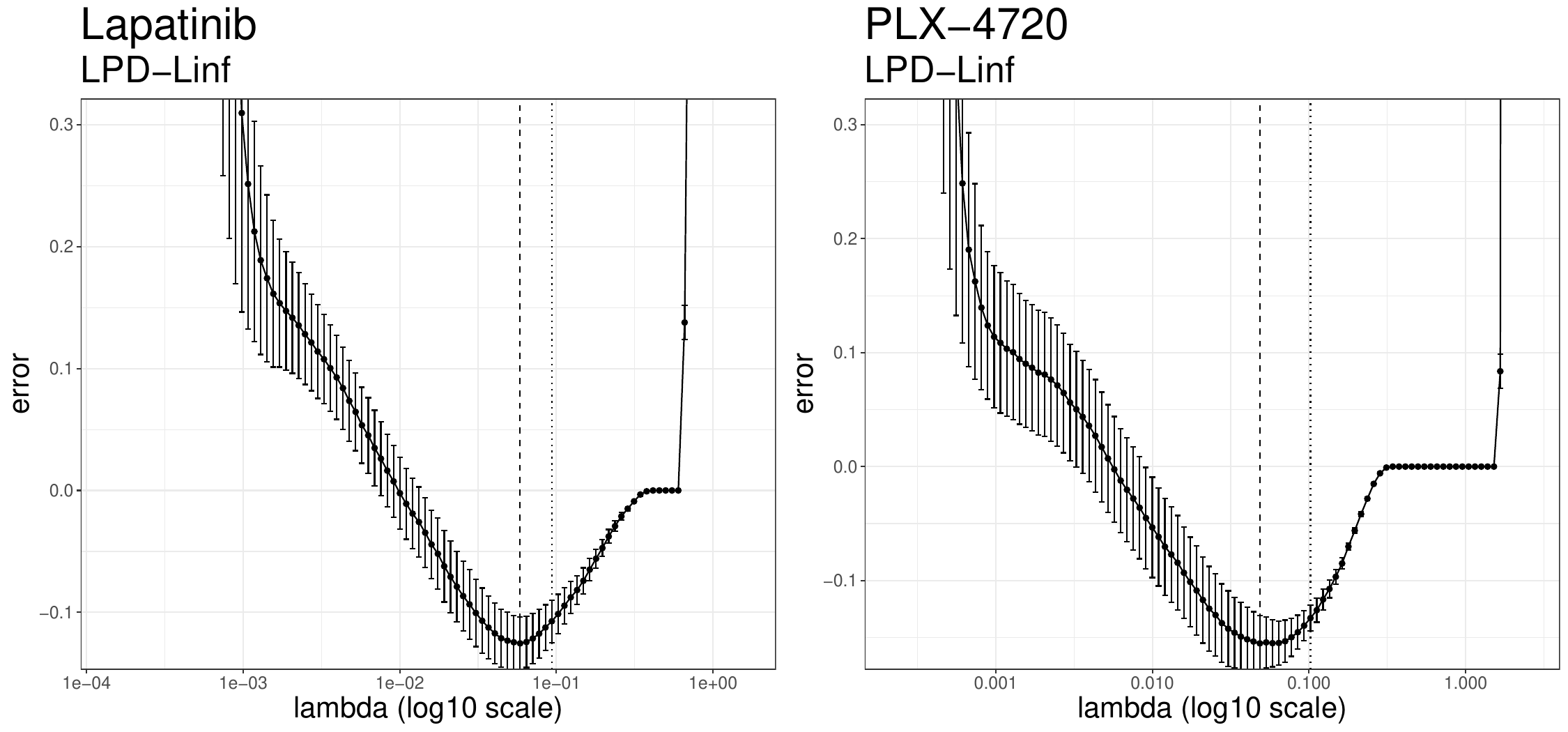}
	\caption{The corrected cross-validation error (solid line). The two vertical lines indicate the optimal tuning parameter (dashed line) and 1-se rule (dotted line), respectively. The error bar is deviated from the center by one standard error.}
	\label{fig:data_cverror}
\end{figure}

We attempted to interpret the estimated coefficients. For simplicity, we applied the 1-se rule (the dotted line in Figure \ref{fig:data_cverror}) that chose a slightly larger tuning parameter and pursued a sparser solution whose accuracy was still acceptable. Table \ref{tab:data_coef_sign} below shows the number of non-zero coefficients and their signs.
\begin{table}[H]
	\footnotesize
	\centering
	\begin{tabular}{ccc}
		\hline
		Drug & Sign & Count \\ 
		\hline
		Lapatinib & (--) &  48 \\ 
		Lapatinib & (+) &  40 \\ 
		Lapatinib & zero & 4088 \\ 
		PLX-4720 & (--) &  58 \\ 
		PLX-4720 & (+) &  29 \\ 
		PLX-4720 & zero & 4089 \\ 
		\hline
	\end{tabular}
	\caption{Signs of the estimated coefficients from the 1-se rule.}
	\label{tab:data_coef_sign}
\end{table}
\noindent
In our analysis, a negative association (coefficient) with $\text{AUC}_{\text{RS}}$ suggests greater sensitivity (of a cell line) when the protein level is high. A tool developed by \cite{Qin:2017} aiming at the discovery of drug sensitivity and gene expression association was used to assist us in demonstrating the robustness of our method. In \cite{Qin:2017}, a positive correlation with the IC50 indicates that the drug is less effective when the expression of a targeted gene is higher and vice versa. However, it is essential to note that the concordance between proteomics and transcriptomics can be weak \citep{Wu:2013}. Integrating the information obtained from each data modality may help predict the effects of gene/protein levels on anti-cancer drug activity \citep{Goncalves:2022}.

For the case of Lapatinib, we found 48 proteins that showed a significant negative association with the $\text{AUC}_{\text{RS}}$. Interestingly, EGFR, the canonical target of Lapatinib, was also found to be among the selected proteins. Among 48 proteins, nine showed concordance with the expression of nine genes (\textit{BAIAP2}, \textit{FAM83H}, \textit{HDHD3}, \textit{HSD17B8}, \textit{KRT19}, \textit{MIEN1}, \textit{PLXNB2}, \textit{REEP6}, and \textit{SEC16A}) affecting the activity of Lapatinib estimated by \cite{Qin:2017} using IC50 and GDSC gene expression data. It has been known that \textit{MIEN1} is amplified along \textit{ERBB2} and exhibits oncogenic potential \citep{Omenn:2014}. It is linked to cisplatin resistance and is highly expressed in Lapatinib-sensitive breast cancer cells than Lapatinib-resistant breast cancer cells \citep{Kumar:2019}.

PLX-4720 has shown \textit{in vitro} and \textit{in vivo} efficacy in treating thyroid cancer and melanoma \citep{Coperchini:2019}. In our analysis, 58 proteins showed a negative association with $\text{AUC}_{\text{RS}}$. 
Regarding thyroid cancer, 8 corresponding genes (\textit{FAHD2A}, \textit{FKBP10}, \textit{GSN}, \textit{QDPR}, \textit{RAB27A}, \textit{RETSAT}, \textit{S100A13}, \textit{TIMM50}) also had negative Spearman's rank correlation coefficient in the analysis by \cite{Qin:2017} (using IC50 and GDSC gene expression data).
Ten out of 12 genes (\textit{AMDHD2}, \textit{CTSB}, \textit{ENDOD1}, \textit{HIBADH}, \textit{KANK2}, \textit{PML}, \textit{RPS27L}, \textit{SP100}, \textit{STX7}, and \textit{TIMMDC1}) showed negative Spearman's rank correlation coefficient in the analysis for melanoma by \cite{Qin:2017}.
These generally concordant results suggest the relevance of our pan-cancer regression modeling approach.

\section{Conclusion} 
This paper tackles the penalized linear regression problem with missing observations where the estimated Gram matrix of covariates is non-PD in general. To handle it, we present a significantly simpler approach for positive definite modification of non-PD matrices inspired by linear shrinkage of covariance matrix. Due to its closed forms, the procedure is scalable even for high-dimensional regression, while the lasso solution based on it still enjoys the same rate of convergence and selection consistency. Through analyzing simulated and real data, we verify that the proposed method has a greater advantage in computational aspect compared to existing methods while ensuring 
theoretical properties such as selection consistency. 

We acknowledged some potential to extend our method to the MAR case by modeling the observation probability $\pi_{i,jk}^{xx}=\pi(\boldsymbol{x}_{i,\text{obs}}; \boldsymbol{\eta})$ using the (fully) observed data. It can be shown that the corresponding IPW estimator is unbiased under the MAR assumption, but its concentration inequalities are more difficult to derive due to the dependency of observed data. This extension is interesting for future work. Moreover, we expressed the estimation performance with the minimum pairwise sample size. \cite{Zheng:2023} came up with measuring individual dependency on missing observations in a different context (estimation of the graphical model). Under suitable assumptions on the graph structure of explanatory variables (e.g. sparsity), representing the individual dependency would give more insights for the regression coefficients. This needs more investigation on the simultaneous estimation of covariance matrix and regression coefficients, and thus we leave it as future work.

As the quadratic loss is closely connected to the Gaussian distribution, a natural extension of our work is to exponential families, i.e. the generalized linear model (GLM). Seemingly, it looks challenging to define a Gram matrix in this context due to the non-linear link function. However, when fitting the genearlized linear model, an adjusted dependent variable is used in the process of an iterative (re-)weighted least squares (\cite{James:2009}). Moreover, one may find that the adjusted dependent variable can be seen as the sum of a linear predictor (evaluated at the current iteration) and the Pearson residual. Based on this observation, we may construct Gram matrices defined between linear predictors and/or Pearson residuals. We plan to explore this extension in future.

{
	To address the sub-optimal convergence rate caused by the trace term in our theories, there might be room for improvement. Currently, we transit the deviation of the smallest eigenvalue of the IPW estimator (see Lemma \ref{lem:alpha_bound}) to the spectral norm using Weyl's inequality; $|\lambda_{\min}(\widehat{\bSigma}^{\rm IPW}) - \lambda_{\min}(\bSigma)| \le || \widehat{\bSigma}^{\rm IPW} - \bSigma ||_2$. However, this inequality may not be tight in a certain class $\tilde{\mathcal{C}}$ of the covariance matrix. If a sharper upper bound of the left-hand side, ideally not depending on the trace term, could be achieved, then the theoretical results could be further improved.
}

\section*{Acknowledgements}

Seongoh Park was supported by the government of the Republic of Korea (MSIT) and the National Research Foundation of Korea (NRF-1711200203); the Sungshin Women’s University Research Grant of H20240073. Johan Lim was supported by the government of the Republic of Korea (MSIT) and the National Research Foundation of Korea (NRF-2021R1A2C1010786)

\bibliographystyle{apalike}
\bibliography{references}

\newpage
\appendix

\section{Non-asymptotic inequality of the IPW estimator in the spectral norm}\label{sec:IPW_spectral}

In this section, we will derive the concentration inequality of the IPW estimator. More specifically, we are interested in the rate of convergence of $||\widehat{\bSigma}^{\rm IPW} - \bSigma ||_2$. Recall the definition of the IPW estimator 
$$
\widehat{\bSigma}^{\rm IPW} = \bS * \left[\dfrac{1}{\pi^{xx}_{jk}}, 1 \le j,k\le p \right],
$$
which is given in (\ref{eq:IPW_correction}). The random variables $\bx_i$, $(\delta_{i1}^x, \ldots, \delta_{ip}^x)$ used above are assumed to satisfy Assumption \ref{assm:sub_G}, \ref{assm:multi_bern}, and \ref{assm:mcar}. For notational convenience, we write the IPW estimator by $\widehat{\bSigma}$. Also, we omit the superscript in $\delta_{ij}^x$, $\pi_{ij\cdots}^{xx}$ and $K^x$.

\begin{thm}\label{thm:IPW_spectral}
	For $t > 1 \vee \log n$, it holds with probability at least $1-3e^{-t}$ that 
	$$
	||\widehat{\bSigma} - \bSigma ||_2 \le 
	C \tr(\bSigma) \max\{K^2, 1\} \max\bigg\{
	\sqrt{\dfrac{\pi_{\max}^{(4)} (t + \log p)}{n}},
	(t + \log n) \dfrac{\pi_{\max}^{(4)}(t + \log p)}{n}
	\bigg\},
	$$
	where $C>0$ is some numerical constant and 
	$$
	\pi_{\max}^{(4)} = \max_{k_1, k_2,\ell_1, \ell_2}\dfrac{\pi_{k_1 k_2 \ell_1 \ell_2}}{\pi_{k_1 \ell_1} \pi_{k_2\ell_2}}.
	$$
\end{thm}
\noindent
Our proof is based on the idea of \cite{Lounici:2014}, but improve it to address the general missing dependency.

We begin with the following decomposition:
$$
||\widehat{\bSigma} - \bSigma ||_2 \le ||\diag(\widehat{\bSigma} - \bSigma) ||_2 + ||\text{OD}(\widehat{\bSigma} - \bSigma) ||_2
$$
where $\diag(\bA)$ is a diagonal matrix with diagonals inherited from $\bA$, and $\text{OD}(\bA) = \bA - \diag(\bA)$. We deal with each of them separately. 

\subsection{Off-diagonal part}
To use Bernstein inequality of bounded matrices later, we consider an event $A_i = \{||X_i||_2^2 \le U\}$ where $U = C \cdot \text{tr}(\bSigma) ( K^2 + 1 ) (t + \log n)$ for some numerical constant $C>0$. We claim the following: \begin{fact}\label{fact:boundedness}
	$\text{P}(\cap_{i=1}^n A_i) \ge 1 - e^{-t}$ for any $t > 0$.
\end{fact}
\noindent
Define a matrix $Z_i$ with zero diagonals
$$
Z_i = \text{OD}\left(\left[\dfrac{\tilde{X}_{ik} \tilde{X}_{i\ell}}{\pi_{k\ell}}\right]_{1 \le k, \ell \le p} \right),
$$
and $\tilde{Z}_i = Z_i \text{I}_{A_i}$.
On the event $\cap_{i=1}^n A_i$, we can get $\text{OD}(\widehat{\bSigma} - \bSigma) = \dfrac{1}{n} \sum\limits_{i=1}^n (Z_i - \mathbb{E} Z_i) = \dfrac{1}{n} \sum\limits_{i=1}^n (\tilde{Z}_i - \mathbb{E} \tilde{Z}_i) - \dfrac{1}{n} \sum\limits_{i=1}^n \mathbb{E} Z_i \text{I}_{A_i^c}$ and thus
\begin{equation}\label{eq:OD_norm}
||\text{OD}(\widehat{\bSigma} - \bSigma)||_2 \le || \dfrac{1}{n} \sum\limits_{i=1}^n (\tilde{Z}_i - \mathbb{E} \tilde{Z}_i) ||_2 + || \dfrac{1}{n} \sum\limits_{i=1}^n \mathbb{E} Z_i \text{I}_{A_i^c} ||_2.
\end{equation}
For the latter term, we get
\begin{equation}\label{eq:OD_second}
\begin{array}{rcl}
|| \dfrac{1}{n} \sum\limits_{i=1}^n \mathbb{E} Z_i \text{I}_{A_i^c} ||_2 & = & 	|| \mathbb{E} Z_1 \text{I}_{A_1^c} ||_2\\
& = & \max\limits_{\theta \in \mathcal{S}_{p-1}} |\mathbb{E} \theta^\top Z_1 \theta \text{I}_{A_1^c} |\\
& \le & \max\limits_{\theta \in \mathcal{S}_{p-1}} \mathbb{E} |\theta^\top Z_1 \theta| \text{I}_{A_1^c} \\
& \le & \max\limits_{\theta \in \mathcal{S}_{p-1}} \sqrt{\mathbb{E} (\theta^\top Z_1 \theta)^2 \mathbb{E} \text{I}_{A_1^c}}\\
& = & \sqrt{\max\limits_{\theta \in \mathcal{S}_{p-1}} \mathbb{E} (\theta^\top Z_1 \theta)^2 \cdot  \text{P}(A_1^c)} \equiv t_2
\end{array}
\end{equation}
Next, note that $\tilde{Z}_1 - \mathbb{E} \tilde{Z}_1$ is bounded conditioning on the set $A$, which is stated and proved more specifically in (F1) of Fact \ref{fact:const}.
Hence, we can use Bernstein inequality for the former, and get the upper bound of $|| \dfrac{1}{n} \sum\limits_{i=1}^n (\tilde{Z}_i - \mathbb{E} \tilde{Z}_i) ||_2$. The following result is from Proposition 2 of \cite{Lounici:2014}.
For $t > 0$, with probability at least $1 - e^{-t}$, we have (conditioning on the set $A$)
\begin{equation}\label{eq:OD_first}
|| \dfrac{1}{n} \sum\limits_{i=1}^n (\tilde{Z}_i - \mathbb{E} \tilde{Z}_i) ||_2 \le 
2\max\left\{
\sigma_{\tilde{Z}} \sqrt{\dfrac{t + \log p}{n}}, 2 \pi_{\max}^{(2)} U \dfrac{t + \log p}{n}
\right\} \equiv t_1,
\end{equation}
where $\sigma^2_{\tilde{Z}}= || \dfrac{1}{n} \sum\limits_{i=1}^n \mathbb{E} (\tilde{Z}_i - \mathbb{E} \tilde{Z}_i)^2 ||_2 = || \mathbb{E}(\tilde{Z}_1 - \mathbb{E} \tilde{Z}_1)^2 ||_2$.

Combining (\ref{eq:OD_norm}), (\ref{eq:OD_second}), and (\ref{eq:OD_first}), we have 
$$
\begin{array}{rcl}
\text{P}(||\text{OD}(\widehat{\bSigma} - \bSigma)||_2 > t_1 + t_2) 
& \le & \text{P}(||\text{OD}(\widehat{\bSigma} - \bSigma)||_2 > t_1 + t_2 | A)  + 
\text{P}(A^c) \\
& \le & \text{P}(|| \dfrac{1}{n} \sum\limits_{i=1}^n (\tilde{Z}_i - \mathbb{E} \tilde{Z}_i) ||_2 \\
&& \quad + || \dfrac{1}{n} \sum\limits_{i=1}^n \mathbb{E} Z_i \text{I}_{A_i^c} ||_2 > t_1 + t_2 | A)  + 
\text{P}(A^c) \\
& \le & \text{P}(|| \dfrac{1}{n} \sum\limits_{i=1}^n (\tilde{Z}_i - \mathbb{E} \tilde{Z}_i) ||_2 > t_1 | A)  + 
\text{P}(A^c) \\
& \le & 2e^{-t}.
\end{array}
$$
The remaining part is to prove the boundedness of $\tilde{Z}_i - \mathbb{E} \tilde{Z}_i$ and calculate constants appearing in $t_1$ and $t_2$. 
\begin{fact}\label{fact:const}
	The following statements hold in deterministic sense.
	\begin{enumerate}
		\item[(F1)] Conditioning on the set $A=\cap_{i=1}^n \{||X_i||_2^2 \le U\}$, we get
		$$
		|| \tilde{Z}_1 - \mathbb{E} \tilde{Z}_1 ||_2 \le 2 \pi_{\max}^{(2)} U,
		$$
		where $\pi_{\max}^{(2)} = \max_{k,\ell} 1/\pi_{k\ell}$.

		\item[(F2)] $\max\limits_{\theta \in \mathcal{S}_{p-1}} \mathbb{E} (\theta^\top Z_1 \theta)^2 \le  C K^4  \pi_{\max}^{(4)} (\tr(\bSigma))^2$ where
		$$
		\pi_{\max}^{(4)} = \max_{k_1, k_2,\ell_1, \ell_2}\dfrac{\pi_{k_1 k_2 \ell_1 \ell_2}}{\pi_{k_1 \ell_1} \pi_{k_2\ell_2}}
		$$
		
		\item[(F3)] $\sigma^2_{\tilde{Z}}= || \mathbb{E}(\tilde{Z}_1 - \mathbb{E} \tilde{Z}_1)^2 ||_2 \le
		C K^4 \pi_{\max}^{(3)} (\tr(\bSigma))^2$
		where
		$$
		\pi_{\max}^{(3)}=\max_{s,k,\ell} \dfrac{\pi_{k\ell s}}{\pi_{ks}\pi_{\ell s}}
		$$
	\end{enumerate}
	
\end{fact}
\noindent 
One can easily check that $\pi_{\max}^{(4)}\ge \max\{\pi_{\max}^{(2)}, \pi_{\max}^{(3)}\}$. Thus, some calculations lead to 
$$
t_1+t_2  \le  C \tr(\bSigma) \max\{K^2, 1\} \max\bigg\{
\sqrt{\dfrac{\pi_{\max}^{(4)} (t + \log p)}{n}},
(t + \log n) \dfrac{\pi_{\max}^{(4)}(t + \log p)}{n}
\bigg\},
$$
for some $C > 0$ if $t > 1 \vee \log n$.

\subsection{Diagonal part}
Remark that the Orlicz norm used in \cite{Lounici:2014} and $\psi_2$-norm in this paper are equivalent, up to a constant factor. Moreover, they both satisfies
$$
||\tilde{X}_{ik}||_{\psi_2} \le ||X_{ik}||_{\psi_2}, \quad ||\tilde{X}_{ik}^2||_{\psi_1} \le 2 ||\tilde{X}_{ik}||_{\psi_2}^2.
$$
Using these facts,  we get
$$
||\tilde{X}_{ik}^2||_{\psi_1} \le 2 ||\tilde{X}_{ik}||_{\psi_2}^2 \le 2||X_{ik}||_{\psi_2}^2 \le 2 \sigma_{kk} K^2.
$$
By Proposition 1 of \cite{Lounici:2014}, we get with probability at least $1 - e^{-t}$
$$
\Big| \dfrac{\sum_{i=1}^n \tilde{X}_{ik}^2}{n\pi_{k}} - \Sigma_{kk} \Big| \le \dfrac{C\sigma_{kk} K^2}{\pi_{k}} \big( \sqrt{\dfrac{t}{n}} \vee \dfrac{t}{n}\big).
$$
This implies that with probability at most $p e^{-t}$
$$
\max_{k}\Big| \dfrac{\sum_{i=1}^n \tilde{X}_{ik}^2}{n\pi_{k}} - \Sigma_{kk} \Big| > C K^2 \max_{k} \dfrac{\sigma_{kk}}{\pi_{k}}  \big( \sqrt{\dfrac{t}{n}} \vee \dfrac{t}{n}\big)
$$
Putting $t \leftarrow t + \log p$, we get

$$
\text{P}\left[ ||\diag(\widehat{\bSigma} - \bSigma) ||_2 > 
C K^2 \max_{k} \dfrac{\sigma_{kk}}{\pi_{k}}  \bigg\{ \sqrt{\dfrac{t + \log p}{n}}, \dfrac{t + \log p}{n} \bigg\}
\right] \le e^{-t}
$$

\subsection{Proof of Fact \ref{fact:boundedness}}
\begin{proof}
	$||X_i||_2^2 - \mathbb{E} ||X_i||_2^2$ is sub-exponential satisfying its $\psi_2$-norm bounded by
	$$
	\begin{array}{rcl}
	\Big|\Big| ||X_i||_2^2 - \mathbb{E} ||X_i||_2^2\Big|\Big|_{\psi_2} 
	& \le & \sum_{j=1}^p || X_{ij}^2 ||_{\psi_2}  + \text{tr}(\bSigma) \\[0.5em]
	&\le& 
	\sum_{j=1}^p 2 \sigma_{jj} K^2 + \text{tr}(\bSigma) \\[0.5em]
	& = &
	\text{tr}(\bSigma)(2 K^2 + 1)
	\end{array}
	$$
	By Proposition 1 of \cite{Lounici:2014},
	$$
	\text{P}\Big[
	||X_i||_2^2 > \text{tr}(\bSigma)\big\{1 + C(2 K^2 + 1) (\sqrt{t} \vee t) \big\}
	\Big] \le e^{-t}, \quad t>0.
	$$
	Putting $t \leftarrow t + \log n$ for $n>2$, we get
	$$
	\text{P}\Big[
	||X_i||_2^2 > \text{tr}(\bSigma)\big\{1 + C(2 K^2 + 1) (t + \log n) \big\} 
	\Big] \le e^{-t} / n , \quad t>0.
	$$
	Note that we can find another constant $C'>0$ such that $\text{tr}(\bSigma)\big\{1 + C(2 K^2 + 1) (t + \log n)\big\} \le C' \cdot \text{tr}(\bSigma)  ( K^2 + 1 ) (t + \log n) \equiv U$.
	By the union argument, we conclude $\text{P}\Big[\cup_{i=1}^n A_i \Big] \le e^{-t}$, for $t>0$.
\end{proof}

\subsection{Proof of (F1) of Fact \ref{fact:const}}
\begin{proof}	
	Define $V_1 = \left[\dfrac{Y_{1k} Y_{1\ell}}{\pi_{k\ell}}\right]_{1 \le k, \ell \le p}$ and $W_1 = \diag(V_1)$, and thus $Z_1 = V_1 - W_1$ holds.
	Since $V_1 - Z_1 = W_1 \succcurlyeq 0$, we begin with
	\begin{equation}\label{eq:Z_1_spectral}
	\begin{array}{rcl}
	|| Z_1 ||_2 & \le & || V_1||_2 \\
	& = & \max\limits_{\theta \in \mathcal{S}_{p-1}} \left| \sum\limits_{k, \ell} \dfrac{Y_{1k}Y_{1\ell} \theta_k \theta_\ell }{\pi_{k\ell}} \cdot \text{I}_{A_1} \right| \\[1em]
	& \le & \max\limits_{\theta \in \mathcal{S}_{p-1}} \sqrt{ \sum\limits_{k, \ell} \dfrac{Y_{1k}^2Y_{1\ell}^2}{\pi_{k\ell}^2} 
		\sum\limits_{k, \ell}\theta_k^2 \theta_\ell^2} \\[1em]
	& \le & \max\limits_{\theta \in \mathcal{S}_{p-1}} \pi_{\max}^{(2)}\sqrt{ \sum\limits_{k, \ell} Y_{1k}^2Y_{1\ell}^2\sum\limits_{k, \ell} \theta_{k}^2\theta_{\ell}^2}  \\[1em]
	& = &  \pi_{\max}^{(2)} ||Y_1||_2^2
	\end{array}
	\end{equation}
	where we used the Cauchy-Schwartz inequality and $\pi_{\max}^{(2)} = \max_{k,\ell} 1/\pi_{k\ell}$.
	Moreover, we know that
	$$
	||Y_1||_2^2 \le ||X_1||_2^2 \le U,
	$$
	where the last inequality holds conditional on the event $A$.
	Combining these with (\ref{eq:Z_1_spectral}), we can get $|| \tilde{Z}_1 ||_2 \le \pi_{\max}^{(2)} U$. Then, since $|| \mathbb{E} \tilde{Z}_1||_2 \le \mathbb{E} ||\tilde{Z}_1||_2 \le \mathbb{E} ||Z_1||_2$, we get
	$$
	|| \tilde{Z}_1 - \mathbb{E} \tilde{Z}_1 ||_2 \le 
	|| \tilde{Z}_1||_2 +  ||\mathbb{E}\tilde{Z}_1||_2 \le 
	|| Z_1||_2 + \mathbb{E} ||Z_1 ||_2  \le 2 \pi_{\max}^{(2)} U
	$$
\end{proof}

\subsection{Proof of (F2) of Fact \ref{fact:const}}
\begin{proof}
	We can get
	$$
	\begin{array}{l}
	\mathbb{E} (\theta^\top Z_1 \theta)^2 = \mathbb{E} \left(\sum\limits_{1 \le k \neq \ell \le p} \dfrac{Y_{1k} Y_{1\ell} \theta_k \theta_\ell}{\pi_{k\ell}} \right)^2 \\[1em]
	= \mathbb{E} \sum\limits_{(k_1, k_2) \neq (\ell_1, \ell_2)} \dfrac{Y_{1k_1} Y_{1\ell_1} \theta_{k_1} \theta_{\ell_1}}{\pi_{k_1 \ell_1}} \dfrac{Y_{1k_2} Y_{1\ell_2} \theta_{k_2} \theta_{\ell_2}}{\pi_{k_2\ell_2}} \\[1em]
	= \mathbb{E} \sum\limits_{k_1, k_2,\ell_1, \ell_2} \dfrac{Y_{1k_1} Y_{1\ell_1} \theta_{k_1} \theta_{\ell_1}}{\pi_{k_1 \ell_1}} \dfrac{Y_{1k_2} Y_{1\ell_2} \theta_{k_2} \theta_{\ell_2}}{\pi_{k_2\ell_2}} - 
	\mathbb{E} \sum\limits_{k_1, k_2}
	\dfrac{Y_{1k_1}^2 Y_{1 k_2}^2 \theta_{k_1}^2 \theta_{k_2}^2}{\pi_{k_1} \pi_{k_2}}\\[1em]
	{\le}  \sum\limits_{k_1, k_2,\ell_1, \ell_2} \dfrac{\pi_{k_1 k_2 \ell_1 \ell_2}}{\pi_{k_1 \ell_1} \pi_{k_2\ell_2}}  
	\mathbb{E} (X_{1k_1} X_{1k_2} X_{1\ell_1}  X_{1\ell_2})
	\theta_{k_1} \theta_{k_2} \theta_{\ell_1} \theta_{\ell_2}\\[1em]
	{\le}  \sqrt{\sum\limits_{k_1, k_2,\ell_1, \ell_2} \left(\dfrac{\pi_{k_1 k_2 \ell_1 \ell_2}}{\pi_{k_1 \ell_1} \pi_{k_2\ell_2}}\right)^2  
		\left(\mathbb{E} X_{1k_1} X_{1k_2} X_{1\ell_1}  X_{1\ell_2}\right)^2
		\sum\limits_{k_1, k_2,\ell_1, \ell_2}\theta_{k_1}^2 \theta_{k_2}^2 \theta_{\ell_1}^2 \theta_{\ell_2}^2}\\[1em]
	{\le}  \pi_{\max}^{(4)}\sqrt{
		\sum\limits_{k_1, k_2,\ell_1, \ell_2}\left(\mathbb{E} X_{1k_1} X_{1k_2} X_{1\ell_1}  X_{1\ell_2}\right)^2},\\[1em]
	\end{array}
	$$
	where we used Cauchy-Schwartz inequality in the second inequality. In the third inequality, we define $\pi_{\max}^{(4)} = \max_{k_1, k_2,\ell_1, \ell_2}\dfrac{\pi_{k_1 k_2 \ell_1 \ell_2}}{\pi_{k_1 \ell_1} \pi_{k_2\ell_2}}$.
	Applying Cauchy-Schwartz inequality twice, we get
	$$
	\mathbb{E} X_{1k_1} X_{1k_2} X_{1\ell_1}  X_{1\ell_2}
	\le
	\sqrt{\mathbb{E} X_{1k_1}^2 X_{1k_2}^2 \mathbb{E}X_{1\ell_1}^2  X_{1\ell_2}^2}
	\le
	\left(\mathbb{E} X_{1k_1}^4 \mathbb{E} X_{1k_2}^4 \mathbb{E}X_{1\ell_1}^4 \mathbb{E} X_{1\ell_2}^4\right)^{1/4}.
	$$
	Thus, we get for any $\theta \in \mathcal{S}^{p-1}$
	$$
	\mathbb{E} (\theta^\top Z_1 \theta)^2 \le \pi_{\max}^{(4)} \left(\sum\limits_{k}\sqrt{\mathbb{E} X_{1k}^4}\right)^2.
	$$
	Finally, using equation (2.1) in \cite{Lounici:2014}, we get
	\begin{equation}\label{eq:moment_4th}
	\mathbb{E}X_{1k}^4 \le C || X_{1k} ||_{\psi_2}^4 \le C K^4 \sigma_{kk}^2,
	\end{equation}
	which concludes the proof.
	
\end{proof}

\subsection{Proof of (F3) of Fact \ref{fact:const}}
\begin{proof}
	
	We observe that
	$$
	|| \mathbb{E}(\tilde{Z}_1 - \mathbb{E} \tilde{Z}_1)^2 ||_2 
	\le 
	|| \mathbb{E}(\tilde{Z}_1)^2||_2
	$$
	since $\mathbb{E}(\tilde{Z}_1)^2 - \mathbb{E}(\tilde{Z}_1 - \mathbb{E} \tilde{Z}_1)^2 = (\mathbb{E} \tilde{Z}_1)^2 \succcurlyeq 0$. Moreover, we get
	$||\mathbb{E} (\tilde{Z}_1)^2||_2 =  \max\limits_{\theta \in \mathcal{S}_{p-1}} \theta^\top \mathbb{E} (Z_1)^2 \theta \text{I}_{A_1}  = || \mathbb{E} (Z_1)^2 ||_2 $.
	
	Also, recall the relationship $Z_1 = V_1 - W_1$, which implies with the triangular inequality that
	$|| \mathbb{E} (Z_1)^2 ||_2 = || \mathbb{E} V_1^2  + \mathbb{E} W_1^2 - \mathbb{E} V_1 W_1 - \mathbb{E} W_1 V_1  ||_2 \le || \mathbb{E} V_1^2 ||_2  + || \mathbb{E} W_1^2 ||_2 + 2 || \mathbb{E} V_1 W_1 ||_2$. Note that
	$$
	\begin{array}{rcl}
	|| \mathbb{E} V_1 W_1 ||_2 &=& \max\limits_{\theta \in \mathcal{S}_{p-1}} | \mathbb{E}\theta^\top  V_1 W_1 \theta | \\
	&\le& 
	\max\limits_{\theta \in \mathcal{S}_{p-1}}  \sqrt{\mathbb{E}(\theta^\top  V_1^2  \theta) \mathbb{E}(\theta^\top W_1^2 \theta)} \\
	&\le& \sqrt{|| \mathbb{E} V_1^2||_2 || \mathbb{E}W_1^2 ||_2}.
	\end{array}
	$$
	Therefore, we get $|| \mathbb{E} (Z_1)^2 ||_2 \le \left(\sqrt{|| \mathbb{E} V_1^2||_2} + \sqrt{|| \mathbb{E}W_1^2 ||_2} \right)^2$. We now calculate the last two terms.
	
	First, we calculate $|| \mathbb{E}W_1^2 ||_2$.
	$$
	||\mathbb{E} (W_1)^2||_2   = 
	\sum\limits_{k} \mathbb{E}Y_{1k}^4 \theta_k^2 / \pi_k^2 
	= 
	\sum\limits_{k} \mathbb{E}X_{1k}^4 \theta_k^2 / \pi_k 
	= 
	\max\limits_{k} \mathbb{E}X_{1k}^4 / \pi_k.
	$$
	Secondly, we compute $||\mathbb{E} (V_1)^2||_2$.
	$$
	\begin{array}{rcl}
	||\mathbb{E} (V_1)^2||_2  & = &
	\max\limits_{\theta \in \mathcal{S}_{p-1}} \sum\limits_{k, \ell, s} \dfrac{\mathbb{E} Y_{1k}Y_{1\ell}Y_{1s}^2 }{\pi_{ks}\pi_{\ell s}} \theta_k \theta_\ell\\ 
	& = &
	\max\limits_{\theta \in \mathcal{S}_{p-1}}   \sum\limits_{s}  \sum\limits_{k, \ell} \dfrac{\pi_{k\ell s}}{\pi_{ks}\pi_{\ell s}} \mathbb{E}X_{1s}^2X_{1k} X_{1\ell}\theta_k  \theta_\ell\\
	& \le &
	\max\limits_{\theta \in \mathcal{S}_{p-1}} \sum\limits_{s}
	\sqrt{\sum\limits_{k, \ell} \left(\mathbb{E} \dfrac{\pi_{k\ell s}}{\pi_{ks}\pi_{\ell s}} X_{1s}^2 X_{1k} X_{1\ell}\right)^2 \sum\limits_{k, \ell}\theta_k^2 \theta_\ell^2 }\\
	& = &
	\pi_{\max}^{(3)}\sum\limits_{s}
	\sqrt{\sum\limits_{k, \ell} \left(\mathbb{E}X_{1s}^2 X_{1k} X_{1\ell}\right)^2 }\\
	\end{array}
	$$
	where we used Cauchy-Schwartz inequality and $\pi_{\max}^{(3)}=\max_{s,k,\ell} \dfrac{\pi_{k\ell s}}{\pi_{ks}\pi_{\ell s}}$. Due to 
	$$
	\mathbb{E}X_{1s}^2 X_{1k} X_{1\ell} \le \sqrt{\mathbb{E}X_{1s}^4 \mathbb{E}X_{1k}^2 X_{1\ell}^2}
	\le 
	\sqrt{\mathbb{E}X_{1s}^4 \sqrt{\mathbb{E}X_{1k}^4 \mathbb{E} X_{1\ell}^4}},
	$$
	we conclude that 
	$$
	||\mathbb{E} (V_1)^2||_2 \le \pi_{\max}^{(3)}\left(\sum\limits_{k} \sqrt{\mathbb{E}X_{1k}^4}\right)^2.
	$$
	Finally, combining all of these with equation (\ref{eq:moment_4th}), we get
	$$
	\begin{array}{l}
	|| \mathbb{E}(\tilde{Z}_1 - \mathbb{E} \tilde{Z}_1)^2 ||_2 \le
	\left(\sqrt{\pi_{\max}^{(3)}}\sum\limits_{k} \sqrt{\mathbb{E}X_{1k}^4} + \sqrt{\max\limits_{k} \mathbb{E}X_{1k}^4 / \pi_k}
	\right)^2 \\
	\qquad \qquad \qquad \qquad \qquad \qquad \qquad \qquad  \le 
	C K^4 \left(\sqrt{\pi_{\max}^{(3)}}\tr(\bSigma) + \sqrt{\max\limits_{k} \sigma_{kk}^2 / \pi_k}
	\right)^2.
	\end{array}
	$$
	which concludes the proof because $\max_k 1/\pi_k \le \pi_{\max}^{(3)}$ and $\max\limits_{k} \sigma_{kk} \le \tr(\bSigma)$.
\end{proof}

\section{Miscellaneous results}\label{sec:misc}

Without the loss of generality, assume that variables in $\mathcal{A}$ come before those in $\mathcal{A}^c$, or we rearrange them to do so.
In all the following proofs, we denote block matrices of $\bA$ decomposed by the subset $\mathcal{A}$ by $\bA_{\mathcal{A}\mathcal{A}}, \bA_{\mathcal{A}\mathcal{A}^c}, \bA_{\mathcal{A}^c\mathcal{A}}, \bA_{\mathcal{A}^c\mathcal{A}^c}$, respectively.

\subsection{Proof of Proposition \ref{prop:consistency}}\label{sec:proof_Lee_result}

Let us review the three conditions used in Theorem 3.4 of \citet{Lee:2015} and apply them to our problem in \eqref{main-eqn:lasso-problem}.

\subsubsection{RSC condiction}
The first condition is the restricted strong convexity (RSC). 
\begin{assumption}[RSC]
	Let $C \subset \mathbb{R}^p$ be some known convex set containing $\btheta^*$. The loss function $\ell$ is RSC when $\exists m, L>0$ such that
	\begin{align*}
	&(1) \ \bt^T \nabla^2 \ell(\btheta) \bt \ge m \bt^T \bt, \quad \forall \btheta \in C \cap M, \quad \forall \bt\in C\cap M - C\cap M \\
	&(2) \ \|\nabla^2 \ell(\btheta) - \nabla^2 \ell(\btheta^*) \| _2 \le L\|\btheta - \btheta^*\|_2, \quad \forall \btheta \in C
	\end{align*}
\end{assumption}
\noindent
The RSC condition is a relaxed version of strong convexity, which is a commonly used assumption for guaranteeing the properties of given loss functions.

In our specified problem, $\nabla^2\ell(\btheta) = \widehat{\bSigma}^{\rm LPD}$. Thus, the RSC condition (2) is satisfied with $L$ with any positive value.
Moreover, for $\ell_1$-norm, the model space is $M = \{\btheta \in \mathbb{R}^p : \btheta_{\mathcal{A}^c} = 0\}$ where $\mathcal{A} \subset [p]$ is the support of the true parameter. We note that 
$$
\min_{\bt\in \mathbb{R}^p:\|\bt\|_2=1, \bt_{\mathcal{A}^c}=0} \bt^{\top} \widehat{\bSigma}^{\rm LPD} \bt =  
\alpha \lambda_{\min}(\widehat{\bSigma}^{\text{IPW}}_{\mathcal{A}\mathcal{A}}) + \mu (1 - \alpha) 
\ge \min\{\lambda_{\min}(\widehat{\bSigma}^{\text{IPW}}_{\mathcal{A}\mathcal{A}}), \mu \}.
$$
Using Weyl's inequality, 
$||\widehat{\bSigma}^{\text{IPW}}_{\mathcal{A}\mathcal{A}} - \bSigma_{\mathcal{A}\mathcal{A}}||_2 \le 0.5 \lambda_{\min}(\bSigma_{\mathcal{A}\mathcal{A}}) $ implies that $\lambda_{\min}(\widehat{\bSigma}^{\text{IPW}}_{\mathcal{A}\mathcal{A}}) \ge 0.5 \lambda_{\min}(\bSigma_{\mathcal{A}\mathcal{A}})$. Now, we set $m=\min\{0.5\lambda_{\min}(\bSigma_{\mathcal{A}\mathcal{A}}), \mu \}$.

\subsubsection{RE condition} 

The second condition is the irrepresentibility (IR) condition. Let us define a few notions to introduce IR 
condition.  The {\it support function} on a convex subset $C \subset \mathbb{R}^p$ is defined as:
$$
h_C(\bx) = \sup \{ {\bx}^{\top} {\by} : {\by} \in C\}.
$$
We say the penalty function $\rho$ is {\it  geometrically decomposable} in terms of $D,I,E \subset \mathbb{R}^p$ if it is decomposed as a sum of support functions:
$$ 
\rho( {\btheta} ) = h_D( {\btheta} ) + h_I( {\btheta} ) + h_{E^\perp}( {\btheta} ),  
$$ 
where $D$ is a convex bounded set, $I$ is a convex bounded set which contains a relative neighborhood of the origin (i.e. $0\in \text{relint}(E)$) and $E$ is a subspace. Now, we can define our second condition, IR condition. 
\begin{assumption}[IR]
	$\exists \tau \in (0,1)$ such that
	$$
	\sup_{\bz\in \partial h_D(M)} V\Big[ \bP_{M^\perp} \{ \bQ \bP_M( \bP_M \bQ \bP_M)^{\dagger} \bP_M \bz - \bz \} \Big] \le 1 - \tau
	$$
\end{assumption}
\noindent
where $\bQ = \nabla^2 \ell({\btheta}^*) = \widehat{\bSigma}^{\rm LPD}$, $\bP_B$ is the projection matrix to $B$, 
\begin{align*}
&\partial h_D(M) = \bigcup_{ {\bu}\in M} \partial h_D( {\bu}) \\
&\gamma_C( {\bx}) = \inf\{ \lambda : {\bx} \in \lambda C\}\\
&V( {\bu} ) = \inf \{\gamma_I( {\by}) + \textbf{1}_{E^\perp}({\bu}- {\by}) \} = \inf_{t\in E^\perp} \gamma_I ( {\bu}- {\bt}),
\end{align*}
%

We can easily check that $\rho$ is geometrically decomposed with the terms of 
\begin{align*}
&E = \mathbb{R}^p \\
&D = \{ {\btheta} : \| {\btheta} \|_{\infty} \le 1, \btheta_{\mathcal{A}^c} = 0\}, \quad\text{span}(D) = M \\
&I = \{ {\btheta} : \| {\btheta} \|_{\infty} \le 1, {\btheta}_{\mathcal{A}} = 0\}, \quad\text{span}(I) = M ^{\perp} \\
&h_D( {\btheta} ) = \|{\btheta}_\mathcal{A}\|_1, \quad h_I( {\btheta}) = \| {\btheta}_{\mathcal{A}^c}\|_1.
\end{align*}
Then, the RE condition becomes equivalent to:
\begin{equation} \label{eqn:recond-2nd} 
\exists \tau \in (0,1)  \quad \text{s.t.} \quad \|\widehat{\bSigma}^{\rm LPD}_{\mathcal{A}^c\mathcal{A}} (\widehat{\bSigma}^{\rm LPD}_{\mathcal{A}\mathcal{A}})^{-1} \|_{\infty} \le 1 - \tau
\end{equation}
which is the classical irrepresentability, proposed in \cite{Zhao:2006}. 
\begin{proof}[Proof of (\ref{eqn:recond-2nd})]
	\begin{align*}
	\partial h_D( {\btheta}) &= \{ {\by}\in D : {\by}^{\top} {\btheta} = h_D( {\btheta} ) \} \\
	&= \{ {\by}\in D : {\by}^{\top} {\btheta} = \| {\btheta}_\mathcal{A}\|_1 \} \\
	&= \text{sgn}( {\btheta}) \\
	\partial h_D(M) &= \{\text{sgn}({\btheta}) : {\btheta} \in M \}
	\end{align*}
	$$
	\bP_M = \begin{bmatrix} \bI_{|\mathcal{A}|} & \bf 0 \\ \bf 0 & \bf 0 \end{bmatrix}, \qquad
	\bP_{M^\perp} = \begin{bmatrix} \bf 0 & \bf 0 \\ \bf 0 & \bI_{p-|\mathcal{A}|} \end{bmatrix} \\
	$$
	\begin{align*}
	(\bP_M \bQ \bP_M)^{\dagger} &= \begin{bmatrix} \bQ_{\mathcal{A}\mathcal{A}} & \bf 0 \\ \bf 0 & \bf 0 \end{bmatrix}^{\dagger} \\
	&= \begin{bmatrix} (\bQ_{\mathcal{A}\mathcal{A}}^* \bQ_{\mathcal{A}\mathcal{A}})^{\dagger}\bQ_{\mathcal{A}\mathcal{A}}^* & \bf 0 \\ \bf 0 & \bf 0 \end{bmatrix} \\
	\end{align*}
	\begin{align*}	
	\bP_{M^\perp} \{ \bQ \bP_M(\bP_M \bQ \bP_M)^{\dagger} \bP_M {\bz} -  {\bz} \} &= \begin{bmatrix} \bf 0 & \bf 0 \\ \bQ_{\mathcal{A}^c\mathcal{A}} & \bf 0 \end{bmatrix} \begin{bmatrix} (\bQ_{\mathcal{A}\mathcal{A}}^* \bQ_{\mathcal{A}\mathcal{A}})^{\dagger}\bQ_{\mathcal{A}\mathcal{A}}^* & \bf 0 \\ \bf 0 & \bf 0 \end{bmatrix} \begin{bmatrix} {\bz}_1 \\ \bf 0 \end{bmatrix} - \begin{bmatrix} \bf 0 \\ {\bz}_2 \end{bmatrix} \\
	&= \begin{bmatrix} \bf 0 \\ \bQ_{\mathcal{A}^c\mathcal{A}}(\bQ_{\mathcal{A}\mathcal{A}}^* \bQ_{\mathcal{A}\mathcal{A}})^{\dagger} \bQ_{\mathcal{A}\mathcal{A}}^* {\bz}_1 - {\bz}_2 \end{bmatrix}
	\end{align*}
	$$
	\begin{array}{l}
	\sup_{{\bz}\in \partial h_D(M)} V\Big[ \bP_{M^\perp} \{ \bQ \bP_M(\bP_M \bQ \bP_M)^{\dagger} \bP_M {\bz} - {\bz} \} \Big]\\
	\qquad = \sup_{ {\bz} \in \partial h_D(M)} V\left(\begin{bmatrix} \bf 0 \\  \bQ_{\mathcal{A}^c\mathcal{A}}(\bQ_{\mathcal{A}\mathcal{A}}^* \bQ_{\mathcal{A}\mathcal{A}})^{\dagger} \bQ_{\mathcal{A}\mathcal{A}}^* {\bz}_1 - {\bz}_2 \end{bmatrix} \right) \\
	\qquad= \sup_{{\btheta}_1 \in \mathbb{R}^{|\mathcal{A}|}} V\left(\begin{bmatrix} \bf 0 \\  \bQ_{\mathcal{A}^c\mathcal{A}}(\bQ_{\mathcal{A}\mathcal{A}}^* \bQ_{\mathcal{A}\mathcal{A}})^{\dagger} \bQ_{\mathcal{A}\mathcal{A}}^*\text{sgn}( {\btheta}_1) \end{bmatrix} \right) \\
	\qquad= \sup_{{\btheta}_1 \in \mathbb{R}^{|\mathcal{A}|}} \|\bQ_{\mathcal{A}^c\mathcal{A}}(\bQ_{\mathcal{A}\mathcal{A}}^* \bQ_{\mathcal{A}\mathcal{A}})^{\dagger} \bQ_{\mathcal{A}\mathcal{A}}^*\text{sgn}( {\btheta}_1) \|_{\infty}
	\end{array}
	$$
	Since $\bQ_{\mathcal{A}\mathcal{A}}$ is invertible due to Assumption \ref{assm:class_cov}, we have
	\begin{align*}
	& \sup_{ {\btheta}_1 \in \mathbb{R}^{|\mathcal{A}|}} \|\bQ_{\mathcal{A}^c\mathcal{A}}(\bQ_{\mathcal{A}\mathcal{A}}^* \bQ_{\mathcal{A}\mathcal{A}})^{\dagger} \bQ_{\mathcal{A}\mathcal{A}}^*\text{sgn}( {\btheta}_1) \|_{\infty} \\
	=& \sup_{{\btheta}_1 \in \mathbb{R}^{|\mathcal{A}|}} \|\bQ_{\mathcal{A}^c\mathcal{A}}\bQ_{\mathcal{A}\mathcal{A}}^{-1}\text{sgn}( {\btheta}_1) \|_{\infty} \\
	=& \|\bQ_{\mathcal{A}^c\mathcal{A}}\bQ_{\mathcal{A}\mathcal{A}}^{-1}\|_{\infty}
	\end{align*}
\end{proof}

\subsubsection{BG condition} 

The last condition is the bounded gradient (BG) condition. Let us first define related constants. The compatibility constant, denoted 
by $\kappa_{\rho}$, between $\rho$ and $\ell_2$-norm on $M$ is defined by 
$$
\kappa_\rho = \sup_{\btheta} \{ \rho({\btheta}) | \btheta \in \mathcal{B}_2 \cap M\},
$$
where $\mathcal{B}_2$ is the $\ell_2$-unit ball. The compatibility constant between the irrepresentable term and $\rho^*$ is given as 
$$
\kappa_{\text{IC}} = \sup_{\rho^*(\bf z) \le 1} V\Big[ \bP_{M^\perp} \{ \bQ \bP_M(\bP_M \bQ \bP_M)^{\dagger} \bP_M {\bf z} - {\bf z} \} \Big].
$$
We can state the third condition with the constants $\kappa_\rho$ and $\kappa_{\rm IC}$, which decides a suitable range of a tuning parameter $\lambda$.
\begin{assumption}[BG]
	$$
	\frac{4\kappa_{\rm IC}}{\tau} \rho^*(\nabla \ell( {\btheta}^*)) < \lambda < \frac{m^2}{2L} \left( 2\kappa_\rho + \frac{\kappa_{\rho}}{\kappa_{\rm IC}} \frac{\tau}{2} \right)^{-2} \frac{\tau}{\kappa_{\rho^*} \kappa_{\rm IC}}.
	$$
\end{assumption}
\noindent
Now, we check the preliminaries for the BG condition. In our case, $\rho$ is the $\ell_1$-norm, $\kappa_\rho = \sqrt{|\mathcal{A}|}$ and $\kappa_{\rho^*} = 1$. As for $\kappa_{\rm IC}$:
\begin{align*}
\kappa_{\rm IC} &= \sup_{\rho^*(\bf z) \le 1} V\Big[ \bP_{M^\perp} \{ \bQ \bP_M(\bP_M \bQ \bP_M)^{\dagger} \bP_M {\bf z} - {\bf z} \} \Big] \\
&= \sup_{\| {\bf z} \|_\infty \le 1} \|\bQ_{\mathcal{A}^c\mathcal{A}}\bQ_{\mathcal{A}\mathcal{A}}^{-1} {\bf z}_1 - {\bf z}_2\|_{\infty} \\
&= \|\bQ_{\mathcal{A}^c\mathcal{A}}\bQ_{\mathcal{A}\mathcal{A}}^{-1}\|_\infty +1
\end{align*}
Recall the BG condition for $\lambda$:
$$
\frac{4\kappa_{\rm IC}}{\tau} \rho^*(\nabla \ell( {\btheta}^*)) < \lambda < \frac{m^2}{2L} \left( 2\kappa_\rho + \frac{\kappa_{\rho}}{\kappa_{\rm IC}} \frac{\tau}{2} \right)^{-2} \frac{\tau}{\kappa_{\rho^*} \kappa_{\rm IC}}.
$$
With the IR condition, we have $\kappa_{\rm IC} \le 2-\tau$. Also, since $L$ can be of any value, the right side of the BG condition holds. So, the following is sufficient for the BG condition:
$$
\frac{4(2-\tau)}{\tau} \|\nabla \ell(\btheta^*)\|_\infty < \lambda. 
$$

\subsubsection{Conclusion}
Under the three conditions above, \cite{Lee:2015} concluded the following results for the solution.
\begin{itemize} 
	\item[1.] The minimizer is unique.
	\item[2.] $\ell_2$ consistency: $\|\hat{\btheta} - {\btheta}^*\|_2 \le \frac{2}{m} \left( \kappa_\rho + \frac{\tau}{4}\frac{\kappa_\rho}{\kappa_{\rm IC}} \right) \lambda$ 
	\item[3.] Model selection consistency : $\hat{\btheta} \in M$. 
\end{itemize}

\noindent 
In our problem \eqref{main-eqn:lasso-problem}, the $\ell_2$ consistency is 
\begin{align*}
\|\widehat{\bbeta}^{\rm LPD} - {\bbeta}^*\|_2 &\le 
\frac{2}{\min\{0.5\lambda_{\min}(\bSigma_{\mathcal{A}\mathcal{A}}), \mu \}} \left( \sqrt{|\mathcal{A}|} + \frac{\tau}{4}\frac{\sqrt{|\mathcal{A}|}}{\|\bQ_{\mathcal{A}^c\mathcal{A}}\bQ_{\mathcal{A}\mathcal{A}}^{-1}\|_\infty +1} \right) \lambda \\
&\le \frac{2}{\min\{0.5\lambda_{\min}(\bSigma_{\mathcal{A}\mathcal{A}}), \mu \}} \left( 1 + \frac{\tau}{4} \right)\sqrt{|\mathcal{A}|} \lambda,
\end{align*}
and the model selection consistency is $\widehat{\bbeta}^{\rm LPD}_{\mathcal{A}^c} = 0$.

\subsection{Proof of Proposition \ref{prop:LPD_solution}}\label{sec:pf_LPD_solution}
\begin{proof}
	Let $a_{d,\max} = \max_j a_{jj}$, $a_{d,\min} = \min_j a_{jj}$.
	$$
	\begin{array}{rcl}
	\Big|\Big| \Phi_{\mu,\alpha^*}	 - \bA \Big|\Big|_{\max} 
	&= &
	\dfrac{(\epsilon - \lambda_{\min}(\bA))\Big|\Big| \bA - \mu \bI \Big|\Big|_{\max}}{\mu - \lambda_{\min}(\bA)} \\
	& = & 
	(\epsilon - \lambda_{\min}(\bA))\dfrac{\max_{i\neq j} |a_{ij}| \vee \max_i |a_{ii} - \mu| }{\mu - \lambda_{\min}(\bA)} \\
	& = & (\epsilon - \lambda_{\min}(\bA))\dfrac{\max_{i\neq j} |a_{ij}| \vee  |a_{d,\max} - \mu| \vee |a_{d,\min} - \mu| }{\mu - \lambda_{\min}(\bA)} \\
	\end{array}
	$$
	We now denote $a_{\text{off}, \max} = \max_{i\neq j} |a_{ij}|$, 
	$\Psi(\mu) = \dfrac{a_{\text{off}, \max} \vee  |a_{d,\max} - \mu| \vee |a_{d,\min} - \mu| }{\mu - \lambda_{\min}(\bA)}$,
	and consider two disjoint cases: Case (i) $(a_{d,\max} - a_{d,\min})/2  > a_{\text{off}, \max}$ and Case (ii) $(a_{d,\max} - a_{d,\min})/2  \le a_{\text{off}, \max}$.
	For each case, we divide up the value of $\mu$ into multiple cases, which is summarized in Figure \ref{fig:mu_cases}.
	
	\begin{figure}[H]
		\centering
		\tikzset{every picture/.style={line width=0.75pt}} 
		
		\begin{tikzpicture}[x=0.75pt,y=0.75pt,yscale=-1,xscale=1]
		
		\draw    (74,50.99) -- (385,50.01) ;
		\draw [shift={(387,50)}, rotate = 179.82] [color={rgb, 255:red, 0; green, 0; blue, 0 }  ][line width=0.75]    (10.93,-3.29) .. controls (6.95,-1.4) and (3.31,-0.3) .. (0,0) .. controls (3.31,0.3) and (6.95,1.4) .. (10.93,3.29)   ;
		\draw [shift={(72,51)}, rotate = 359.82] [color={rgb, 255:red, 0; green, 0; blue, 0 }  ][line width=0.75]    (10.93,-3.29) .. controls (6.95,-1.4) and (3.31,-0.3) .. (0,0) .. controls (3.31,0.3) and (6.95,1.4) .. (10.93,3.29)   ;
		\draw    (241,44) -- (241,56) ;
		\draw  [dash pattern={on 0.84pt off 2.51pt}]  (241,0.5) -- (241,56) ;
		\draw    (74,191.99) -- (385,191.01) ;
		\draw [shift={(387,191)}, rotate = 179.82] [color={rgb, 255:red, 0; green, 0; blue, 0 }  ][line width=0.75]    (10.93,-3.29) .. controls (6.95,-1.4) and (3.31,-0.3) .. (0,0) .. controls (3.31,0.3) and (6.95,1.4) .. (10.93,3.29)   ;
		\draw [shift={(72,192)}, rotate = 359.82] [color={rgb, 255:red, 0; green, 0; blue, 0 }  ][line width=0.75]    (10.93,-3.29) .. controls (6.95,-1.4) and (3.31,-0.3) .. (0,0) .. controls (3.31,0.3) and (6.95,1.4) .. (10.93,3.29)   ;
		\draw    (241,185) -- (241,197) ;
		\draw  [dash pattern={on 0.84pt off 2.51pt}]  (241,140.5) -- (241,196) ;
		\draw    (141,186) -- (141,198) ;
		\draw  [dash pattern={on 0.84pt off 2.51pt}]  (141,141.5) -- (141,197) ;
		\draw    (341,186) -- (341,198) ;
		\draw  [dash pattern={on 0.84pt off 2.51pt}]  (341,141.5) -- (341,197) ;
		
		\draw (392,39.4) node [anchor=north west][inner sep=0.75pt]    {$\mu $};
		\draw (183,59.4) node [anchor=north west][inner sep=0.75pt]    {$\frac{a_{d,\max} +a_{d,\min}}{2}$};
		\draw (128,14) node [anchor=north west][inner sep=0.75pt]   [align=left] {Case (i)-2};
		\draw (270,12) node [anchor=north west][inner sep=0.75pt]   [align=left] {Case (i)-1};
		\draw (392,180.4) node [anchor=north west][inner sep=0.75pt]    {$\mu $};
		\draw (193,200.4) node [anchor=north west][inner sep=0.75pt]    {$\frac{a_{d,\max} +a_{d,\min}}{2}$};
		\draw (93,199.4) node [anchor=north west][inner sep=0.75pt]    {$a_{d,\max} -a_{\text{off}}$};
		\draw (313,199.4) node [anchor=north west][inner sep=0.75pt]    {$a_{d,\max} -a_{\text{off}}$};
		\draw (348,146) node [anchor=north west][inner sep=0.75pt]   [align=left] {Case (ii)-1};
		\draw (255,146) node [anchor=north west][inner sep=0.75pt]   [align=left] {Case (ii)-2};
		\draw (155,146) node [anchor=north west][inner sep=0.75pt]   [align=left] {Case (ii)-3};
		\draw (65,146) node [anchor=north west][inner sep=0.75pt]   [align=left] {Case (ii)-4};
		\draw (4,7) node [anchor=north west][inner sep=0.75pt]   [align=left] {{\large \textbf{Case (i)}}};
		\draw (4,117) node [anchor=north west][inner sep=0.75pt]   [align=left] {{\large \textbf{Case (ii)}}};

		\end{tikzpicture}
		\caption{Summary of cases used in the proof. Case (i) (top) and Case (ii) (bottom).}
		\label{fig:mu_cases}
	\end{figure}
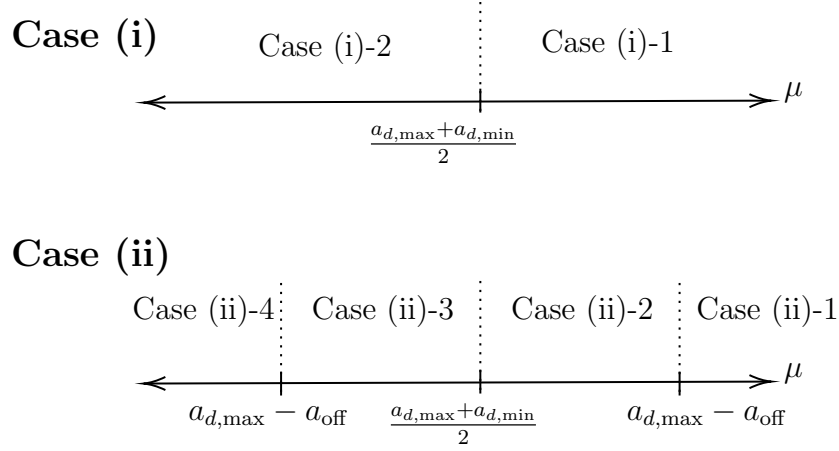

	\subsubsection*{Case (i): $(a_{d,\max} - a_{d,\min})/2  > a_{\text{off}, \max}$}
	
	For this case, we consider two sub-cases based on the value of $\mu$.
	
	\subsubsection*{\footnotesize Case (i)-1: $\mu > (a_{d,\max} + a_{d,\min})/2$}
	Under Case (i)-1, we have $|a_{d,\max} - \mu| < |a_{d,\min} - \mu|$.	
	Moreover, note that by Case (i)
	$$
	\dfrac{a_{d,\max} + a_{d,\min}}{2} = 	\dfrac{a_{d,\max} - a_{d,\min}}{2} + a_{d,\min} > a_{d,\min} + a_{\text{off}, \max}
	$$
	and thus $\mu - a_{d,\min} > a_{\text{off}, \max}$. 
	Combining these two, we can simplify $\Psi$ by
	\begin{equation}\label{eq:Psi_case_i_1}
	\Psi(\mu) = \dfrac{|a_{d,\min} - \mu|}{\mu - \lambda_{\min}(\bA)} = \dfrac{\mu - a_{d,\min}}{\mu - \lambda_{\min}(\bA)} = \dfrac{\lambda_{\min}(\bA) - a_{d,\min}}{\mu - \lambda_{\min}(\bA)} + 1.
	\end{equation}
	From the last expression, we can see that $\Psi$ is increasing in $\mu$ because $a_{d,\min} > \lambda_{\min}(\bA)$. Thus, the minimum value under the case considered is 
	$$
	\min \Big\{\Psi(\mu):  \mu > (a_{d,\max} + a_{d,\min})/2 \Big\}\ge \dfrac{(a_{d,\max} - a_{d,\min})/2}{(a_{d,\max} + a_{d,\min})/2 - \lambda_{\min}(\bA)},
	$$
	where the right-hand side is achieved by plugging-in $\mu = \dfrac{a_{d,\max} + a_{d,\min}}{2}$ into (\ref{eq:Psi_case_i_1}).
	
	\subsubsection*{\footnotesize Case (i)-2: $\mu \le (a_{d,\max} + a_{d,\min})/2$}
	
	Under Case (i)-2, we have $|a_{d,\max} - \mu| \ge |a_{d,\min} - \mu|$.	
	Moreover, note that by Case (i)
	$$
	a_{\text{off}, \max} < \dfrac{a_{d,\max} - a_{d,\min}}{2} = 	a_{d,\max} - \dfrac{a_{d,\max} + a_{d,\min}}{2}
	$$
	and thus $a_{d,\max} - \mu > a_{\text{off}, \max}$. 
	Combining these two, we can simplify $\Psi$ by
	\begin{equation}\label{eq:Psi_case_i_2}
	\Psi(\mu) = \dfrac{|a_{d,\max} - \mu|}{\mu - \lambda_{\min}(\bA)} = \dfrac{a_{d,\max} - \mu}{\mu - \lambda_{\min}(\bA)} = \dfrac{a_{d,\max} - \lambda_{\min}(\bA)}{\mu - \lambda_{\min}(\bA)} - 1.
	\end{equation}
	The last expression tells us that $\Psi$ is decreasing in $\mu$ because $a_{d,\max} > \lambda_{\min}(\bA)$. Then, we get
	$$
	\min \Big\{\Psi(\mu):  \mu \le (a_{d,\max} + a_{d,\min})/2 \Big\} = \dfrac{(a_{d,\max} - a_{d,\min})/2}{(a_{d,\max} + a_{d,\min})/2 - \lambda_{\min}(\bA)}.
	$$

	Combining the two results from Case (i)-1,2, we conclude that if $(a_{d,\max} - a_{d,\min})/2  > a_{\text{off}, \max}$, then the minimum value of $\Psi$ is 
	$$
	\min_{\mu: \mu \ge \epsilon} \Psi(\mu) = \dfrac{(a_{d,\max} - a_{d,\min})/2}{(a_{d,\max} + a_{d,\min})/2 - \lambda_{\min}(\bA)}
	$$
	at $\mu = (a_{d,\max} + a_{d,\min})/2$.

	\subsubsection*{Case (ii): $(a_{d,\max} - a_{d,\min})/2 \le a_{\text{off}, \max}$}
	
	Similarly to before, we consider sub-cases based on the value of $\mu$.
	
	\subsubsection*{\footnotesize Case (ii)-1: $\mu > a_{d,\min} + a_{\text{off}, \max}$}
	
	Note that $a_{d,\min} + a_{\text{off}, \max} \ge (a_{d,\max} + a_{d,\min})/2$ under Case (ii). Then, we have 
	$|a_{d,\max} - \mu| < |a_{d,\min} - \mu| = \mu - a_{d,\min}$. Moreover, by Case (ii)-1, $|a_{d,\min} - \mu| = \mu - a_{d,\min} > a_{\text{off}, \max}$.

	Thus, we can simplify $\Psi$ by
	\begin{equation}\label{eq:Psi_case_ii_1}
	\Psi(\mu) = \dfrac{|a_{d,\min} - \mu|}{\mu - \lambda_{\min}(\bA)} = \dfrac{\mu - a_{d,\min}}{\mu - \lambda_{\min}(\bA)} = \dfrac{\lambda_{\min}(\bA) - a_{d,\min}}{\mu - \lambda_{\min}(\bA)} + 1.
	\end{equation}
	
	\subsubsection*{\footnotesize Case (ii)-2: $(a_{d,\max} + a_{d,\min})/2 < \mu \le a_{d,\min} + a_{\text{off}, \max}$}
	In Case (ii)-2, we still have $|a_{d,\max} - \mu| < |a_{d,\min} - \mu| = \mu - a_{d,\min}$ as in Case (ii)-1, but $|a_{d,\min} - \mu| = \mu - a_{d,\min} \ge a_{\text{off}, \max}$ holds.

	\subsubsection*{\footnotesize Case (ii)-3: $a_{d,\max} - a_{\text{off}, \max} < \mu \le (a_{d,\max} + a_{d,\min})/2$}
	
	From $\mu \le (a_{d,\max} + a_{d,\min})/2$, we have $|a_{d,\max} - \mu| \ge |a_{d,\min} - \mu|$. Moreover, since $a_{d,\max} - a_{\text{off}, \max} < \mu$, $|a_{d,\max} - \mu| = a_{d,\max} - \mu < a_{\text{off}, \max}$. 
	
	\subsubsection*{\footnotesize Case (ii)-4: $\mu \le a_{d,\max} - a_{\text{off}, \max}$}
	
	Note that $a_{d,\max} - a_{\text{off}, \max} \le (a_{d,\max} + a_{d,\min})/2$ under Case (ii).
	Thus, we have $|a_{d,\max} - \mu| \ge |a_{d,\min} - \mu|$. Since $\mu \le a_{d,\max} - a_{\text{off}, \max}$, $|a_{d,\max} - \mu| = a_{d,\max} - \mu \ge a_{\text{off}, \max}$.

	Combining the four cases, we can summarize that
	$$
	\Psi(\mu) = \begin{cases}
	\dfrac{\mu - a_{d,\min}}{\mu - \lambda_{\min}(\bA)},	& \text{for Case (ii)-1}  \\[1em]
	\dfrac{a_{\text{off}, \max}}{\mu - \lambda_{\min}(\bA)}, & \text{for Case (ii)-2,3}\\[1em]
	\dfrac{a_{d,\max} - \mu}{\mu - \lambda_{\min}(\bA)}, & \text{for Case (ii)-4}
	\end{cases}
	$$
	We note that this function decreases until $\mu < a_{d,\min} + a_{\text{off}, \max}$ and increases after that point, which implies $\mu = a_{d,\min} + a_{\text{off}, \max}$ give the minimum value
	$$
	\min_{\mu: \mu \ge \epsilon} \Psi(\mu) = \dfrac{a_{\text{off}, \max}}{a_{d,\min} + a_{\text{off}, \max} - \lambda_{\min}(\bA)}.
	$$
\end{proof}

\section{Proof of the main theorems}\label{sec:proofs_main}

\subsection{Proof of Theorem \ref{thm:IR_LPD_combined}}\label{sec:pf_thm_IR}
The proof of Theorem \ref{thm:IR_LPD_combined} is based on Theorem \ref{thm:IR_LPD_neg}, \ref{thm:IR_LPD_pos}, which are stated below.
\begin{thm}\label{thm:IR_LPD_neg}
	Let Assumption \ref{assm:sub_G}, \ref{assm:multi_bern}, \ref{assm:mcar}, \ref{assm:class_cov} hold.
	Let us focus on the case of the estimator $\widehat{\bSigma}^{\rm IPW}$ such that 
	$\widehat{\bSigma}^{\rm IPW}_{\mathcal{A}\mathcal{A}}$ is non-singular  and the smallest eigenvalue satisfies $\lambda_{\min}(\widehat{\bSigma}^{\rm IPW})\le 0$.
	For any $\mu > \epsilon$, we construct the LPD estimator $\Phi_{\mu, \alpha^*}(\widehat{\bSigma}^{\rm IPW})$ with $\alpha^* = (\mu - \epsilon) / (\mu - \lambda_{\min}(\widehat{\bSigma}^{\rm IPW}))$.
	Then, the LPD estimator satisfies the irrepresentability condition for some constant $\tilde{\tau} \in (0,1)$,
	if the events hold true
	\begin{equation}\label{eq:random_closeness_neg}
	\begin{array}{rcl}
	\left\|\widehat{\bSigma}^{\rm IPW}_{\mathcal{A}\mathcal{A}} - \bSigma_{\mathcal{A}\mathcal{A}}\right\|_{\infty} + 
	\left\|\widehat{\bSigma}^{\rm IPW}_{\mathcal{A}^c\mathcal{A}} - \bSigma_{\mathcal{A}^c\mathcal{A}} \right\|_{\infty} + \dfrac{\mu}{\mu-\epsilon} \left\|\widehat{\bSigma}^{\rm IPW} - \bSigma \right\|_2
	&\le& \dfrac{\tau }{\left\|\bSigma_{\mathcal{A}\mathcal{A}}^{-1}\right\|_{\infty}}, \\[1em]
	\end{array}
	\end{equation}
\end{thm}
\noindent
The proof is pended until Supplementary Materials \ref{sec:pf_thm_IR_neg}. The other case when the smallest eigenvalue is positive is addressed by the following theorem.

\begin{thm}\label{thm:IR_LPD_pos}
	Let Assumption \ref{assm:sub_G}, \ref{assm:multi_bern}, \ref{assm:mcar}, \ref{assm:class_cov}(b) hold where $\tau \in (0,1)$ is the constant from Assumption \ref{assm:class_cov}(b).
	Let us focus on the case of the estimator $\widehat{\bSigma}^{\rm IPW}$ such that 
	$\widehat{\bSigma}^{\rm IPW}_{\mathcal{A}\mathcal{A}}$ is non-singular  and the smallest eigenvalue satisfies $\lambda_{\min}(\widehat{\bSigma}^{\rm IPW}) > 0$.
	Then, the LPD estimator $\Phi_{\mu, \alpha^*}(\widehat{\bSigma}^{\rm IPW})$, which is reduced to $\widehat{\bSigma}^{\rm IPW}$ with $\alpha^* = 1$, satisfies the irrepresentability condition for some constant $\tilde{\tau} \in (0,1)$,
	if the event holds true
	\begin{equation}\label{eq:random_closeness_pos}
	\left\|\widehat{\bSigma}^{\rm IPW}_{\mathcal{A}\mathcal{A}} - \bSigma_{\mathcal{A}\mathcal{A}}\right\|_{\infty} + 
	\left\|\widehat{\bSigma}^{\rm IPW}_{\mathcal{A}^c\mathcal{A}} - \bSigma_{\mathcal{A}^c\mathcal{A}} \right\|_{\infty}
	\le \tau / \left\|\bSigma_{\mathcal{A}\mathcal{A}}^{-1}\right\|_{\infty}.
	\end{equation}
\end{thm}
\noindent
The proof is pended until Supplementary Materials \ref{sec:pf_thm_IR_pos}.

\begin{proof}[Proof of Theorem \ref{thm:IR_LPD_combined}]
	We calculate the probability of the event $E$ that the LPD estimator satisfies the irrepresentability condition as follows. Let the event $A = \{\lambda_{\min}(\widehat{\bSigma}^{\rm IPW}) > 0\}$.
	$$
	\begin{array}{rl}
	{\rm P}\left(E\right) &= {\rm P}\left(E \big| A \right){\rm P}\left(A \right) + 
	{\rm P}\left(E \big| A^c \right){\rm P}\left(A^c \right)\\
	& \ge 
	{\rm P}\left((\ref{eq:random_closeness_pos}) \text{ holds} \big| A \right){\rm P}\left(A \right) + 
	{\rm P}\left( (\ref{eq:random_closeness_neg}) \text{ holds} \big| A^c \right){\rm P}\left(A^c \right) (\because \text{ Theorem \ref{thm:IR_LPD_neg}, \ref{thm:IR_LPD_pos}})\\
	& \ge 
	{\rm P}\left((\ref{eq:random_closeness_neg}) \text{ holds} \big| A \right){\rm P}\left(A \right) + 
	{\rm P}\left( (\ref{eq:random_closeness_neg}) \text{ holds} \big| A^c \right){\rm P}\left(A^c \right) (\because (\ref{eq:random_closeness_neg}) \Rightarrow (\ref{eq:random_closeness_pos}))\\
	& = {\rm P}\left( (\ref{eq:random_closeness_neg}) \text{ holds} \right).
	\end{array}
	$$
	\noindent
	Note that for $\tilde{\bSigma} = \widehat{\bSigma}^{\rm IPW} - \bSigma$, we have
	$$
	\left\|\tilde{\bSigma}_{\mathcal{A}\mathcal{A}} \right\|_{\infty} + 
	\left\|\tilde{\bSigma}_{\mathcal{A}^c\mathcal{A}} \right\|_{\infty} \le 
	2 \left\|\tilde{\bSigma}\right\|_{\infty, \mathcal{A}}  = 2 \left\|
	\tilde{\bSigma} \begin{bmatrix}
	\bI & \bzero \\
	\bzero & \bzero
	\end{bmatrix}
	\right\|_{\infty} \le 
	2 \left\|\tilde{\bSigma}\right\|_{\infty} \le 
	2\left\| \tilde{\bSigma} \right\|_2.
	$$	
	Then, using $\mu/(\mu-\epsilon) \le 2$ for $\mu \ge 2 \epsilon$, 
	a sufficient condition for (\ref{eq:random_closeness_neg}) is 
	$$
	\left\| \tilde{\bSigma} \right\|_2 \le \dfrac{\tau }{4 \left\|\bSigma_{\mathcal{A}\mathcal{A}}^{-1}\right\|_{\infty}}.
	$$
	Theorem \ref{thm:IPW_spectral} states that for any $u>0$, if $n > \pi_{\max}^{(4)} (u + 1)^3 \log^3(p \vee n)$, then it holds with probability at least $1-3/p^u$
	$$
	|| \widehat{\bSigma}^{\rm IPW} - \bSigma ||_2 \le
	C \tr(\bSigma) \max\{(K^x)^2, 1\} \sqrt{u+1} \sqrt{\dfrac{\pi_{\max}^{(4)} \log p}{n}}.
	$$
	Hence, if the following condition is satisfied
	$$
	C \tr(\bSigma) \max\{(K^x)^2, 1\} \sqrt{u+1} \sqrt{\dfrac{\pi_{\max}^{(4)} \log p}{n}}
	\le
	\dfrac{\tau }{4 \left\|\bSigma_{\mathcal{A}\mathcal{A}}^{-1}\right\|_{\infty}},
	$$
	then we can guarantee ${\rm P}\left( (\ref{eq:random_closeness_neg}) \text{ holds} \right) \ge 1-3/p^u$, where the above gives another sample size condition:
	$$
	n / (\pi_{\max}^{(4)} \log p) \ge 4C \bigg\{ 
	\dfrac{\tr(\bSigma) \max\{(K^x)^2, 1\} \sqrt{u+1}}{\tau/\left\|\bSigma_{\mathcal{A}\mathcal{A}}^{-1}\right\|_{\infty}}
	\bigg\}^2.
	$$

	Finally, we deal with (C3) of Proposition \ref{prop:consistency}. By Weyl's inequality, the condition is satisfied if
	$||\widehat{\bSigma}^{\text{IPW}}_{\mathcal{A}\mathcal{A}} - \bSigma_{\mathcal{A}\mathcal{A}}||_2 \le 0.5 \lambda_{\min}(\bSigma_{\mathcal{A}\mathcal{A}}) $ holds. Following the proof of Theorem \ref{thm:IR_LPD_combined}, we can have a similar probabilistic argument for the event $\{||\widehat{\bSigma}^{\text{IPW}}_{\mathcal{A}\mathcal{A}} - \bSigma_{\mathcal{A}\mathcal{A}}||_2 \le 0.5 \lambda_{\min}(\bSigma_{\mathcal{A}\mathcal{A}}) \}$. That is, $||\widehat{\bSigma}^{\text{IPW}}_{\mathcal{A}\mathcal{A}} - \bSigma_{\mathcal{A}\mathcal{A}}||_2 \le 0.5 \lambda_{\min}(\bSigma_{\mathcal{A}\mathcal{A}})$ with probability greater than $1-3/p^u$ for $u > 0$ if the sample size satisfies 
	$$
	\dfrac{n}{\pi_{\max, \mathcal{A}}^{(4)} \log |\mathcal{A}|} \ge c \bigg\{ 
	\dfrac{\tr(\bSigma_{\mathcal{A}\mathcal{A}}) \max\{(K^x)^2, 1\} \sqrt{u+1}}{1/\lambda_{\min}(\bSigma_{\mathcal{A}\mathcal{A}})}
	\bigg\}^2, \quad n > c \,\pi_{\max, \mathcal{A}}^{(4)} (u + 1)^3 \log^3(|\mathcal{A}| \vee n),
	$$
	for some $c>0$. 
	Here, $\pi_{\max, \mathcal{A}}^{(4)} = \max_{k_1, k_2,\ell_1, \ell_2 \in \mathcal{A}} \pi^{xx}_{k_1 k_2 \ell_1 \ell_2}/(\pi^{xx}_{k_1 \ell_1} \pi^{xx}_{k_2\ell_2})$.

\end{proof}

\subsection{Proof of Theorem \ref{thm:IR_LPD_neg}}\label{sec:pf_thm_IR_neg}

It should be noted that the proof of the theorem only depends on the distances between $\widehat{\bSigma}^{\rm IPW}$ and $\bSigma$ (or their block matrices), but not any other characteristic of the IPW estimate or the population covariance matrix.

We define the matrix norms that appear in the following proof.
$$
\begin{array}{c}
\eta_1 = \left\|\bSigma_{\mathcal{A}\mathcal{A}}^{-1}\right\|_{\infty}, \quad \eta_2 = \left\|\bSigma_{\mathcal{A}^c\mathcal{A}} \bSigma_{\mathcal{A}\mathcal{A}}^{-1}\right\|_{\infty}\\[1em]	
\delta_1 = \left\|\widehat{\bSigma}^{\rm IPW}_{\mathcal{A}\mathcal{A}} - \bSigma_{\mathcal{A}\mathcal{A}}\right\|_{\infty}, \quad \delta_2 = \left\| \widehat{\bSigma}^{\rm IPW}_{\mathcal{A}^c\mathcal{A}} - \bSigma_{\mathcal{A}^c\mathcal{A}}\right\|_{\infty}, \quad \delta_3= \left\|\widehat{\bSigma}^{\rm IPW} - \bSigma \right\|_2.
\end{array}
$$
We first introduce the lemma to ease calculation.
\begin{lem}\label{lem:IR_LPD}
	Let $\widehat{\bSigma}^{\rm LPD} = \Phi_{\mu, \alpha}(\widehat{\bSigma}^{\rm IPW})$. Assume 
	\begin{equation}\label{eq:IR_lemma_cond}
	\eta_1 \delta_1 < 1 \text{ and } \quad \dfrac{(1-\alpha)\mu}{\alpha} \|\big(\widehat{\bSigma}^{\rm IPW}_{\mathcal{A}\mathcal{A}}\big)^{-1}\|_{\infty} < 1.		
	\end{equation}
	Then, we have
	$$
	\left\|\widehat{\bSigma}^{\rm LPD}_{\mathcal{A}^c\mathcal{A}}\big(\widehat{\bSigma}^{\rm LPD}_{\mathcal{A}\mathcal{A}}\big)^{-1}\right\|_{\infty} \le \dfrac{\eta_1 \delta_2 + \eta_2}{1 - \eta_1 \delta_1 - \alpha^{-1} (1 - \alpha)\mu \eta_1}.
	$$
\end{lem}	
\noindent
The proof is given in Supplementary Materials \ref{sec:pf_lem_IR_LPD}. Using Lemma \ref{lem:IR_LPD} and the irrpresentability condition for $\bSigma$ (i.e. $\eta_2 < 1 - \tau$) together, we get  
\begin{equation}\label{eq:IR_LPD}
\left\|\widehat{\bSigma}^{\rm LPD}_{\mathcal{A}^c\mathcal{A}}\big(\widehat{\bSigma}^{\rm LPD}_{\mathcal{A}\mathcal{A}}\big)^{-1}\right\|_{\infty} < \dfrac{\eta_1 \delta_2 + 1-\tau}{1 - \eta_1 \delta_1 - \alpha^{-1} (1 - \alpha)\mu \eta_1}.
\end{equation}
It remains to claim the right-hand side of the above is strictly less than $1$, which is equivalent to show
$$
\delta_1 + \delta_2 < \tau/ \eta_1 -  \alpha^{-1}(1 - \alpha) \mu.
$$
Plugging-in $\alpha^* = (\mu - \epsilon) / (\mu - \lambda_{\min}(\widehat{\bSigma}^{\rm IPW}))$ and using 
$\lambda_{\min}(\widehat{\bSigma}^{\rm IPW}) \ge - 
\delta_3 + \lambda_{\min}(\bSigma)$ derived by Weyl's inequality, we get a sufficient condition for (\ref{eq:IR_LPD})
\begin{equation}\label{eq:random_closeness_neg_alias}
\delta_1 + \delta_2  + \dfrac{\mu \delta_3}{\mu -\epsilon} < \dfrac{\tau}{\eta_1} + \dfrac{\mu (\lambda_{\min}(\bSigma_{\mathcal{A}\mathcal{A}}) - \epsilon)}{\mu - \epsilon}.
\end{equation}
Remark that the right-hand side term is greater than $0$ if $\min\{\mu, \lambda_{\min}(\bSigma_{\mathcal{A}\mathcal{A}})\} > \epsilon$.


We remain to show (\ref{eq:IR_lemma_cond}) holds with high probability when plugging-in $\alpha^* = (\mu - \epsilon) / (\mu - \lambda_{\min}(\widehat{\bSigma}^{\rm IPW}))$, but instead, we will calculate the probability of another sufficient condition (\ref{eq:IR_lemma_cond_suff}) described in the following lemma. One can easily check that (\ref{eq:IR_lemma_cond_suff}) is implied by (\ref{eq:random_closeness_neg_alias}) because $\mu / (\mu - \epsilon)> 1$ and $\tau < 1$, which concludes the proof.
\begin{lem}\label{lem:IR_equiv_cond}
	Consider the class of covariance matrices such that $1/ \eta_1 - \epsilon + \lambda_{\min}(\bSigma_{\mathcal{A}\mathcal{A}}) > 0$. 
	Let us focus on the case of the estimator $\widehat{\bSigma}^{\rm IPW}$ with $\lambda_{\min}(\widehat{\bSigma}^{\rm IPW}) < 0$. If we choose $\mu > \epsilon$, then 
	\begin{equation}\label{eq:IR_lemma_cond_suff}
	\delta_1 + \dfrac{\mu \delta_3}{\mu -\epsilon}
	\le 1/\eta_1 + \dfrac{\mu(\lambda_{\min}(\bSigma_{\mathcal{A}\mathcal{A}}) - \epsilon)}{\mu - \epsilon},
	\end{equation}
	implies (\ref{eq:IR_lemma_cond}).
\end{lem}
\noindent
The proof of the lemma is given in Supplementary Materials \ref{sec:pf_lem_IR_LPD}. 
%

\subsection{Proof of lemmas used in Theorem  \ref{thm:IR_LPD_neg}}\label{sec:pf_lem_IR_LPD}

\begin{proof}[Proof of Lemma \ref{lem:IR_LPD}]
	We introduce three inequalities and suspend their proofs.
	\begin{eqnarray}
	\|\widehat{\bSigma}^{\rm LPD}_{\mathcal{A}^c\mathcal{A}}\big(\widehat{\bSigma}^{\rm LPD}_{\mathcal{A}\mathcal{A}}\big)^{-1}\|_{\infty} &\le &
	\dfrac{\|\widehat{\bSigma}^{\rm IPW}_{\mathcal{A}^c\mathcal{A}}\big(\widehat{\bSigma}^{\rm IPW}_{\mathcal{A}\mathcal{A}}\big)^{-1}\|_{\infty}}{1 - \alpha^{-1} (1-\alpha)\mu \|\big(\widehat{\bSigma}^{\rm IPW}_{\mathcal{A}\mathcal{A}}\big)^{-1}\|_{\infty}}, \label{eq:proof_lemma_step_1}\\[0.5em]
	&& \qquad \qquad \text{if } 	\dfrac{(1-\alpha)\mu}{\alpha} \|\big(\widehat{\bSigma}^{\rm IPW}_{\mathcal{A}\mathcal{A}}\big)^{-1}\|_{\infty} < 1, \nonumber\\
	\|\widehat{\bSigma}^{\rm IPW}_{\mathcal{A}^c\mathcal{A}}\big(\widehat{\bSigma}^{\rm IPW}_{\mathcal{A}\mathcal{A}}\big)^{-1} - \bSigma_{\mathcal{A}^c\mathcal{A}}\bSigma_{\mathcal{A}\mathcal{A}}^{-1}\|_\infty &\le& \frac{ \eta_1 \left( \eta_2 \delta_1 + \delta_2 \right)}{1-\eta_1 \delta_1}, \quad \text{if } \eta_1 \delta_1 < 1, \label{eq:proof_lemma_step_2}\\
	\|\big(\widehat{\bSigma}^{\rm IPW}_{\mathcal{A}\mathcal{A}}\big)^{-1}\|_\infty &\le& \frac{\eta_1}{1-\eta_1\delta_1}, \quad \text{if } \eta_1 \delta_1 < 1,\label{eq:proof_lemma_step_3}
	\end{eqnarray}	
	Combining the triangular inequality with (\ref{eq:proof_lemma_step_1}), we get
	$$
	\|\widehat{\bSigma}^{\rm LPD}_{\mathcal{A}^c\mathcal{A}}\big(\widehat{\bSigma}^{\rm LPD}_{\mathcal{A}\mathcal{A}}\big)^{-1}\|_{\infty} \le
	\dfrac{\|\widehat{\bSigma}^{\rm IPW}_{\mathcal{A}^c\mathcal{A}}\big(\widehat{\bSigma}^{\rm IPW}_{\mathcal{A}\mathcal{A}}\big)^{-1} - \bSigma_{\mathcal{A}^c\mathcal{A}}\bSigma_{\mathcal{A}\mathcal{A}}^{-1}\|_{\infty} +  \|\bSigma_{\mathcal{A}^c\mathcal{A}}\bSigma_{\mathcal{A}\mathcal{A}}^{-1}\|_{\infty}}{1 - \alpha^{-1} (1-\alpha)\mu \|\big(\widehat{\bSigma}^{\rm IPW}_{\mathcal{A}\mathcal{A}}\big)^{-1}\|_{\infty}}.
	$$
	This completes the proof if (\ref{eq:proof_lemma_step_2}), (\ref{eq:proof_lemma_step_3}) are combined with the upper bound. 
	
	We now prove the above inequalities. 
	The proofs of (\ref{eq:proof_lemma_step_2}) and (\ref{eq:proof_lemma_step_3}) are from that of Lemma A2 by \cite{Mai:2012}, but we show them here for completeness. Using the basic property of operator norms,
	\begin{align*}
	\|\big(\widehat{\bSigma}^{\rm IPW}_{\mathcal{A}\mathcal{A}}\big)^{-1} - \bSigma_{\mathcal{A}\mathcal{A}}^{-1} \|_{\infty} 
	&= \|\bSigma_{\mathcal{A}\mathcal{A}}^{-1}(\widehat{\bSigma}^{\rm IPW}_{\mathcal{A}\mathcal{A}} - \bSigma_{\mathcal{A}\mathcal{A}})\big(\widehat{\bSigma}^{\rm IPW}_{\mathcal{A}\mathcal{A}}\big)^{-1}\|_\infty \\
	&\le \|\bSigma_{\mathcal{A}\mathcal{A}}^{-1}\|_\infty \cdot \|\widehat{\bSigma}^{\rm IPW}_{\mathcal{A}\mathcal{A}} - \bSigma_{\mathcal{A}\mathcal{A}}\|_\infty \cdot \|\big(\widehat{\bSigma}^{\rm IPW}_{\mathcal{A}\mathcal{A}}\big)^{-1}\|_\infty \\
	&\le \|\bSigma_{\mathcal{A}\mathcal{A}}^{-1}\|_\infty \times \|\widehat{\bSigma}^{\rm IPW}_{\mathcal{A}\mathcal{A}} - \bSigma_{\mathcal{A}\mathcal{A}}\|_\infty \\
	& \qquad \times \big(\|\big(\widehat{\bSigma}^{\rm IPW}_{\mathcal{A}\mathcal{A}}\big)^{-1} - \bSigma_{\mathcal{A}\mathcal{A}}^{-1}\|_\infty + \|\bSigma_{\mathcal{A}\mathcal{A}}^{-1}\|_\infty \big).
	\end{align*}
	Arranging the inequality, we get
	$$
	\|\big(\widehat{\bSigma}^{\rm IPW}_{\mathcal{A}\mathcal{A}}\big)^{-1} - \bSigma_{\mathcal{A}\mathcal{A}}^{-1} \|_{\infty} \le \frac{\|\bSigma_{\mathcal{A}\mathcal{A}}^{-1}\|_\infty^2 \|\widehat{\bSigma}^{\rm IPW}_{\mathcal{A}\mathcal{A}} - \bSigma_{\mathcal{A}\mathcal{A}}\|_\infty}{1-\|\bSigma_{\mathcal{A}\mathcal{A}}^{-1}\|_\infty \|\widehat{\bSigma}^{\rm IPW}_{\mathcal{A}\mathcal{A}} - \bSigma_{\mathcal{A}\mathcal{A}}\|_\infty },
	$$
	since $\|\bSigma_{\mathcal{A}\mathcal{A}}^{-1}\|_\infty \|\widehat{\bSigma}^{\rm IPW}_{\mathcal{A}\mathcal{A}} - \bSigma_{\mathcal{A}\mathcal{A}}\|_\infty < 1$ by the assumption.
	Then, by the triangular inequality,
	\begin{equation}\label{eq:misc_1}
	\begin{array}{rcl}
	\|\big(\widehat{\bSigma}^{\rm IPW}_{\mathcal{A}\mathcal{A}}\big)^{-1}\|_{\infty}
	&\le&  \|\big(\widehat{\bSigma}^{\rm IPW}_{\mathcal{A}\mathcal{A}}\big)^{-1} - \bSigma_{\mathcal{A}\mathcal{A}}^{-1} \|_{\infty} + \| \bSigma_{\mathcal{A}\mathcal{A}}^{-1} \|_{\infty}\\[1em]
	& \le& 
	\dfrac{\|\bSigma_{\mathcal{A}\mathcal{A}}^{-1}\|_\infty^2 \|\widehat{\bSigma}^{\rm IPW}_{\mathcal{A}\mathcal{A}} - \bSigma_{\mathcal{A}\mathcal{A}}\|_\infty}{1-\|\bSigma_{\mathcal{A}\mathcal{A}}^{-1}\|_\infty \|\widehat{\bSigma}^{\rm IPW}_{\mathcal{A}\mathcal{A}} - \bSigma_{\mathcal{A}\mathcal{A}}\|_\infty } + \| \bSigma_{\mathcal{A}\mathcal{A}}^{-1} \|_{\infty},
	\end{array}
	\end{equation}
	which achieves (\ref{eq:proof_lemma_step_3}). Next, we also exploit the basic properties of norms to get
	$$
	\begin{array}{l}
	\|\widehat{\bSigma}^{\rm IPW}_{\mathcal{A}^c\mathcal{A}}\big(\widehat{\bSigma}^{\rm IPW}_{\mathcal{A}\mathcal{A}}\big)^{-1} - \bSigma_{\mathcal{A}^c\mathcal{A}}\bSigma_{\mathcal{A}\mathcal{A}}^{-1}\|_\infty\\[1em] 
	\qquad = \|(\widehat{\bSigma}^{\rm IPW}_{\mathcal{A}^c\mathcal{A}} - \bSigma_{\mathcal{A}^c\mathcal{A}} \bSigma_{\mathcal{A}\mathcal{A}}^{-1} \widehat{\bSigma}^{\rm IPW}_{\mathcal{A}\mathcal{A}}) \big(\widehat{\bSigma}^{\rm IPW}_{\mathcal{A}\mathcal{A}}\big)^{-1}\|_\infty \\[1em]
	\qquad = \|(\widehat{\bSigma}^{\rm IPW}_{\mathcal{A}^c\mathcal{A}} - \bSigma_{\mathcal{A}^c\mathcal{A}} + \bSigma_{\mathcal{A}^c\mathcal{A}}\bSigma_{\mathcal{A}\mathcal{A}}^{-1}\bSigma_{\mathcal{A}\mathcal{A}} - \bSigma_{\mathcal{A}^c\mathcal{A}}\bSigma_{\mathcal{A}\mathcal{A}}^{-1}\widehat{\bSigma}^{\rm IPW}_{\mathcal{A}\mathcal{A}}) \big(\widehat{\bSigma}^{\rm IPW}_{\mathcal{A}\mathcal{A}}\big)^{-1}\|_\infty \\[1em]
	\qquad \le  \|\widehat{\bSigma}^{\rm IPW}_{\mathcal{A}^c\mathcal{A}} - \bSigma_{\mathcal{A}^c\mathcal{A}} +\bSigma_{\mathcal{A}^c\mathcal{A}}\bSigma_{\mathcal{A}\mathcal{A}}^{-1}(\bSigma_{\mathcal{A}\mathcal{A}}-\widehat{\bSigma}^{\rm IPW}_{\mathcal{A}\mathcal{A}}) \|_{\infty} \|(\big(\widehat{\bSigma}^{\rm IPW}_{\mathcal{A}\mathcal{A}}\big)^{-1}\|_\infty \\[1em]
	\qquad \le  \big(\|\widehat{\bSigma}^{\rm IPW}_{\mathcal{A}^c\mathcal{A}} - \bSigma_{\mathcal{A}^c\mathcal{A}}\|_{\infty} + \|\bSigma_{\mathcal{A}^c\mathcal{A}}\bSigma_{\mathcal{A}\mathcal{A}}^{-1}\|_\infty  \|\widehat{\bSigma}^{\rm IPW}_{\mathcal{A}\mathcal{A}} - \bSigma_{\mathcal{A}\mathcal{A}}\|_{\infty} \big) \|\big(\widehat{\bSigma}^{\rm IPW}_{\mathcal{A}\mathcal{A}}\big)^{-1}\|_\infty.
	\end{array}
	$$
	By using (\ref{eq:proof_lemma_step_3}) in the last inequality, we obtain (\ref{eq:proof_lemma_step_2}). To prove (\ref{eq:proof_lemma_step_1}), we observe 
	\begin{align*}
	\|\widehat{\bSigma}^{\rm LPD}_{\mathcal{A}^c\mathcal{A}}\big(\widehat{\bSigma}^{\rm LPD}_{\mathcal{A}\mathcal{A}}\big)^{-1}\|_{\infty} &= \|\alpha \widehat{\bSigma}^{\rm IPW}_{\mathcal{A}^c\mathcal{A}}(\alpha \widehat{\bSigma}^{\rm IPW}_{\mathcal{A}\mathcal{A}} + (1-\alpha)\mu \bI)^{-1}\|_{\infty} \\
	&= \|\widehat{\bSigma}^{\rm IPW}_{\mathcal{A}^c\mathcal{A}} \big(\widehat{\bSigma}^{\rm IPW}_{\mathcal{A}\mathcal{A}}\big)^{-1} (\bI + \alpha^{-1}(1-\alpha)\mu \big(\widehat{\bSigma}^{\rm IPW}_{\mathcal{A}\mathcal{A}}\big)^{-1})^{-1}\|_{\infty} \\
	&\le \|\widehat{\bSigma}^{\rm IPW}_{\mathcal{A}^c\mathcal{A}} \big(\widehat{\bSigma}^{\rm IPW}_{\mathcal{A}\mathcal{A}}\big)^{-1}\|_{\infty}  \|(\bI + \alpha^{-1}(1-\alpha)\mu \big(\widehat{\bSigma}^{\rm IPW}_{\mathcal{A}\mathcal{A}}\big)^{-1})^{-1}\|_{\infty} \\
	&\le \|\widehat{\bSigma}^{\rm IPW}_{\mathcal{A}^c\mathcal{A}}\big(\widehat{\bSigma}^{\rm IPW}_{\mathcal{A}\mathcal{A}}\big)^{-1}\|_{\infty}\big(1 - \alpha^{-1} (1-\alpha)\mu \|\big(\widehat{\bSigma}^{\rm IPW}_{\mathcal{A}\mathcal{A}}\big)^{-1}\|_{\infty}\big)^{-1}
	\end{align*}
	where the last inequality depends on that for any operator norm $\|\cdot \|$ and a matrix $\bU$,
	$$
	\|(\bI+\bU)^{-1}\| \le \frac{1}{1-\|\bU\|}, \quad \text{if } \|\bU\|<1.
	$$
	To use it, we need the following condition
	$$
	\alpha^{-1} (1-\alpha)\mu \|\big(\widehat{\bSigma}^{\rm IPW}_{\mathcal{A}\mathcal{A}}\big)^{-1}\|_{\infty} < 1.
	$$
	
\end{proof}

\begin{proof}[Proof of Lemma \ref{lem:IR_equiv_cond}]
	Putting $\alpha^* = (\mu - \epsilon) / (\mu - \lambda_{\min}(\widehat{\bSigma}^{\rm IPW}))$, we want to show
	\begin{equation}\label{eq:IR_equiv_cond_suff}
	\dfrac{(1-\alpha^*)\mu}{\alpha^*} \|\big(\widehat{\bSigma}^{\rm IPW}_{\mathcal{A}\mathcal{A}}\big)^{-1}\|_{\infty} = 
	\dfrac{\mu}{\mu - \epsilon} (\epsilon - \lambda_{\min}(\widehat{\bSigma}^{\rm IPW})) \|\big(\widehat{\bSigma}^{\rm IPW}_{\mathcal{A}\mathcal{A}}\big)^{-1}\|_{\infty} < 1.		
	\end{equation}
	Remark that by Weyl's inequality
	$$
	\lambda_{\min}(\widehat{\bSigma}^{\rm IPW}) \ge - 
	\left\|\widehat{\bSigma}^{\rm IPW} - \bSigma \right\|_2 + \lambda_{\min}(\bSigma),
	$$
	and recall (\ref{eq:misc_1})
	$$
	\left\|(\widehat{\bSigma}^{\rm IPW}_{\mathcal{A}\mathcal{A}})^{-1}\right\|_{\infty}
	\le \frac{\eta_1}{1-\eta_1\delta_1}.
	$$
	Some basic algebra with these two leads to a sufficient condition of (\ref{eq:IR_equiv_cond_suff}):
	$$
	\left\|\widehat{\bSigma}^{\rm IPW}_{\mathcal{A}\mathcal{A}} - \bSigma_{\mathcal{A}\mathcal{A}}\right\|_{\infty} + 
	\dfrac{\mu 	\left\|\widehat{\bSigma}^{\rm IPW} - \bSigma \right\|_2}{\mu - \epsilon}
	\le 1 / \left\|\bSigma_{\mathcal{A}\mathcal{A}}^{-1}\right\|_{\infty} + \dfrac{\mu(\lambda_{\min}(\bSigma) - \epsilon)}{\mu - \epsilon}.
	$$	
\end{proof}

%
\subsection{Proof of Theorem \ref{thm:IR_LPD_pos}}\label{sec:pf_thm_IR_pos}

\begin{proof}
	If the smallest eigenvalue of the IPW estimator is positive, the LPD estimator of it is the IPW estimator, i.e. $\alpha^*=1$. By following the same proof of Lemma \ref{lem:IR_LPD}, we have
	$$
	\left\|\widehat{\bSigma}^{\rm IPW}_{\mathcal{A}^c\mathcal{A}}\big(\widehat{\bSigma}^{\rm IPW}_{\mathcal{A}\mathcal{A}}\big)^{-1}\right\|_{\infty} \le \dfrac{\eta_1 \delta_2 + \eta_2}{1 - \eta_1 \delta_1}, \quad \text{if } \eta_1 \delta_1 < 1.
	$$
	where we use the same definitions of the matrix norms:
	$$
	\begin{array}{c}
	\eta_1 = \left\|\bSigma_{\mathcal{A}\mathcal{A}}^{-1}\right\|_{\infty}, \quad \eta_2 = \left\|\bSigma_{\mathcal{A}^c\mathcal{A}} \bSigma_{\mathcal{A}\mathcal{A}}^{-1}\right\|_{\infty}\\	
	\delta_1 = \left\|\widehat{\bSigma}^{\rm IPW}_{\mathcal{A}\mathcal{A}} - \bSigma_{\mathcal{A}\mathcal{A}}\right\|_{\infty}, \quad \delta_2 = \left\| \widehat{\bSigma}^{\rm IPW}_{\mathcal{A}^c\mathcal{A}} - \bSigma_{\mathcal{A}^c\mathcal{A}}\right\|_{\infty}.
	\end{array}
	$$
	Using $\eta_2 < 1 - \tau$, it is sufficient for the irrepresentability condition of $\widehat{\bSigma}^{\rm IPW}$ to show 
	$$
	\dfrac{\eta_1 \delta_2 + 1 - \tau}{1 - \eta_1 \delta_1} < 1.
	$$
	The above is equivalent to $\delta_1 + \delta_2 < \tau/\eta_1$.
	
\end{proof}

\subsection{Proof of Theorem \ref{thm:BG}}\label{sec:pf_thm_BG}
\begin{proof}
	Using $y_i = \bx_i^{\top} \bbeta^* + \epsilon_i$ in calculating $\hat{\brho}^{\rm IPW}$, we can obtain 
	$$
	\begin{array}{rcl}
	\nabla\ell(\bbeta^*; \widehat{\bSigma}^{\rm LPD}, \hat{\brho}^{\rm IPW}) &= & 
	\widehat{\bSigma}^{\rm LPD}\bbeta^*  - \hat{\brho}^{\rm IPW} \\
	&=&
	\left(\widehat{\bSigma}^{\rm LPD} - \bV\right)\bbeta^*  - \bw \\
	\end{array}
	$$
	where $\bV \in \mathbb{R}^{p\times p}$ and $\bw\in \mathbb{R}^p$ have its element respectively by
	$$
	\begin{array}{c}
	v_{jk} = n^{-1} \sum\limits_{i=1}^n x_{ij} x_{ik} \delta^x_{ij} \delta^y_i / \pi^{xy}_{j}, \quad 1\le j,k\le p, \\
	w_j = n^{-1} \sum\limits_{i=1}^n x_{ij} \epsilon_i \delta^x_{ij} \delta^y_i/\pi^{xy}_{j}, \quad 1\le j \le p
	\end{array}
	$$
	where $\pi^{xy}_{j}={\rm P}(\delta^y_1= \delta^x_{1j}=1)$. Hence, the norm of the gradient is 
	$$
	\begin{array}{rcl}
	\|\nabla\ell(\bbeta^*; \widehat{\bSigma}^{\rm LPD}, \hat{\brho}^{\rm IPW})\|_{\infty} & \le & 
	\left\|\left(\widehat{\bSigma}^{\rm LPD} - \bV\right)\bbeta^*\right\|_{\infty}  +
	\|\bw\|_{\infty}\\
	& = &
	\max\limits_{1 \le j \le p} \sum\limits_{k\in \mathcal{A}} \left| \left(\widehat{\bSigma}^{\rm LPD} - \bV\right)_{jk} \right| |\beta^*_k| +
	\|\bw\|_{\infty}\\
	& \le &
	\|\widehat{\bSigma}^{\rm LPD} - \bV \|_{\infty, \mathcal{A}}  \; \beta^*_{\max} +
	\|\bw\|_{\infty}
	\end{array}
	$$
	where the first inequality is from the triangular inequality, the next equality holds because $\beta^*_k=0$ for $k\in \mathcal{A}^{\mathsf{c}}$, and the last inequality is obvious from definitions $\beta^*_{\max} = \max\limits_{1 \le j \le p} |\beta^*_j|$ and $\|\bB\|_{\infty, \mathcal{A}} = \max\limits_{1 \le j \le p} \sum\limits_{k\in \mathcal{A}} |b_{jk}|$ for any matrix $\bB=(b_{jk})_{p\times p}$. Note that $\|\bB\|_{\infty, \mathcal{A}}$ is a semi-norm on $\mathbb{R}^{p\times p}$ given a non-empty set $\mathcal{A}$ (i.e. $\|\bB\|_{\infty, \mathcal{A}}=0$ does not imply $\bB=0$). Finally, using 
	$\widehat{\bSigma}^{\rm LPD} - \bV = 
	\alpha^* (\widehat{\bSigma}^{\rm IPW} - \bSigma) + (1-\alpha^*)(\mu \bI - \bSigma) - (\bV - \bSigma)$ and 
	the triangular inequality, we get
	\begin{equation}\label{eq:BG_four_terms}
	\begin{array}{l}
	\|\nabla\ell(\bbeta^*; \widehat{\bSigma}^{\rm LPD}, \hat{\brho}^{\rm IPW})\|_{\infty} \le 
	\bigg( \|\widehat{\bSigma}^{\rm IPW} - \bSigma \|_{\infty, \mathcal{A}} +
	(1-\alpha^*)\|\mu\bI - \bSigma \|_{\infty, \mathcal{A}} \\
	\qquad\qquad\qquad\qquad\qquad\qquad\qquad
	+\|\bSigma - \bV \|_{\infty, \mathcal{A}} \bigg) \beta^*_{\max} +
	\|\bw\|_{\infty}.
	\end{array}
	\end{equation}

	We use Lemma 1 of \cite{Park:2022_stat} to the terms above except the second. Let us define a function $f$ by
	$$
	f(n, p, \mathcal{B}) = |\mathcal{B}| \sqrt{\dfrac{2 \log p + \log |\mathcal{B}|}{2n}}, \quad \mathcal{B} \subset [p],
	$$
	$\sigma_{\max}=\max_{jj} \sigma_{jj}$, and probabilities $\pi^{xx}_{\min, \mathcal{A}} = \min\limits_{1\le j\le p, k \in \mathcal{A}}\pi^{xx}_{jk}, \pi^{xx}_{\min} = \min\limits_{1\le j, k\le p}\pi^{xx}_{jk}, \pi^{xy}_{\min} = \min\limits_{1\le j\le p}\pi^{xy}_{j}$. 
	Then, we can easily get the followings: for some numerical constants $c_1, c_2, c_3, C_1, C_2, C_3>0$ such that
	\begin{equation}\label{eq:BG_first_random}
	{\rm P}_{\delta,x}\left(
	\|\widehat{\bSigma}^{\rm IPW} - \bSigma\|_{\infty, \mathcal{A}}  \ge 
	\dfrac{C_1(K^x)^2 \sigma_{\max}}{\sqrt{\pi^{xx}_{\min, \mathcal{A}}}}f(n, p, \mathcal{A})
	\right)
	\le
	2/p,
	\end{equation}
	if $\dfrac{n}{2\log p  + \log |\mathcal{A}|} > \dfrac{1}{c_1 \pi^{xx}_{\min, \mathcal{A}}}$, 
	\begin{equation}\label{eq:BG_third_random}
	{\rm P}_{\delta,x}\left(
	\| \bV - \bSigma\|_{\infty, \mathcal{A}}  \ge 
	\dfrac{C_2 (K^x)^2\sigma_{\max}}{\sqrt{\pi^{xy}_{\min}}}f(n, p, \mathcal{A})
	\right)
	\le
	2/p,
	\end{equation}
	if $\dfrac{n}{2\log p  + \log |\mathcal{A}|} > \dfrac{1}{c_2 \pi^{xy}_{\min}}$, and
	\begin{equation}\label{eq:BG_fourth_random}
	{\rm P}_{\delta,x}\left(
	\| \bw \|_{\infty}  \ge 
	\dfrac{C_3 \sqrt{\sigma_{\max}\sigma_{\epsilon\epsilon}} K^x  K^{\epsilon}}{\sqrt{\pi^{xy}_{\min}}}f(n, p, [1])
	\right)
	\le
	2/p,
	\end{equation}
	if $\dfrac{n}{3\log p} > \dfrac{1}{c_3 \pi^{xy}_{\min}}$.
	Moreover, we get the concentration of the second term: for some $c_4, C_4>0$
	\begin{equation}\label{eq:BG_second_random}
	\begin{array}{l}
	{\rm P}_{\delta,x}\bigg(
	(1-\alpha^*)\|\mu\bI - \bSigma \|_{\infty, \mathcal{A}} \ge 
	C_4 \tr(\bSigma) \max\{(K^x)^2, 1\}\\
	\qquad \qquad \qquad \qquad \qquad \qquad \qquad \times \left( 1 + \dfrac{\|\bSigma \|_{\infty, \mathcal{A}}}{\mu}\right)\sqrt{\pi_{\max}^{(4)}}f(n, p, [1])
	\bigg)
	\le
	3/p,
	\end{array}
	\end{equation}
	if $n > c_4 \pi_{\max}^{(4)} \log^3(p \vee n)$.
	The proof of (\ref{eq:BG_second_random}) is pended until the end of the proof.

	Combining these results, it holds with probability greater than $1-9/p$
	$$
	\|\nabla\ell(\bbeta^*; \widehat{\bSigma}^{\rm LPD}, \hat{\brho}^{\rm IPW})\|_{\infty} \le 
	L \cdot f(n, p, \mathcal{A}),
	$$
	if $n > c \max\Big\{
	\log p / \pi^{xy}_{\min}, \pi^{(4)}_{\max} \log^3(p \vee n)
	\Big\}$ 
	for some numerical constant $c>0$. The factor $L>0$ is a function of parameters given by
	$$
	\begin{array}{l}
	L \propto  \beta_{\max}^*  \max\{(K^x)^2, 1\}  \sqrt{\pi_{\max}^{(4)}}\tr(\bSigma) \left( 1 + \dfrac{\|\bSigma \|_{\infty, \mathcal{A}}}{\mu}\right) \\
	\qquad \qquad \qquad \qquad  + 
	\dfrac{\max\Big\{\sqrt{\sigma_{\max} \sigma_{\epsilon\epsilon}} K^x  K^{\epsilon}, \sigma_{\max} (K^x)^2\Big\}}{\sqrt{\pi^{xy}_{\min}}}.
	\end{array}
	$$
	To derive the constant $L$, we used $\pi_{\max}^{(4)} \ge 1/\pi^{xx}_{\min, \mathcal{A}}$. Note that if $\lambda_{\min}(\widehat{\bSigma}^{\rm IPW}) > 0$, the second term in (\ref{eq:BG_four_terms}) no longer exists since $\alpha^*=0$. Then, we only need to combine (\ref{eq:BG_first_random}), (\ref{eq:BG_third_random}), (\ref{eq:BG_fourth_random}), which leads to another $L'>0$ smaller than $L$. The constant given in the statement of the theorem is deriven considering it.
	%
	
	Now, we prove (\ref{eq:BG_second_random}), which depends on the following lemma.
	\begin{lem}\label{lem:alpha_bound}
		Assume $\epsilon$ is smaller than the smallest eigenvalue of $\bSigma$.
		For $\alpha^* = 
		{\rm I}(\lambda_{\min}(\widehat{\bSigma}^{\rm IPW}) > 0) + 
		(\mu - \epsilon) / (\mu - \lambda_{\min}(\widehat{\bSigma}^{\rm IPW})) {\rm I}(\lambda_{\min}(\widehat{\bSigma}^{\rm IPW}) \le 0)$, we have
		$$
		1-\alpha^* \le \|\widehat{\bSigma}^{\rm IPW} - \bSigma \|_2 / \mu
		$$
	\end{lem}
	\noindent
	\begin{proof}
		By definition of $\alpha^*$, we have
		$$
		1-\alpha^*
		= (\epsilon - \lambda_{\min}(\widehat{\bSigma}^{\rm IPW})) / (\mu - \lambda_{\min}(\widehat{\bSigma}^{\rm IPW})) {\rm I}(\lambda_{\min}(\widehat{\bSigma}^{\rm IPW}) \le 0).
		$$
		Now, we observe
		$$
		\begin{array}{rcl}
		\dfrac{\epsilon - \lambda_{\min}(\widehat{\bSigma}^{\rm IPW})}{\mu - \lambda_{\min}(\widehat{\bSigma}^{\rm IPW})}{\rm I}(\lambda_{\min}(\widehat{\bSigma}^{\rm IPW}) \le 0) &\le& \dfrac{(\epsilon - \lambda_{\min}(\widehat{\bSigma}^{\rm IPW}))_+}{\mu} \\
		&\le& \dfrac{(\lambda_{\min}(\bSigma) - \lambda_{\min}(\widehat{\bSigma}^{\rm IPW}))_+}{\mu} \\
		&\le&\dfrac{\|\widehat{\bSigma}^{\rm IPW} - \bSigma\|_2}{\mu}
		\end{array}
		$$
		where we use Weyl's inequality in the last inequality.
	\end{proof}
	\noindent
	By applying Lemma \ref{lem:alpha_bound}, we get
	\begin{equation}\label{eq:BG_third}
	(1-\alpha^*)\|\mu\bI - \bSigma \|_{\infty, \mathcal{A}} \le \|\widehat{\bSigma}^{\rm IPW} - \bSigma\|_2\dfrac{\|\mu\bI - \bSigma \|_{\infty, \mathcal{A}}}{\mu} \le 
	\|\widehat{\bSigma}^{\rm IPW} - \bSigma\|_2\left( 1 + \dfrac{\|\bSigma \|_{\infty, \mathcal{A}}}{\mu}\right),
	\end{equation}
	From Theorem \ref{thm:IPW_spectral}, if the sample size condition $n > \pi_{\max}^{(4)} (\alpha + 1)^3 \log^3(p \vee n)$ is satisfied, it holds with probability at least $1-3/p^{\alpha}$ that 
	\begin{equation}
	||\widehat{\bSigma}^{\rm IPW} - \bSigma ||_2 \le 
	C \tr(\bSigma) \max\{(K^x)^2, 1\}
	\sqrt{\frac{\pi_{\max}^{(4)}(\alpha + 1) \log p}{n}},
	\end{equation}
	where $C>0$ is some numerical constant. This concludes that if $n > 16\pi_{\max}^{(4)} \log^3(p \vee n)$
	$$
	\begin{array}{l}
	{\rm P}_{\delta,x}\bigg(
	(1-\alpha^*)\|\mu\bI - \bSigma \|_{\infty, \mathcal{A}} \ge 
	C \tr(\bSigma) \max\{(K^x)^2, 1\}\\
	\qquad\qquad\qquad\qquad\qquad\qquad\qquad \times \left( 1 + \dfrac{\|\bSigma \|_{\infty, \mathcal{A}}}{\mu}\right)
	\sqrt{\frac{2\pi_{\max}^{(4)} \log p}{n}}
	\;\bigg)
	\le
	3/p.
	\end{array}
	$$
\end{proof}

\section{Additional details/results of simulation study}

\subsection{The corrected cross-validation}\label{app:corr_cv}
For the cross-validation, we split data into $K$ folds. Let $\widehat{\bbeta}_k(\lambda)$ be the solution of any penalized regression estimated with tuning parameter at $\lambda$ and with all samples but in the $k$-th fold. Given a set $\Lambda$ of candidates, we aim to find the best one that minimizes the prediction error on the $k$-th fold:
$$
\hat{\lambda}_{opt} = \argmin_{\lambda \in \Lambda} 
\sum_{k=1}^K (\widehat{\bbeta}_k(\lambda))^\top (\widehat{\bSigma}^{\rm IPW}_k)_+ \widehat{\bbeta}_k(\lambda) - 2 \hat{\brho}_k \widehat{\bbeta}_k(\lambda).
$$
Here, we define
$$
(\widehat{\bSigma}^{\rm IPW}_k)_+ = \begin{cases}
\mu \alpha \widehat{\bSigma}^{\rm IPW}_k + (1-\alpha) \bI, & \text{for cases of LPD, NCL}\\
\min\limits_{\bSigma \succeq 0} \left\| \widehat{\bSigma}^{\rm IPW}_k - \bSigma \right\|_{\max}, & \text{for cases of CoCo},
\end{cases}
$$
and $\widehat{\bSigma}^{\rm IPW}_k$ is the IPW estimate calculated over samples in the $k$-th fold, and $\hat{\brho}_k$ is similarly defined.

\subsection{Method comparison}\label{sec:supp_method}

We focus on comparing a list of variants of LPD. For spectral norm and $\ell_\infty$-norm, any value over some lower bound, say $\mu_{lb}$, will do, so we suggest trying $k \cdot \mu_{lwr}$, $k=1,3, 5$, to see how much their performances are different. Considering these variants, we name our proposals by LPD-\textit{norm}-\textit{k} where $\text{\textit{norm}} \in \{\text{S}, \text{F}, \text{I}, \text{E}\}$ and $k \in \{1,3,5\}$, resulting 8 estimators (LPD-S-1, LPD-S-3, LPD-S-5, 
LPD-F-1,
LPD-I-1, LPD-I-3, LPD-I-5,
LPD-E-1).

\begin{table}[H]
	\tiny
	\centering
	\begin{tabular}{|c|c|c|c|c|c|c|}
		\hline & \multicolumn{6}{|c|}{$ p=200,s=0.05 $}\\
		\cline{2-7}
		& PE & MSE & pAUC & F$_1$ & TP & FP \\
		\hline TL & 1.915 (0.609)& 3.656 (1.145)& 0.953 (0.031)& 0.439 (0.071)& 9.680 (0.513)& 25.560 (7.484) \\ 
		NL & 3.694 (1.034)& 6.160 (1.638)& 0.879 (0.063)& 0.396 (0.069)& 8.620 (1.086)& 25.720 (7.420) \\ 
		CoCo & 3.385 (0.927)& 6.441 (1.772)& 0.830 (0.065)& 0.400 (0.076)& 8.440 (1.163)& 24.460 (6.102) \\ 
		NCL & 5.158 (1.222)& 6.292 (1.601)& 0.508 (0.075)& 0.453 (0.093)& 8.140 (1.309)& 19.060 (10.442) \\ 
		LPD-E-1 & 3.290 (0.840)& 6.308 (1.659)& 0.879 (0.054)& 0.369 (0.070)& 8.780 (0.996)& 29.840 (7.313) \\ 
		LPD-F-1 & 3.608 (0.927)& 6.534 (1.708)& 0.881 (0.053)& 0.350 (0.063)& 8.880 (0.982)& 32.920 (7.948) \\ 
		LPD-L-1 & 3.311 (0.867)& 6.262 (1.640)& 0.879 (0.053)& 0.370 (0.066)& 8.800 (1.050)& 29.640 (7.551) \\ 
		LPD-L-3 & 3.242 (0.844)& 6.131 (1.548)& 0.878 (0.056)& 0.377 (0.062)& 8.780 (1.036)& 28.320 (5.223) \\ 
		LPD-L-5 & 3.260 (0.806)& 6.182 (1.515)& 0.880 (0.054)& 0.376 (0.066)& 8.820 (1.004)& 28.780 (6.075) \\ 
		LPD-S-1 & 3.256 (0.828)& 6.181 (1.572)& 0.879 (0.055)& 0.376 (0.067)& 8.780 (0.996)& 28.680 (6.149) \\ 
		LPD-S-3 & 3.251 (0.817)& 6.165 (1.530)& 0.878 (0.054)& 0.376 (0.064)& 8.800 (1.050)& 28.680 (5.527) \\ 
		LPD-S-5 & 3.300 (0.839)& 6.282 (1.578)& 0.878 (0.055)& 0.363 (0.067)& 8.780 (0.996)& 30.560 (7.654) \\ 
		\hline & \multicolumn{6}{|c|}{$ p=500,s=0.05 $}\\
		\cline{2-7}
		& PE & MSE & pAUC & F$_1$ & TP & FP \\
		\hline TL & 6.039 (1.193)& 11.825 (2.347)& 0.809 (0.048)& 0.420 (0.050)& 22.980 (1.286)& 62.980 (16.109) \\ 
		NL & 17.374 (4.272)& 27.698 (3.981)& 0.535 (0.081)& 0.278 (0.055)& 12.240 (2.966)& 50.440 (9.311) \\ 
		CoCo & 16.370 (2.833)& 31.179 (4.848)& 0.596 (0.046)& 0.276 (0.051)& 11.880 (2.847)& 49.060 (9.421) \\ 
		NCL & 28.492 (7.734)& 27.538 (3.863)& 0.504 (0.061)& 0.212 (0.055)& 14.560 (5.035)& 106.460 (55.869) \\ 
		LPD-E-1 & 18.634 (3.463)& 29.315 (4.630)& 0.703 (0.057)& 0.247 (0.044)& 14.760 (2.959)& 80.900 (19.125) \\ 
		LPD-F-1 & 26.511 (6.173)& 31.870 (5.696)& 0.702 (0.054)& 0.238 (0.045)& 14.920 (2.687)& 88.020 (25.206) \\ 
		LPD-L-1 & 14.017 (2.209)& 26.636 (3.549)& 0.703 (0.056)& 0.250 (0.045)& 14.580 (2.829)& 78.020 (17.977) \\ 
		LPD-L-3 & 14.030 (2.391)& 26.661 (4.044)& 0.704 (0.054)& 0.251 (0.044)& 14.560 (2.865)& 77.400 (17.331) \\ 
		LPD-L-5 & 13.869 (2.186)& 26.393 (3.570)& 0.704 (0.055)& 0.252 (0.043)& 14.540 (2.887)& 76.380 (14.380) \\ 
		LPD-S-1 & 13.923 (2.078)& 26.499 (3.362)& 0.704 (0.055)& 0.251 (0.042)& 14.440 (2.786)& 76.700 (17.765) \\ 
		LPD-S-3 & 13.853 (2.097)& 26.377 (3.434)& 0.703 (0.053)& 0.253 (0.043)& 14.520 (2.880)& 75.660 (15.904) \\ 
		LPD-S-5 & 14.129 (2.182)& 26.761 (3.763)& 0.703 (0.055)& 0.251 (0.047)& 14.600 (2.871)& 78.200 (21.832) \\ 
		\hline
	\end{tabular}
	\caption{Method comparison for $p=200, 500$ and $s=0.05, 0.1$. Each performance measures are averaged over $R=100$ repetitions (standard deviation in parenthesis).}
\end{table}

\begin{table}[H]
	\tiny
	\centering
	\begin{tabular}{|c|c|c|c|c|c|c|}
		\hline & \multicolumn{6}{|c|}{$ p=200,s=0.1 $}\\
		\cline{2-7}
		& PE & MSE & pAUC & F$_1$ & TP & FP \\
		\hline TL & 3.220 (0.763)& 6.251 (1.483)& 0.916 (0.034)& 0.532 (0.066)& 19.600 (0.606)& 35.220 (9.790) \\ 
		NL & 11.020 (3.241)& 15.799 (3.181)& 0.755 (0.061)& 0.434 (0.059)& 14.240 (2.273)& 31.440 (5.444) \\ 
		CoCo & 9.878 (2.507)& 17.890 (4.268)& 0.715 (0.053)& 0.431 (0.068)& 13.640 (2.145)& 29.980 (7.150) \\ 
		NCL & 17.212 (3.866)& 17.602 (2.613)& 0.614 (0.045)& 0.386 (0.100)& 14.280 (2.241)& 46.520 (27.309) \\ 
		LPD-E-1 & 9.085 (1.956)& 17.196 (3.661)& 0.765 (0.054)& 0.406 (0.056)& 14.880 (2.086)& 38.960 (9.167) \\ 
		LPD-F-1 & 10.020 (2.320)& 17.907 (3.941)& 0.765 (0.054)& 0.394 (0.054)& 14.900 (2.082)& 41.260 (8.689) \\ 
		LPD-L-1 & 8.914 (2.040)& 16.123 (3.352)& 0.764 (0.054)& 0.414 (0.056)& 14.700 (2.053)& 36.660 (7.176) \\ 
		LPD-L-3 & 8.868 (1.969)& 16.161 (3.436)& 0.768 (0.054)& 0.415 (0.055)& 14.780 (2.122)& 36.660 (6.394) \\ 
		LPD-L-5 & 8.916 (2.131)& 16.137 (3.395)& 0.765 (0.055)& 0.414 (0.056)& 14.780 (2.141)& 36.800 (6.958) \\ 
		LPD-S-1 & 8.819 (2.044)& 16.157 (3.432)& 0.765 (0.055)& 0.413 (0.052)& 14.740 (2.058)& 36.780 (6.538) \\ 
		LPD-S-3 & 8.840 (2.057)& 16.113 (3.424)& 0.764 (0.053)& 0.414 (0.056)& 14.700 (2.112)& 36.500 (6.519) \\ 
		LPD-S-5 & 9.045 (2.218)& 16.381 (3.655)& 0.764 (0.056)& 0.411 (0.059)& 14.760 (2.036)& 37.660 (8.277) \\ 
		\hline & \multicolumn{6}{|c|}{$ p=500,s=0.1 $}\\
		\cline{2-7}
		& PE & MSE & pAUC & F$_1$ & TP & FP \\
		\hline TL & 14.102 (2.010)& 27.752 (4.021)& 0.684 (0.045)& 0.474 (0.048)& 43.740 (2.284)& 92.480 (21.073) \\ 
		NL & 48.511 (11.754)& 75.830 (9.527)& 0.392 (0.062)& 0.272 (0.056)& 16.840 (3.966)& 56.320 (7.377) \\ 
		CoCo & 47.069 (8.296)& 90.279 (15.734)& 0.547 (0.032)& 0.254 (0.048)& 15.180 (3.336)& 53.820 (8.075) \\ 
		NCL & 76.743 (26.682)& 64.362 (9.807)& 0.492 (0.038)& 0.245 (0.038)& 25.380 (7.545)& 130.100 (42.421) \\ 
		LPD-E-1 & 59.310 (12.606)& 81.429 (11.177)& 0.606 (0.045)& 0.260 (0.047)& 20.820 (4.341)& 89.180 (17.235) \\ 
		LPD-F-1 & 93.961 (23.197)& 91.393 (14.167)& 0.606 (0.044)& 0.252 (0.044)& 21.160 (4.560)& 96.360 (18.729) \\ 
		LPD-L-1 & 37.572 (5.268)& 72.016 (9.589)& 0.601 (0.044)& 0.261 (0.044)& 20.900 (4.273)& 89.580 (15.831) \\ 
		LPD-L-3 & 37.343 (5.633)& 71.308 (10.009)& 0.606 (0.043)& 0.263 (0.047)& 20.620 (4.125)& 86.680 (17.115) \\ 
		LPD-L-5 & 37.214 (5.183)& 71.073 (9.155)& 0.606 (0.044)& 0.263 (0.047)& 20.800 (4.536)& 87.240 (14.981) \\ 
		LPD-S-1 & 37.091 (4.728)& 70.722 (8.250)& 0.603 (0.042)& 0.264 (0.046)& 20.600 (4.267)& 85.180 (16.184) \\ 
		LPD-S-3 & 36.894 (4.797)& 70.567 (8.786)& 0.604 (0.045)& 0.264 (0.049)& 20.600 (4.290)& 85.440 (14.098) \\ 
		LPD-S-5 & 36.937 (5.200)& 70.630 (9.674)& 0.605 (0.046)& 0.264 (0.048)& 20.420 (4.121)& 84.700 (15.538) \\ 
		\hline
	\end{tabular}
	\caption{Method comparison for $p=200, 500$ and $s=0.05, 0.1$. Each performance measures are averaged over $R=100$ repetitions (standard deviation in parenthesis).}
\end{table}
Among four matrix norms considered here, $\ell_\infty$-norm (LPD-L) and spectral norm (LPD-S) perform best, while different $\mu$ values do not result in any significant changes in practice. The other two norms do not achieve comparative results when the dimension increases to $p=500$.

\subsection{Missng mechanism}\label{sec:supp_missing}

Also, we fix the multiplicative factor $k=1$ for all matrix norms in LPD.



\begin{table}[H]
	\tiny
	\centering
	\begin{tabular}{|c|c|c|c|c|c|c|}
		\hline & \multicolumn{6}{|c|}{$ \theta=0.9,\text{ MAR} $}\\
		\cline{2-7}
		& PE & MSE & pAUC & F$_1$ & TP & FP \\
		\hline TL & 1.860 (0.536)& 3.558 (1.059)& 0.948 (0.039)& 0.455 (0.063)& 9.700 (0.544)& 23.640 (5.784) \\ 
		NL & 3.654 (1.052)& 5.989 (1.528)& 0.866 (0.067)& 0.389 (0.076)& 8.500 (1.074)& 26.220 (7.731) \\ 
		CoCo & 3.229 (0.861)& 6.179 (1.627)& 0.832 (0.064)& 0.387 (0.084)& 8.340 (1.171)& 25.980 (8.482) \\ 
		NCL & 4.823 (1.126)& 6.149 (1.613)& 0.548 (0.091)& 0.428 (0.113)& 8.080 (1.275)& 23.260 (17.444) \\ 
		LPD-E-1 & 3.316 (0.907)& 6.227 (1.672)& 0.879 (0.058)& 0.346 (0.071)& 8.680 (0.935)& 32.940 (9.182) \\ 
		LPD-F-1 & 3.451 (0.937)& 6.240 (1.652)& 0.877 (0.059)& 0.343 (0.065)& 8.740 (0.944)& 33.660 (9.164) \\ 
		LPD-L-1 & 3.147 (0.836)& 5.934 (1.482)& 0.876 (0.060)& 0.371 (0.065)& 8.520 (1.054)& 28.240 (6.962) \\ 
		LPD-S-1 & 3.094 (0.815)& 5.893 (1.484)& 0.877 (0.060)& 0.366 (0.065)& 8.500 (1.015)& 28.760 (6.133) \\ 
		\hline & \multicolumn{6}{|c|}{$ \theta=0.7,\text{ MAR} $}\\
		\cline{2-7}
		& PE & MSE & pAUC & F$_1$ & TP & FP \\
		\hline TL & 1.828 (0.490)& 3.512 (0.991)& 0.956 (0.037)& 0.438 (0.076)& 9.740 (0.600)& 26.040 (7.982) \\ 
		NL & 9.796 (2.676)& 8.887 (1.463)& 0.718 (0.100)& 0.290 (0.073)& 5.600 (1.400)& 24.060 (9.646) \\ 
		CoCo & 6.027 (1.422)& 10.851 (2.433)& 0.666 (0.096)& 0.303 (0.075)& 5.480 (1.344)& 21.080 (5.606) \\ 
		NCL & 6.813 (1.513)& 10.039 (1.974)& 0.466 (0.081)& 0.312 (0.091)& 4.980 (1.363)& 17.500 (5.694) \\ 
		LPD-E-1 & 7.048 (3.141)& 11.014 (3.025)& 0.743 (0.093)& 0.253 (0.060)& 6.400 (1.539)& 34.400 (7.910) \\ 
		LPD-F-1 & 21.120 (34.859)& 14.843 (8.075)& 0.746 (0.096)& 0.235 (0.078)& 6.140 (2.204)& 36.020 (9.079) \\ 
		LPD-L-1 & 5.344 (1.177)& 9.132 (1.592)& 0.744 (0.096)& 0.285 (0.061)& 6.540 (1.216)& 29.960 (5.577) \\ 
		LPD-S-1 & 5.238 (1.050)& 9.163 (1.526)& 0.742 (0.093)& 0.283 (0.060)& 6.520 (1.233)& 30.180 (6.521) \\ 
		\hline & \multicolumn{6}{|c|}{$ \theta=0.9,\text{ MNAR} $}\\
		\cline{2-7}
		& PE & MSE & pAUC & F$_1$ & TP & FP \\
		\hline TL & 1.937 (0.558)& 3.697 (1.087)& 0.951 (0.033)& 0.430 (0.073)& 9.700 (0.463)& 26.700 (8.122) \\ 
		NL & 3.952 (1.097)& 6.682 (1.552)& 0.857 (0.063)& 0.369 (0.077)& 8.080 (1.412)& 26.500 (7.492) \\ 
		CoCo & 3.698 (1.010)& 7.055 (1.988)& 0.817 (0.066)& 0.361 (0.075)& 8.060 (1.219)& 27.820 (8.578) \\ 
		NCL & 5.062 (1.149)& 6.917 (1.581)& 0.584 (0.070)& 0.372 (0.109)& 7.720 (1.325)& 28.600 (19.799) \\ 
		LPD-E-1 & 3.624 (0.817)& 6.807 (1.588)& 0.852 (0.063)& 0.341 (0.065)& 8.200 (1.229)& 30.840 (7.980) \\ 
		LPD-F-1 & 3.679 (0.758)& 6.784 (1.474)& 0.851 (0.064)& 0.336 (0.050)& 8.320 (1.186)& 31.680 (6.485) \\ 
		LPD-L-1 & 3.470 (0.893)& 6.602 (1.685)& 0.850 (0.064)& 0.351 (0.064)& 8.220 (1.217)& 29.360 (7.331) \\ 
		LPD-S-1 & 3.478 (0.786)& 6.586 (1.509)& 0.851 (0.061)& 0.353 (0.066)& 8.220 (1.282)& 29.300 (8.117) \\ 
		\hline & \multicolumn{6}{|c|}{$ \theta=0.7,\text{ MNAR} $}\\
		\cline{2-7}
		& PE & MSE & pAUC & F$_1$ & TP & FP \\
		\hline TL & 1.927 (0.536)& 3.708 (1.036)& 0.945 (0.039)& 0.426 (0.064)& 9.700 (0.505)& 27.000 (8.732) \\ 
		NL & 10.107 (3.407)& 9.440 (1.697)& 0.688 (0.080)& 0.286 (0.089)& 5.280 (1.371)& 22.620 (6.648) \\ 
		CoCo & 6.750 (2.215)& 12.217 (4.246)& 0.660 (0.072)& 0.286 (0.082)& 5.080 (1.226)& 21.100 (5.486) \\ 
		NCL & 7.116 (1.667)& 10.195 (2.007)& 0.472 (0.073)& 0.306 (0.093)& 4.820 (1.466)& 17.400 (7.741) \\ 
		LPD-E-1 & 6.930 (2.367)& 10.865 (2.421)& 0.759 (0.082)& 0.251 (0.064)& 6.320 (1.362)& 35.020 (7.878) \\ 
		LPD-F-1 & 10.617 (5.046)& 13.477 (4.554)& 0.759 (0.084)& 0.234 (0.067)& 6.500 (1.821)& 39.740 (11.940) \\ 
		LPD-L-1 & 5.384 (1.176)& 9.481 (1.686)& 0.756 (0.083)& 0.255 (0.063)& 6.320 (1.504)& 33.760 (7.224) \\ 
		LPD-S-1 & 5.351 (1.223)& 9.491 (1.843)& 0.760 (0.082)& 0.260 (0.066)& 6.300 (1.432)& 32.740 (6.452) \\ 
		\hline
	\end{tabular}
	
	\caption{Sensitivity analysis for $\theta=0.7, 0.9$ and different missing mechanisms. Each performance measures are averaged over $R=100$ repetitions (standard deviation in parenthesis).}
	
\end{table}
\end{document}